%% file: dmp.tex
\documentclass[conference]{IEEEtran}

\usepackage[prologue,dvipsnames,table]{xcolor}

\usepackage{amsmath} 
\usepackage{mathtools} 
\usepackage{soul}      
\usepackage{mathpartir}
\usepackage{multirow} 
\usepackage{stmaryrd} 
\usepackage{amssymb}  
\usepackage{amsthm}   
\usepackage{thmtools}
\usepackage{thm-restate}
\usepackage{halloweenmath}      
\usepackage{tabularx}
\usepackage[retain-explicit-plus=true]{siunitx} 
\usepackage{wasysym}  

\usepackage{booktabs}

\usepackage{todonotes}
\usepackage[inline]{enumitem}   
\usepackage{cite}
\usepackage{hyperref}
\hypersetup{
  colorlinks,
  linkcolor={link_color},
  citecolor={link_color},
  urlcolor={link_color}
}

\usetikzlibrary{decorations.pathmorphing}
\input{semantics/macros.sty}
\input{semantics/rules/step_rule.tex}
\input{semantics/rules/updated_mem_rules.tex}

\input{semantics/rules/fetch_rules.tex}
\input{semantics/rules/execute_jmp_rules.tex}
\input{semantics/rules/execute_mem_rules.sty}
\input{semantics/rules/retire_rules.tex}
\input{semantics/prefetchers.tex}
\input{textmacros.sty}

\usepackage{acronym}

\newcommand{\riscv}{RISC\nobreakdash-V\xspace}
\newcommand{\artifactrepo}{\url{https://github.com/JensTheGrens/secdmp}}
\newcommand{\zenodoartifact}{\url{https://doi.org/10.5281/zenodo.23691560}}

\acrodef{BTB}{branch target buffer}
\acrodef{CASCO}{contract-aware secure compilation}
\acrodef{CI}{continuous integration}
\acrodef{CDB}{common data bus}
\acrodef{CSR}{control and status register}
\acrodef{DMP}{data memory-dependent prefetcher}
\acrodef{FPGA}{field-programmable gate array}
\acrodef{HDL}{hardware description language}
\acrodef{ISA}{instruction set architecture}
\acrodef{ISR}{interrupt service routine}
\acrodef{IVT}{interrupt vector table}
\acrodef{MMU}{memory management unit}
\acrodef{MPU}{memory protection unit}
\acrodef{NMI}{non-maskable interrupt}
\acrodef{PHT}{pattern history table}
\acrodef{PMP}{physical memory protection}
\acrodef{PSF}{predictive store forwarding}
\acrodef{ROB}{reorder buffer}
\acrodef{ROP}{return-oriented programming}
\acrodef{RSB}{return stack buffer}
\acrodef{SSB}{speculative store bypass}
\acrodef{STL}{store-to-load forwarding}
\acrodef{TCB}{trusted computing base}
\acrodef{TEE}{trusted execution environment}
\acrodef{TLB}{translation lookaside buffer}
\acrodef{VCD}{value change dump}

\def\extendedVersion{1}

\def\hideComments{0}

\newcommand{\extendedElse}[2]{\if\extendedVersion1#1\else#2\fi}

\if\hideComments1
    \newcommand{\marton}[1]{}
    \newcommand{\lesly}[1]{}
    \newcommand{\jens}[1]{}
\else
    \newcommand{\marton}[1]{\todo[color=magenta!40,inline]{marton: #1}}
    \newcommand{\lesly}[1]{\todo[color=green!40,inline]{Lesly: #1}}
    \newcommand{\jens}[1]{\todo[color=yellow!40,inline]{Jens: #1}}
\fi

\usepackage[available,functional,reproduced]{ndssbadges}

\begin{document}
%

\if\extendedVersion1
\title{\secdmp{}: Provably Secure Data Memory-Dependent Prefetching\\(Extended Version)}
\else
\title{\secdmp{}: Provably Secure Data Memory-Dependent Prefetching}
\fi

\author{
    \IEEEauthorblockN{Jens Sprengers}
    \IEEEauthorblockA{DistriNet, KU Leuven\\
        jens.sprengers@kuleuven.be}
    \and
    \IEEEauthorblockN{Marton Bognar}
    \IEEEauthorblockA{DistriNet, KU Leuven\\
        marton.bognar@kuleuven.be}
    \and
    \IEEEauthorblockN{Lesly-Ann Daniel}
    \IEEEauthorblockA{EURECOM\\
        lesly-ann.daniel@eurecom.fr}
    \and
    \IEEEauthorblockN{Frank Piessens}
    \IEEEauthorblockA{DistriNet, KU Leuven\\
        frank.piessens@kuleuven.be}
}


\if\extendedVersion0

\IEEEoverridecommandlockouts
\makeatletter\def\@IEEEpubidpullup{6.5\baselineskip}\makeatother
\IEEEpubid{\parbox{\columnwidth}{
		Network and Distributed System Security (NDSS) Symposium 2027\\
		22--26 March 2027, Seoul, Republic of Korea\\
		ISBN 978-1-970672-09-1\\
		https://dx.doi.org/10.14722/ndss.2027.230173\\
		www.ndss-symposium.org
}
\hspace{\columnsep}\makebox[\columnwidth]{}}

\fi

\maketitle

\begin{abstract}
    \input{sections/00-abstract.tex}
\end{abstract}


%
\IEEEpeerreviewmaketitle

\input{sections/10-introduction.tex}
\input{sections/20-background.tex}
\input{sections/25-threat-model.tex}
\input{sections/30-leakage-classification.tex}
\input{sections/35-overview.tex}
\input{sections/40-formalization.tex}
\input{sections/50-secdmp.tex}
\input{sections/60-implementation-evaluation.tex}

\input{sections/70-discussion-related-future.tex}
\input{sections/80-conclusion.tex}


\section*{Acknowledgment}
This research was partially funded by the Internal Funds KU Leuven, the Cybersecurity Research Program Flanders, and the Research Foundation -- Flanders (FWO) under grant number G081322N.
This work has benefited from a government grant managed by the National Research Agency under France 2030 with reference ANR-22-PECY-0009.

\section*{Ethics Considerations}
This research analyzes existing vulnerabilities in processors but does not demonstrate any attacks on production systems, only on the Proteus research core.
We propose a defense that may help secure future products, but do not expose unknown vulnerabilities in current CPUs.
However, as our research presents the first concrete demonstration that selective disabling does not prevent all DMP leakage, we notified Intel about an ambiguity in their documentation regarding selectively setting the DOIT bit on January 27, 2026.



\bibliographystyle{IEEEtran}
\bibliography{references,bibliography/generated}

\appendices
\crefalias{section}{appendix}
\crefalias{subsection}{appendix}

\input{sections/100-appendix.tex}

\end{document}

%% file: semantics/rules/step_rule.tex
\newcommand\stepRule[0]{
  \infer[step] {
    \directive \mydef \nextdir(\micro) \\
    \smicro{\directive}{} {\conf{\mem, \reg, \buf, \micro}} {\conf{\mem',
        \reg', \buf', \micro'}}
    {} \\
    \micro'' \mydef \update(\micro, \micro')
  }{
    \smicro{}{}
    {\conf{\mem, \reg, \buf, \micro}}
    {\conf{\mem', \reg', \buf', \micro''}}
    {}
  }
}

%% file: semantics/rules/updated_mem_rules.tex
\newcommand\abstractDMPactivate[0]{
  \inferrule[dmp-activate]
  {
    \mconf{} \mydef \conf{\mem, \reg, \buf, \micro} \quad
    \loc \mydef \getpc(\mconf{}, \directive) \\
    \dmpactivate(\micro{}, \loc, \directive) \quad
    \val{a}^{\bot} \mydef \loadaddr(\mconf{}, \directive) \\
    \smicro
    {\directive}{\absbase{}}
    {\conf{\mem, \reg, \buf, \micro}}
    {\conf{ \mem', \reg', \buf',\micro'}}
    {}\\
    \micro_{\text{train}} \mydef \train(\micro, \loc, \mem, \val{a}^{\bot}) \quad
    \micro_{\text{pref}} \mydef \prefetch(\micro, \loc, \mem, \val{a}^{\bot}) \\
    \micro'' \mydef \update(\micro, \micro', \mathsf{activate}, \micro_{\text{train}}, \micro_{\text{pref}}, \val{a}^{\bot}) \\
  }
  {
    \smicro
    {\directive{}}{\dmpbase{}}
    {\conf{\mem, \reg, \buf, \micro}}
    {\conf{\mem', \reg', \buf', \micro''}}
    {}
  }
}

\newcommand\abstractDMPnoactivate[0]{
  \inferrule[dmp-no-activate]
  {
    \loc \mydef \getpc(\mconf{}, \directive) \\
    \neg \dmpactivate(\micro{}, \loc, \directive) \\
    \smicro
    {\directive}{\absbase{}}
    {\conf{\mem, \reg, \buf, \micro}}
    {\conf{\mem', \reg', \buf', \micro'}}
    {}\\
    \micro'' \mydef \update(\micro, \micro', \mathsf{no\mbox{-}activate}) \\
  }
  {
    \smicro
    {\directive{}}{\dmpbase{}}
    {\conf{\mem, \reg, \buf, \micro}}
    {\conf{\mem', \reg', \buf', \micro''}}
    {}
  }
}

\newcommand\updatedLoad[0]{
  \inferrule[execute-load-dmp]
  {
    |\buf| = i-1 \\
    inst = \load{x}{e} @ T \\
    \changedrule{\loc \mydef \semexpr{\pc}{\apl(\buf, \reg)}} \\
    \val{a} \mydef \semexpr{e}{\apl(\buf, \reg)} \\
    \val{v} \mydef \mem(\val{a}) \\
    inst' \mydef \mov{x}{\val{v}} @ T\\
    \store{\_}{\_} \not\in \buf\\
    \changedrule{\micro_\text{train} \mydef \train(\micro, \loc, \mem, \val{a})} \quad
    \changedrule{\micro_\text{pref} \mydef \prefetch(\micro, \loc, \mem, \val{a})} \\
    \micro' \mydef \update(\micro, \changedrule{\micro_\text{train}}, \changedrule{\micro_\text{pref}}, \leaked{\val{a}}) \\
  }
  {
    \smicro
    {\execute{i}}{\dmp{}}
    {\conf{\mem, \reg, \buf \cdot inst \cdot \buf', \micro}}
    {\conf{ \mem, \reg, \buf \cdot inst' \cdot \buf',\micro'}}
    {}
  }
}

\newcommand\updatedSecLoad[0]{
  \inferrule[execute-load]
  {
    |\buf| = i-1 \\
    inst = \load{x}{e} @ T \\
    \loc \mydef \semexpr{\pc}{\apl(\buf, \reg)} \\
    \val{a} \mydef \semexpr{e}{\apl(\buf, \reg)} \\
    \val{v} \mydef \mem(\val{a}) \\
    inst' \mydef \mov{x}{\val{v}} @ T\\
    \store{\_}{\_} \not\in \buf\\
    \changedrule{\micro_\text{train} \mydef \train(\micro, \loc, \sproj{\mem}, \val{a})} \quad
    \changedrule{\micro_\text{pref} \mydef \prefetch(\micro, \loc, \sproj{\mem}, \val{a})} \\
    \micro' \mydef \update(\micro, \changedrule{\micro_\text{train}}, \changedrule{\micro_\text{pref}}, \leaked{\val{a}}) \\
  }
  {
    \smicro
    {\execute{i}}{\secdmp{}}
    {\conf{\mem, \reg, \buf \cdot inst \cdot \buf', \micro}}
    {\conf{ \mem, \reg, \buf \cdot inst' \cdot \buf',\micro'}}
    {}
  }
}

%% file: semantics/rules/fetch_rules.tex
\newcommand\fetchPredictBranchJump[0]{
  \inferrule[fetch-predict-branch-jmp]
  {
    \loc \mydef \semexpr{\pc}{\apl(\buf, \reg)} \\
    \prog{\loc} \in \{\beqz{e}{\_}, \jmp{e}\} \\
    \loc' \mydef \predict(\micro) \\
    \micro' \mydef \update(\micro, \leaked{\loc}) \\
  }
  {
    \smicro{\fetch}{}
    {\conf{\mem,\reg, \buf, \micro}}
    {\conf{\mem,\reg, \buf \cdot \pc \gets \loc' @ \loc, \micro'}}
    {}
  }
}

\newcommand\fetchOthers[0]{
  \inferrule[fetch-others]
  {
    \loc{} \mydef \semexpr{\pc}{\apl(\buf, \reg)} \\
    \prog{\loc} \not\in \{\beqz{\_}{\_}, \jmp{\_}\} \\
    \loc' \mydef \loc + 1 \\
    \micro' \mydef \update(\micro, \leaked{\loc}) \\
  }
  {
    \smicro
    {\fetch}{}
    {\conf{\mem,\reg, \buf, \micro}} {\conf{\mem,\reg, \buf \cdot
        \prog{\loc}@\varepsilon \cdot \incrpc{\loc'}
          @\varepsilon, \micro'}} {}
  }
}

%% file: semantics/rules/execute_jmp_rules.tex
\newcommand\branchCommit[0]{
  \infer[branch-commit]
  {
    |\buf| = i-1 \\
    inst = \pc \gets \loc @ \loc_0 \\
    \prog{\loc_0} = \beqz{e}{\loc'} \\
    \val{c} \mydef \semexpr{e}{\apl(\buf, \reg)} \\
    \loc'' \mydef \ite{\val{c} = 0}{\loc'}{\loc_0 + 1}\\
    \loc = \loc''\\
    inst' = \pc \gets \loc @ \varepsilon \\
    \micro' \mydef \update(\micro, \leaked{\loc}, \leaked{\loc_0})
  }{
    \smicro
    {\execute{i}}{}
    {\conf{\mem,\reg, \buf \cdot  inst \cdot \buf', \micro}}
    {\conf{\mem, \reg, \buf \cdot inst' \cdot \buf', \micro'}}
    {}
  }
}

\newcommand\branchRollback[0]{
  \infer[branch-rollback]
  {
    |\buf| = i-1 \\
    inst = \pc \gets \loc @ \loc_0 \\
    \prog{\loc_0} = \beqz{e}{\loc'} \\
    \val{c} \mydef \semexpr{e}{\apl(\buf, \reg)} \\
    \loc'' \mydef \ite{\val{c} = 0}{\loc'}{\loc_0 + 1}\\
    \loc \neq \loc''\\
    inst' = \pc \gets \loc'' @ \varepsilon \\
    \micro' \mydef \update(\micro, \leaked{\loc''}, \leaked{\loc_0})
  }{
    \smicro
    {\execute{i}}{}
    {\conf{\mem,\reg, \buf \cdot  inst \cdot \buf', \micro}}
    {\conf{\mem, \reg, \buf \cdot inst', \micro'}}
    {}
  }
}

\newcommand\jmpCommit[0]{
  \infer[jmp-commit]
  {
    |\buf| = i-1 \\
    inst = \pc \gets \loc @ \loc_0 \\
    \prog{\loc_0} = \jmp{e} \\
    \loc' \mydef{} \semexpr{e}{\apl(\buf, \reg)} \\
    \loc' = \loc \\
    inst' = \pc \gets \loc @ \varepsilon \\
    \micro' \mydef \update(\micro, \leaked{\loc'}, \leaked{\loc_0})
  }
  {
    \smicro
    {\execute{i}}{}
    {\conf{\mem,\reg, \buf \cdot  inst \cdot \buf', \micro}}
    {\conf{\mem, \reg, \buf \cdot inst' \cdot \buf', \micro'}}
    {}
  }
}

\newcommand\jmpRollback[0]{
  \infer[jmp-rollback]
  {
    |\buf| = i-1 \\
    inst = \pc \gets \loc @ \loc_0 \\
    \prog{\loc_0} = \jmp{e} \\
    \loc' \mydef{} \semexpr{e}{\apl(\buf, \reg)} \\
    \loc' \neq \loc \\
    inst' = \pc \gets \loc' @ \varepsilon \\
    \micro' \mydef \update(\micro, \leaked{\loc'}, \leaked{\loc_0})
  }
  {
    \smicro
    {\execute{i}}{}
    {\conf{\mem,\reg, \buf \cdot  inst \cdot \buf', \micro}}
    {\conf{\mem, \reg, \buf \cdot inst', \micro'}}
    {}
  }
}

%% file: semantics/rules/retire_rules.tex
\newcommand\retireAssign[0]{
\infer[retire-assign]
  {
    \buf = \mov{r}{v}  @ \varepsilon \cdot \buf' \\
  }
  {
    \smicro
    {\retire}
    {}
    {\conf{\mem,\reg, \buf, \micro}}
    {\conf{\mem, \reg\change{r}{v}, \buf', \micro}}
    {}
  }
}

\newcommand\retireStoreLow[0]{
  \inferrule[retire-store-low] {
    \buf = \store{a}{\val{v}} @ \varepsilon \cdot \buf' \\
    \micro' = \update(\micro, \leaked{\val{a}}, \leaked{\val{v}}) \\
    \memsec(\val{a}) = \botsec \\
  }{ \smicro {\retire}{}
    {\conf{\mem, \reg, \buf, \micro}}
    {\conf{\mem\change{\val{a}}{\val{v}}, \reg, \buf', \micro'}}
    {}
  }
}

\newcommand\retireStoreHigh[0]{
  \inferrule[retire-store-high] {
    \buf = \store{a}{\val{v}} @ \varepsilon \cdot \buf' \\
    \micro' = \update(\micro, \leaked{\val{a}}) \\
    \memsec(\val{a}) = \topsec\\
  }{ \smicro {\retire}{}
    {\conf{\mem, \reg, \buf, \micro}}
    {\conf{\mem\change{\val{a}}{\val{v}}, \reg, \buf', \micro'}}
    {}
  }
}

%% file: semantics/prefetchers.tex
\newcommand\secureAppleDmp{
    \begin{equation*}
        \begin{array}{rl}
            \codeline{1}  & \function{function}~\prefetch(\micro, \loc, \mem, \val{a})\colon                                     \\
            \codeline{2}  & \tab \val{base} := \val{a} - \val{a} \bmod{n}                                                        \\
            \codeline{3}  & \tab \val{end} := \val{base} + n - 1                                                                 \\
            \codeline{4}  & \tab \function{for}~\val{address} := \val{base}~\function{to}~\val{end}~\function{do}                \\
            \codeline{5}  & \tab \tab \changedrule{\function{if}~\memsec(\val{address}) = \botsec~\function{then}}               \\
            \codeline{6}  & \tab \tab \tab \val{target} := \mem(\val{address})                                                   \\
            \codeline{7}  & \tab \tab \tab \function{if}~\leaked{\pointerHeuristic(\val{target}, \val{address})}~\function{then} \\
            \codeline{8}  & \tab \tab \tab \tab \micro := \update(\micro, \leaked{\val{target}}, \leaked{\val{target} + n})      \\ 
            \codeline{9}  & \tab \function{return}~\micro                                                                        \\
            \\
            \codeline{10} & \function{function}~\pointerHeuristic(\val{target}, \val{address})\colon                             \\
            \codeline{11} & \tab \function{return}~\val{target}[55\colon\!32] = \val{address}[55\colon\!32]                      \\
        \end{array}
    \end{equation*}
}

\newcommand\secureIntelDmp{
    \begin{equation*}
        \begin{array}{rl}
            \codeline{1}  & \function{function}~\train(\micro, \loc, \mem, \val{a})\colon                                                                    \\
            \codeline{2}  & \tab \val{previousLoads} := \getLoadHistory(\micro)                                                                              \\
            \codeline{3}  & \tab \function{for}~\angles{\loc', \val{a}'}~\function{in}~\val{previousLoads}~\function{do}                                     \\
            \codeline{4}  & \tab \tab \changedrule{\function{if}~\memsec(\val{a}') = \botsec~\function{then}}                                                \\
            \codeline{5}  & \tab \tab \tab \micro := \update(\micro, \leaked{\loc'}, \leaked{\val{a}'}, \leaked{\val{a}}, \leaked{\mem(\val{a}') = \val{a}}) \\
            \codeline{6}  & \tab \micro := \appendToLoadHistory(\micro, \angles{\leaked{\loc}, \leaked{\val{a}}})                                            \\
            \codeline{7}  & \tab \function{return}~\micro                                                                                                    \\
            \\
            \codeline{8}  & \function{function}~\prefetch(\micro, \loc, \mem, \val{a})\colon                                                                 \\
            \codeline{9}  & \tab \val{stride} := \predictStride(\micro, \loc)                                                                                \\
            \codeline{10} & \tab \function{if}~\leaked{\val{stride} \neq 0}~\function{then}                                                                  \\
            \codeline{11} & \tab \tab \val{address} := \val{a} + \val{stride} \times 16                                                                      \\
            \codeline{12} & \tab \tab \function{if}~\leaked{\val{address}[47\colon\!18] = \val{a}[47\colon\!18]}~\function{then}                             \\
            \codeline{13} & \tab \tab \tab \changedrule{\function{if}~\memsec(\val{address}) = \botsec~\function{then}}                                      \\
            \codeline{14} & \tab \tab \tab \tab \val{target} := \mem(\val{address})                                                                          \\
            \codeline{15} & \tab \tab \tab \tab \function{if}~\leaked{\pointerHeuristic(\val{target})}~\function{then}                                       \\
            \codeline{16} & \tab \tab \tab \tab \tab \micro := \update(\micro, \leaked{\val{target}})                                                        \\
            \codeline{17} & \tab \function{return}~\micro                                                                                                    \\
            \\
            \codeline{18} & \function{function}~\pointerHeuristic(\val{target})\colon                                                                        \\
            \codeline{19} & \tab \function{if}~\val{LAM}~\function{then}                                                                                     \\
            \codeline{20} & \tab \tab \function{return}~\val{target}[63] = 0~\function{and}~\val{target}[47] = 0                                             \\
            \codeline{21} & \tab \function{else}                                                                                                             \\
            \codeline{22} & \tab \tab \function{return}~\val{target}[63\colon\!47] = 0                                                                       \\
        \end{array}
    \end{equation*}
}

\newcommand\secureCdp{
    \begin{equation*}
        \begin{array}{rl}
            \codeline{1}  & \function{function}~\prefetch(\micro, \loc, \mem, \val{a})\colon                                              \\
            \codeline{2}  & \tab \function{return}~\prefetchrec(\micro, \loc, \mem, \val{a}, 0)                                           \\
            \\
            \codeline{3}  & \function{function}~\prefetchrec(\micro, \loc, \mem, \val{a}, \val{depth})\colon                              \\
            \codeline{4}  & \tab \function{if}~\val{depth} > \val{maxRecursionDepth}~\function{then}                                      \\
            \codeline{5}  & \tab \tab \function{return}~\micro                                                                            \\
            \codeline{6}  & \tab \val{base} := \val{a} - \val{a} \bmod{n}                                                                 \\
            \codeline{7}  & \tab \val{end} := \val{base} + n - 1                                                                          \\
            \codeline{8}  & \tab \function{for}~\val{address} := \val{base}~\function{to}~\val{end}~\function{do}                         \\
            \codeline{9}  & \tab \tab \changedrule{\function{if}~\memsec(\val{address}) = \botsec~\function{then}}                        \\
            \codeline{10} & \tab \tab \tab \val{target} := \mem(\val{address})                                                            \\
            \codeline{11} & \tab \tab \tab \function{if}~\leaked{\pointerHeuristic(\val{target}, \val{address})}~\function{then}          \\
            \codeline{12} & \tab \tab \tab \tab \function{for}~\val{line} := \val{0}~\function{to}~\val{nextLinePrefetches}~\function{do} \\
            \codeline{13} & \tab \tab \tab \tab \tab \micro := \update(\micro, \leaked{\val{target} + n \times line})                     \\
            \codeline{14} & \tab \tab \tab \tab \tab  \micro := \prefetchrec(\micro, \loc, \mem,                                          \\
                          & \tab \tab \tab \tab \tab \tab \tab \tab \val{target} + n \times line, \val{depth} + 1)                        \\
            \codeline{15} & \tab \function{return}~\micro                                                                                 \\
            \\
            \codeline{16} & \function{function}~\pointerHeuristic(\val{target}, \val{address})\colon                                      \\
            \codeline{17} & \tab \function{return}~\val{target}.highBits(\val{compareBits})                                               \\
                          & \tab \tab \tab \tab \tab = \val{address}.highBits(\val{compareBits})                                          \\
        \end{array}
    \end{equation*}
}

\newcommand\secureImp{
    \begin{equation*}
        \begin{array}{rl}
            \codeline{1}  & \function{function}~\train(\micro, \loc, \mem, \val{a})\colon                                                                       \\
            \codeline{2}  & \tab \micro := \update(\micro, \angles{\leaked{\loc}, \leaked{\val{a}}})                                                            \\
            \codeline{3}  & \tab \function{if}~\function{not}~\leaked{\isStrideAccess(\micro, \loc, \val{a})}~\function{then}                                   \\
            \codeline{4}  & \tab \tab \function{for}~\angles{\loc', \val{a_1}, \val{a_2}}~\function{in}~\getIndexAddresses(\micro)~\function{do}                \\
            \codeline{5}  & \tab \tab \tab \changedrule{\function{if}~\memsec(\val{a_1}) = \botsec~\function{and}~\memsec(\val{a_2}) = \botsec~\function{then}} \\
            \codeline{6}  & \tab \tab \tab \tab \function{for}~\val{s}~\function{in}~\shiftvalues~\function{do}                                                 \\
            \codeline{7}  & \tab \tab \tab \tab \tab \function{for}~\val{a_c}~\function{in}~\getCandidates(\micro)~\function{do}                                \\
            \codeline{8}  & \tab \tab \tab \tab \tab \tab \function{if}~\leaked{\val{a} - (\mem(\val{a_2}) \ll \val{s})                                         \\
                          & \tab \tab \tab \tab \tab \tab \tab = \val{a_c} - (\mem(\val{a_1}) \ll \val{s})}                                                     \\
                          & \tab \tab \tab \tab \tab \tab \function{then}                                                                                       \\
            \codeline{9}  & \tab \tab \tab \tab \tab \tab \tab \micro := \setIndirectPattern(\micro, \leaked{\loc'}, \leaked{\val{s}},                          \\
                          & \tab \tab \tab \tab \tab \tab \tab \tab \tab \tab \leaked{\val{a} - (\mem(\val{a_2}) \ll \val{s})})                                 \\
            \codeline{10} & \tab \tab \function{for}~\angles{\loc', \val{s}, \val{b}, \val{a_i}}~\function{in}~\getIndirectPatterns(\micro)~\function{do}       \\
            \codeline{11} & \tab \tab \tab \changedrule{\function{if}~\memsec(\val{a_i}) = \botsec~\function{then}}                                             \\
            \codeline{12} & \tab \tab \tab \tab \function{if}~\leaked{\val{a} = \val{b} + (\mem(\val{a_i}) \ll \val{s})}~\function{then}                        \\
            \codeline{13} & \tab \tab \tab \tab \tab \micro := \increaseConfidence(\micro, \leaked{\loc'})                                                      \\
            \codeline{14} & \tab \tab \micro := \appendCandidate(\micro, \leaked{\val{a}})                                                                      \\
            \codeline{15} & \tab \function{return}~\micro                                                                                                       \\
            \\
            \codeline{16} & \function{function}~\prefetch(\micro, \loc, \mem, \val{a})\colon                                                                    \\
            \codeline{17} & \tab \function{if}~\leaked{\isConfident(\micro, \loc)}~\function{then}                                                              \\
            \codeline{18} & \tab \tab \val{address} := \val{a} + \getDistance(\micro, \loc)                                                                     \\
            \codeline{19} & \tab \tab \changedrule{\function{if}~\memsec(\val{address}) = \botsec~\function{then}}                                              \\
            \codeline{20} & \tab \tab \tab \angles{\val{s}, \val{b}} := \getIndirectPattern(\micro, \loc)                                                       \\
            \codeline{21} & \tab \tab \tab \val{target} := \val{b} + (\mem(\val{address}) \ll \val{s})                                                          \\
            \codeline{22} & \tab \tab \tab \micro := \update(\micro, \leaked{\val{address}}, \leaked{\val{target}})                                             \\
            \codeline{23} & \tab \function{return}~\micro                                                                                                       \\
        \end{array}
    \end{equation*}
}

%% file: sections/00-abstract.tex
Data memory-dependent prefetchers (DMPs) are a microarchitectural optimization recently introduced in mainstream desktop and server CPUs.
Although DMPs offer performance improvements, they have been shown to enable side-channel leakage of secrets, even from constant-time code that has been hardened against Spectre attacks.
Unfortunately, existing defenses either incur significant performance overhead or, as we show, fail to prevent all DMP leakage.

In this work, we address this gap by designing a defense that is both secure and efficient.
First, we propose a classification of DMP leakage, categorize existing and novel attacks, and analyze which leakage classes are mitigated by existing defenses.
We validate this analysis empirically by implementing the first open-source DMPs on a RISC-V processor, showing that the only existing defense that aims to preserve DMP performance fails to protect against entire leakage classes.
Next, we propose SecDMP, a provably secure and efficient hardware-based defense.
SecDMP requires secrets to be marked in memory and ensures that DMPs never encode these secrets in the microarchitecture---without compromising performance for public data.
We formally prove that SecDMP introduces no leakage beyond a baseline semantics without DMPs.
Notably, when combined with speculative-execution defenses, SecDMP restores complete microarchitectural security for constant-time code.
Finally, our empirical evaluation shows that SecDMP prevents all DMP leakage while incurring minimal performance overhead and hardware cost.

%% file: sections/10-introduction.tex
\section{Introduction}
Microarchitectural side channels pose a significant challenge to secure software.
Since the first timing attack~\cite{kocher96timing}, a growing body of work has shown that many microarchitectural components and optimizations leak sensitive data.
Well-known leakage sources include instruction~\cite{aciicmez07icache} and data caches~\cite{bernstein05cache}, port contention~\cite{aldaya19port}, branch predictor state~\cite{yu23allyourpc}, and (classical) hardware prefetchers~\cite{bhattacharya12pref}.
Additionally, modern CPUs employ speculative execution induced by mechanisms such as branch prediction and memory disambiguation, which Spectre attacks~\cite{kocher19spectre,horn18spectrev4} exploit to leak secrets during transient execution caused by mispredictions.

To defend against microarchitectural side-channel attacks, the most commonly adopted approach is to implement software in a constant-time fashion~\cite{bernstein12nacl,pornin18constanttime,almeida17jasmin,almeida16verifying}.
This approach requires that secret data is never used as operands of computations that may produce operand-dependent observable microarchitectural effects.
Specifically, secrets must not be used as branch conditions, memory addresses, or operands of variable-time instructions (\eg{} divisions).
While constant-time programming itself does not protect against Spectre, defenses have been proposed in both software and hardware that secure constant-time code against these attacks~\cite{cauligi20ctfoundations,daniel23prospect}.

\textbf{\Acfp{DMP}.}
Recent work has shown that new generations of Apple, Intel, and Arm CPUs implement a novel type of prefetcher, known as a \acf{DMP}~\cite{chen24gofetch,vicarte22augury,wang25peekawalk,arm25dmp}.
Unlike classical prefetchers, which base their predictions on the history of accessed memory addresses, \acp{DMP} also consider the \emph{loaded data} itself.
This has severe security implications, as \acp{DMP} can encode secrets into microarchitectural state without the program ever explicitly (not even transiently) performing this encoding.
For instance, GoFetch~\cite{chen24gofetch} showed that the \ac{DMP} on modern Apple processors can leak secrets from constant-time cryptographic implementations.
Mitigating this type of leakage is particularly challenging, as it does not rely on any specific property of the victim program aside from containing sensitive data in memory.
The most practical proposed mitigations are to disable the \ac{DMP}, either completely~\cite{vicarte22augury,chen24gofetch,wang25peekawalk} or selectively~\cite{schluter25prefence,chen24gofetch}.

\textbf{This paper.}
In this work, we propose a systematic and formal approach to analyzing and mitigating \ac{DMP} leakage.
We start by classifying how DMPs access and leak secret data.
Our main insight is that DMPs can be activated not only on data \emph{architecturally} loaded by the victim program, but also \emph{indirectly} via
\begin{enumerate*}[label=(\arabic*)]
    \item adjacent data within a cache block,
    \item data loaded transiently, and
    \item data loaded by prefetchers (either the DMP itself or other prefetchers).
\end{enumerate*}
We analyze which commercial DMPs are prone to each leakage class based on prior reverse-engineering efforts, report which classes have been demonstrated in prior attacks, and identify the remaining unknown design points that could enable new attacks.
Importantly, we show that selective disabling~\cite{schluter25prefence} does not mitigate all leakage classes, making it an error-prone defense whose security is hard to reason about.
In addition to this informal analysis, we formally model existing DMP designs as extensions of a standard speculative-execution semantics.

\textbf{SecDMP.}
To overcome the limitations of existing defenses, we propose \secdmp{}, a new hardware-based defense mechanism for securing DMPs.
\secdmp{} requires that secrets are marked in memory---a technique employed in many Spectre defenses~\cite{daniel23prospect,yu19oisa,schwarz20context,fustos19spectreguard}---and prevents prefetches and state updates that are based on secret data.
Using our formal model, we prove that \secdmp{} prevents all secret leakage through the DMP.
Specifically, when combined with a baseline speculative semantics without a DMP, \secdmp{} \emph{does not introduce additional leakage} beyond the baseline.
In particular, when combined with a strong Spectre defense such as \prospect{}~\cite{daniel23prospect}, \secdmp{} restores full microarchitectural security for constant-time cryptographic code.

On the practical side, we implement several DMPs modeled after commercial designs~\cite{vicarte22augury,chen24gofetch,wang25peekawalk} and academic proposals~\cite{cooksey02dmp,yu15imp} on the Proteus \riscv CPU~\cite{bognar23proteus}.
Using our implementations, we demonstrate all DMP leakage classes through proof-of-concept (PoC) programs, based on existing and novel attacks.
To the best of our knowledge, these are the first open-source DMP implementations and PoC DMP attacks on \riscv.
To assess \secdmp{}'s cost and effectiveness, we evaluate two configurations: \secdmp{} on a speculative out-of-order baseline, which only prevents leakage through the DMP; and \secdmp{} combined with \prospect{}~\cite{daniel23prospect}, restoring full microarchitectural security for constant-time code.
Results show that \secdmp{} incurs minimal hardware cost and performance overhead: it increases chip area by \SI{0.5}{\percent} without affecting the critical path, and has an average performance overhead of \SI{0.36}{\percent} on our benchmarks with a maximum of \SI{3.3}{\percent}.
Furthermore, we experimentally confirm that \secdmp{} eliminates all secret leakage across our benchmark programs.

\textbf{Contributions.}
In summary, our contributions are:

\begin{itemize}
    \item We classify how \acp{DMP} leak secrets (\cref{sec:leakage-classification}). Our classification reveals that current mitigations are insufficient or expensive, which we experimentally validate.
    \item We present a modular formal model of DMPs that we instantiate with existing designs (\cref{sec:formalization}).
    \item We propose \secdmp{}, a hardware-based defense mechanism that prevents DMPs from leaking secrets (\cref{sec:overview}), with a formal proof of security (\cref{sec:proofs}).
    \item We implement the first open-source DMPs on a research processor, and implement and evaluate \secdmp{} and its combination with \prospect{}. Our results show that \secdmp{} is effective, incurring low hardware cost and minimal performance overhead (\cref{sec:implementation-evaluation}).
\end{itemize}

\textbf{Availability.}
All artifacts are available at \artifactrepo{}.
\extendedElse{This paper is an extended version of~\cite{sprengers27secdmp}, containing the full formalization and security proofs.}{An extended version of this paper with the full formalization and security proofs is available at~\cite{sprengers27secdmpextended}.}

%% file: sections/20-background.tex
\section{Background}\label{sec:background}

\subsection{Microarchitectural Side-Channel Attacks}\label{sec:side-channel-attacks}

Microarchitectural side-channel attacks expose a victim's secrets by observing secret-dependent effects on hardware components shared between attacker and victim, such as caches, branch predictors, and prefetchers~\cite{zhang24sidechannels,ge18survey,szefer19survey}.
Notably, a wide range of attacks expose a victim's memory usage by (partially) leaking accessed memory addresses, which might depend on secret data.
Prominent examples include cache attacks such as Flush+Reload~\cite{yarom14flushreload} and Prime+Probe~\cite{osvik06cache}, address-translation attacks targeting TLBs and page walks~\cite{gras18tlbleed,tatar22tlbdr,gras17anc,wang25peekawalk}, and contention-based attacks on shared interconnects~\cite{paccagnella21lotr}.
While most of these attacks leak information within a single CPU and at cache-block granularity, prior work has demonstrated leakage across CPUs~\cite{irazoqui16crossprocessor,pessl16drama} and even recovery of fine-grained address bits~\cite{sieck23teejam,moghimi19memjam,yarom17cachebleed}.

\subsubsection{Constant-time programming}
Constant-time is a programming discipline addressing the root cause of side-channel leakage by ensuring that secrets do not influence observable microarchitectural behavior~\cite{almeida16verifying}, and has become the standard practice for securing cryptographic software~\cite{bernstein12nacl,pornin18constanttime,almeida17jasmin}.
Concretely, it requires memory accesses, control flow, and variable-latency instruction operands to be secret-independent.

\subsubsection{Spectre}
Speculative execution improves performance by executing instructions ahead of time based on predictions.
Modern processors include various speculation mechanisms~\cite{xiong21survey,canella19systematic}, such as conditional and indirect branch prediction~\cite{kocher19spectre}, memory disambiguation~\cite{horn18spectrev4}, load value prediction~\cite{kim25flop}, and load address prediction~\cite{kim25slap}.
Instructions executed based on incorrect predictions, known as \emph{transient} instructions, are architecturally reverted to preserve functional correctness.
However, their microarchitectural effects remain and can be observed through side-channel attacks.
Spectre attacks~\cite{kocher19spectre} exploit this by deliberately causing a victim program to mispredict and transiently execute instructions that produce secret-dependent microarchitectural effects.

\subsubsection{Spectre defenses}
Transient execution attacks showed that classical constant-time programming is insufficient: even architecturally constant-time programs may leak secrets via misspeculated execution~\cite{cauligi20ctfoundations}.
As a result, the constant-time discipline has been extended to \emph{speculative constant-time}~\cite{cauligi20ctfoundations}, which additionally requires the protection of secrets during transient execution.
However, enforcing speculative constant-time purely in software is extremely challenging.
Consequently, a lot of research has focused on hardware-based defenses against leakage during speculation~\cite{hu21hardware}.
Among these, a particularly promising class relies on explicitly informing the hardware which data should be protected~\cite{daniel23prospect,yu19oisa,schwarz20context,fustos19spectreguard}.
This enables the hardware to restrict speculation on secrets while preserving full performance for public data.

\subsection{Hardware Prefetchers}
Prefetching is an optimization technique to reduce memory latency.
Hardware prefetchers predict which addresses will be accessed by a program in the near future and bring the corresponding data into the cache before they are requested by the program.
If the prediction is correct, prefetching increases performance as it avoids the penalty of a cache miss.

Hardware prefetchers can be categorized into two groups: \emph{classical prefetchers} and \emph{\acfp{DMP}}~\cite{vicarte21pandora,vicarte22augury,chen24gofetch}.
Classical prefetchers are widely implemented on commercial processors and make predictions by monitoring a program's accessed memory addresses.
For instance, a \emph{classical stride prefetcher} detects array traversals with a constant stride, prefetching elements ahead in the array~\cite{vicarte22augury}.
\emph{\acp{DMP}}, however, have only recently been implemented on commercial processors following years of academic proposals~\cite{vicarte22augury,chen24gofetch}.
DMPs base their predictions not only on the accessed memory addresses but also on the \emph{data itself} stored in memory.
This enables them to identify irregular access patterns for which classical prefetchers fail~\cite{vicarte21pandora}.
For example, they can speed up pointer dereferences by scanning memory for pointers and prefetching their targets.
In 2022, Vicarte \etal~\cite{vicarte22augury} demonstrated for the first time the existence of a \ac{DMP} on a commercial processor, namely the Apple M1.
As of now, several Apple~\cite{vicarte22augury,schluter23fetchbench,chen24gofetch}, Intel~\cite{wang25peekawalk}, and Arm~\cite{arm25dmp} processors are known to include a DMP.

Next, we identify three key mechanisms in existing DMP designs based on the literature~\cite{cooksey02dmp,yu15imp} and reverse-engineering~\cite{vicarte22augury,schluter23fetchbench,chen24gofetch,wang25peekawalk}.
We discuss concrete examples of commercial and academic designs in \cref{sec:dmp-models}.

\subsubsection{Address heuristics}
DMPs read data to derive target addresses to prefetch, optionally applying a heuristic to filter out invalid targets.
For example, \emph{pointer-chasing DMPs} scan cache blocks for values that resemble valid pointers and prefetch the corresponding addresses~\cite{vicarte22augury,wang25peekawalk}, relying on \emph{pointer heuristics} to distinguish pointers from other data.

\subsubsection{Recursion and chain prefetching}
DMPs typically read data loaded directly by the program~\cite{cooksey02dmp,yu15imp,vicarte22augury,wang25peekawalk}, but they may also read data loaded by a \emph{prefetcher}.
First, DMPs can scan data that they prefetched themselves, which is known as \emph{recursive prefetching}~\cite{cooksey02dmp}.
For instance, such DMPs can recursively prefetch entire linked data structures ahead of time.
Second, DMPs can also scan data prefetched by other prefetchers~\cite{chen24gofetch}; we call this \emph{chain prefetching}.

\subsubsection{Activation patterns}
Some DMPs only activate after detecting a specific memory access pattern~\cite{yu15imp,wang25peekawalk}.
They first train on the program's memory accesses to learn a data-dependent pattern, then prefetch ahead in the learned pattern.
These DMPs are often recursive, as data-dependent patterns require reading data ahead in the pattern to derive the next prefetch target.
For example, the Intel DMP can traverse an array of pointers~\cite{wang25peekawalk}: it first prefetches ahead in the array, then reads the returned pointer and prefetches its target address.

\subsection{Security Risks of DMPs} \label{sec:dmp-security-risks}
The effects of prefetchers on microarchitectural state match those of load requests: prefetching a (virtual) address first triggers an address translation, and, if successful, fetches the data into the cache.
As a result, any microarchitectural attack that can detect loaded addresses can be used to infer prefetch targets, including standard cache attacks.
Additionally, some attacks exploit prefetchers' \emph{internal states} directly~\cite{chen23afterimage,hetterich25shadowload}.
As such, any secret data that influences the behavior of a prefetcher (\ie the information analyzed to determine its internal state and prefetch targets) can potentially leak.

Vicarte \etal~\cite{vicarte21pandora} were the first to warn about the security risks of DMPs, showing that the academic IMP DMP~\cite{yu15imp} can serve as a universal read gadget through out-of-bounds prefetching.
Later, Augury~\cite{vicarte22augury} identified the first commercial DMP on Apple processors and showed that it can leak data beyond the bounds of an array of pointers.
Barthe \etal formalized this leakage behavior in LmSpec~\cite{barthe24future}, detecting leakage in constant-time cryptographic implementations under this model.
Subsequently, GoFetch~\cite{chen24gofetch} demonstrated how the Apple DMP can be exploited in practice to leak secret keys from constant-time code, performing attacks on several cryptographic primitives.
Because Apple's pointer heuristic is fairly strict, the GoFetch attacks are limited to specific scenarios where the attacker forces a pointer to appear in the victim's memory only for some variants of the secret.
Peek-a-Walk~\cite{wang25peekawalk} showed that an attacker can deliberately mistrain the Intel DMP so that it activates during the victim's execution and leaks data at an address of the attacker's choosing within certain constraints.
Their attack demonstrates that DMPs with less restrictive pointer heuristics have more potential to expose secrets.
In this case, the impact is limited, as it is not possible to train the Intel DMP across processes~\cite{wang25peekawalk}.

To address the identified leakage caused by DMPs, a number of countermeasures were proposed in the literature.
We discuss and analyze these in \cref{sec:existing-defenses}.

%% file: sections/25-threat-model.tex
\section{Threat Model}

In line with related work, we assume an attacker who can monitor all microarchitectural side channels introduced in \cref{sec:side-channel-attacks}, including those induced by speculative execution.
The attacker is abstract and may operate from an intra-process sandbox~\cite{wang25peekawalk}, a sibling SMT context, or another processor core.
Our goal is to protect against side-channel leakage introduced by DMPs---specifically, we care about leakage relative to a baseline with DMPs disabled.
In \cref{sec:secdmp-combination-secure-speculation}, we additionally show that full microarchitectural security is achieved when complementary countermeasures are in place for other leakage, either in software (speculative constant-time~\cite{cauligi20ctfoundations}) or in hardware (secure speculation defenses~\cite{daniel23prospect}).

%% file: sections/30-leakage-classification.tex
\section{Classification of DMP Leakage}
\label{sec:leakage-classification}
To understand the different ways in which DMPs may leak secret data, we propose a leakage classification based on existing DMP designs.
We focus our analysis on designs described in the literature, consisting of the commercial designs of the Apple~\cite{vicarte22augury,schluter23fetchbench,chen24gofetch} and Intel~\cite{wang25peekawalk} DMPs as reverse-engineered by prior work, and two academic designs~\cite{cooksey02dmp,yu15imp} with unique properties.
First, we present our classification in \cref{sec:leakage-classes} and describe the existing DMP designs and how they fit into the leakage classes in \cref{sec:dmp-models}.
Next, we categorize existing DMP attacks into our leakage classes, devise new attack techniques, and demonstrate the leakage and attacks in practice (\cref{sec:leakage-exploitation}).
For the demonstration, we implement several DMPs on a \riscv processor.
Finally, we describe existing defenses (\cref{sec:existing-defenses}) and provide an informal analysis and empirical validation of which leakage classes they mitigate (\cref{sec:defense-classification}).
Our results show that the most promising defense, selective disabling, does not mitigate all leakage, highlighting the need for secure DMP defenses.

\subsection{DMP Leakage Classes}
\label{sec:leakage-classes}
All attacks that exploit DMPs involve at least three steps:
\begin{enumerate*}[label=(\arabic*)]
    \item the DMP scans secret data,
    \item the data is encoded into the microarchitecture, and
    \item the attacker uses existing side-channel attacks to recover it~\cite{vicarte22augury}.
\end{enumerate*}
As discussed previously, DMPs have the potential to leak any data they scan in step (1), making this the crucial step to reason about DMP leakage.
We therefore classify DMP leakage based on the sources of data that the DMP scans for both training and prefetching.

\subsubsection{Architecturally accessed data (Arch)}
The \emph{Arch} class represents all data that is architecturally accessed (via a non-transient instruction) by the program.
This is the baseline source of data for any DMP: when the program architecturally loads data from memory, the DMP may scan it to derive prefetch targets, potentially leaking the data.

\subsubsection{Indirectly accessed data (Ind)}
In addition to \emph{Arch}, DMPs can also scan and leak data that the program never architecturally accesses, divided into three subclasses.

\paragraph{Sub-cache-block data (Sub)}
When accessing memory, the processor typically loads an entire cache block, \eg a 64-byte cache block for a 4-byte load instruction.
The data adjacent to the architecturally accessed value in the cache block may also be scanned by DMPs.

\paragraph{Transient loads (Trans)}
Due to incorrect speculative execution, the processor may issue load requests that would never occur during architectural execution.
Some DMPs may read the data accessed by such transient loads.
We expect that no real-world DMP distinguishes between transient and non-transient loads, and thus leaks through the \emph{Trans} class.
Distinguishing them would require complex hardware and reduce performance, as at the time a speculative load is executed, it is not yet known whether it will be squashed.

\paragraph{Prefetched data (Pref)}
Finally, DMPs can activate on data without any load instruction (architectural or transient) explicitly requesting it.
This happens when the data was accessed through prefetching.
There are two ways in which a DMP can scan data belonging to the \emph{Pref} class: recursion and chain prefetching; reading data prefetched by itself or by another prefetcher, respectively.

\subsection{DMP Designs and Classification} \label{sec:dmp-models}
We now describe the four DMP designs of interest and categorize their leakage into the classes above.
The results are summarized in \cref{tab:leakage-classification}.

\begin{table}
    \centering
    \caption{Leakage classes applicable to different DMP designs, explicitly documented (\CIRCLE) or presumed (\Circle).}
    \label{tab:leakage-classification}
    \begin{tabular}{l c c c c}\toprule
        \multirow{2}{*}{DMP}                           & \multirow{2}{*}{\emph{Arch}} & \multicolumn{3}{c}{\emph{Ind}}                              \\ \cmidrule{3-5}
                                                       &                              & \emph{Sub}                     & \emph{Trans} & \emph{Pref} \\ \midrule
        Apple DMP~\cite{vicarte22augury,chen24gofetch} & \CIRCLE                      & \CIRCLE                        & \CIRCLE      & \CIRCLE     \\
        Intel DMP~\cite{wang25peekawalk}               & \CIRCLE                      &                                & \Circle      & \CIRCLE     \\
        CDP~\cite{cooksey02dmp}                        & \CIRCLE                      & \CIRCLE                        & \Circle      & \CIRCLE     \\
        IMP~\cite{yu15imp}                             & \CIRCLE                      &                                & \Circle      & \CIRCLE     \\ \bottomrule
    \end{tabular}
\end{table}

\subsubsection{Apple DMP}\label{sec:dmp-models-apple}
The Apple DMP~\cite{chen24gofetch} is a pointer-chasing DMP.
When a data block is loaded into the L1 cache, the DMP applies a pointer heuristic to every 64-bit value in the block, referred to as candidate pointers.
The DMP scans all blocks loaded into the L1 cache for candidate pointers, including from transient loads, hence leaking through the \emph{Arch}, \emph{Sub}, and \emph{Trans} classes.
A prefetch request is then initiated for each candidate pointer that passes the heuristic, which checks whether bits 55--32 of the candidate pointer match those of the current block address.
While GoFetch~\cite{chen24gofetch} does not report any recursive prefetching behavior, it does report the DMP triggering a next-line prefetcher, fetching two adjacent blocks.

GoFetch additionally observes behavior that is not explained by its model: looping through an array of random values triggers the DMP to dereference pointers located after the array, demonstrating \emph{Pref} leakage.
While the exact mechanism responsible for this leakage is unknown, we believe it to be chain prefetching, where a stride prefetcher first fetches data beyond the array, which the DMP then scans for pointers.

\subsubsection{Intel DMP}\label{sec:dmp-models-intel}
The Intel DMP~\cite{wang25peekawalk} is a pointer-chasing DMP requiring a specific activation pattern: accessing and dereferencing the pointers inside an array of pointers (AoP).
During its training phase, the DMP learns the stride of array accesses issued in a loop (detected by storing and comparing the 10 lowest bits of the program counter (PC) of the load instructions), and checks whether the accessed array entries are dereferenced.
The DMP thus analyzes loaded addresses and values, leaking through the \emph{Arch} class.
Whether the Intel DMP also reads transient loads is unknown, but we presume \emph{Trans} leakage to be present (cf. \cref{sec:leakage-classes}).

When the DMP detects the AoP pattern and learns the stride \val{k}, it activates whenever it encounters a load instruction with the matching 10 lowest PC bits: it prefetches \val{k*16} entries ahead in the array, applies a pointer heuristic to the value when the prefetch returns, and dereferences the value if the heuristic succeeds.
This recursive behavior thus enables leakage through the \emph{Pref} class: the DMP can read data it prefetched.

Finally, the Intel DMP limits the stride by requiring the current entry and the entry to prefetch to be located in the same \SI{256}{\kibi\byte} aligned memory region, and its pointer heuristic restricts the prefetching to user memory.

\subsubsection{CDP}\label{sec:dmp-models-cdp}
The content-directed data prefetcher (CDP)~\cite{cooksey02dmp}, predating commercial designs, is a recursive variant of the Apple DMP: it scans cache blocks for pointers from both load and prefetch requests, leaking through the \emph{Arch}, \emph{Sub}, and \emph{Pref} classes.
CDP's design does not explicitly mention transient loads; however, since the DMP monitors all memory traffic for a specific cache level, it would also scan transient loads, capturing the \emph{Trans} class.
CDP applies a pointer heuristic to each candidate pointer of the scanned cache block, checking whether a configurable number of upper bits of the candidate pointer and the address at which it is stored match.
CDP also features a configurable number of next lines to prefetch alongside any target address that passes the heuristic.
To prevent recursive prefetching from continuing indefinitely, it has a configurable maximum recursion depth.

\subsubsection{IMP}\label{sec:dmp-models-imp}
The indirect memory prefetcher (IMP)~\cite{yu15imp} monitors memory traffic for the pattern \val{A[B[i]]}, \ie indexing one array with values from another.
After detecting this pattern, IMP prefetches and reads values ahead in the first array, adds them to the base of the second array, and prefetches the resulting address.
This recursive behavior enables leakage through the \emph{Arch} and \emph{Pref} classes.
IMP's design does not mention transient loads, but when implemented as described, it would not distinguish them from non-transient loads, leaking through the \emph{Trans} class.
In addition to the 2-level variant, the authors proposed a 3-level variant \val{A[B[C[i]]]}.
Notably, IMP does not use a pointer heuristic to filter prefetch targets.

\subsection{Exploiting Leakage Classes} \label{sec:leakage-exploitation}
We now illustrate how programs can leak secrets through each leakage class, classifying existing DMP attacks and devising new techniques.
We demonstrate the leakage classes and attacks in practice by crafting simple test programs for the Proteus \riscv processor~\cite{bognar23proteus}.
To enable the demonstration, we extended Proteus with the four DMPs from \cref{sec:dmp-models}, \ie the Apple DMP~\cite{vicarte22augury,chen24gofetch}, Intel DMP~\cite{wang25peekawalk}, CDP~\cite{cooksey02dmp} and IMP~\cite{yu15imp}, modeling them as faithfully as possible after the descriptions in the literature.
The test programs and results of the leakage demonstration are listed in \cref{tab:programs-classification}.
For more details regarding the implementation and evaluation methodology, we refer to \cref{sec:implementation-evaluation}.

\newcommand{\xcat}[2]{\multirow{#1}{*}{\rotatebox[origin=c]{90}{\footnotesize\textsl{#2}}}}
\definecolor{existingbg}{RGB}{252,228,228}
\newcommand{\hlexisting}[1]{{\setlength{\fboxsep}{0pt}\smash{\colorbox{existingbg}{\strut #1}}}}
\begin{table}
    \centering
    \caption{Leakage classes targeted by each of our test programs. Programs adapted from existing attacks are \protect\hlexisting{highlighted}.}
    \label{tab:programs-classification}
    \begin{tabular}{c >{\columncolor{white}[3pt][\tabcolsep]}l c c c c}
        \toprule
                                                & \multirow{2}{*}{Program}           & \multirow{2}{*}{\emph{Arch}} & \multicolumn{3}{c}{\emph{Ind}}                              \\
        \cmidrule{4-6}
                                                & \multicolumn{1}{l}{}               &                              & \emph{Sub}                     & \emph{Trans} & \emph{Pref} \\
        \midrule
        \rowcolor{existingbg} \cellcolor{white} & IMP Leak~\cite{vicarte21pandora}   &                              &                                &              & \CIRCLE     \\
        \rowcolor{existingbg} \cellcolor{white} & Augury PoC~\cite{vicarte22augury}  &                              &                                &              & \CIRCLE     \\
        \rowcolor{existingbg} \cellcolor{white} & GoFetch PoC~\cite{chen24gofetch}   & \CIRCLE                      &                                &              &             \\
        \rowcolor{existingbg} \cellcolor{white} & Peek-a-Walk~\cite{wang25peekawalk} &                              &                                &              & \CIRCLE     \\
                                                & DMP-Spectre                        &                              &                                & \CIRCLE      &             \\
        \xcat{-6}{Attacks}                      & Recursion                          &                              &                                &              & \CIRCLE     \\
        \midrule
        \addlinespace[5pt]
        \xcat{3}{General}                       & Load Secret                        & \CIRCLE                      &                                &              &             \\
                                                & Same Cache Block                   &                              & \CIRCLE                        &              &             \\
                                                & Transient Load                     &                              &                                & \CIRCLE      &             \\
        \addlinespace[2pt]
        \bottomrule
    \end{tabular}
\end{table}

\subsubsection{Classification of existing attacks}
We revisit the attacks from \cref{sec:dmp-security-risks} and classify them according to our leakage classes.
Vicarte \etal~\cite{vicarte21pandora} showed that the 3-level IMP~\cite{yu15imp} can serve as a universal read gadget through out-of-bounds prefetching, hence leaking through the \emph{Pref} class.
Augury~\cite{vicarte22augury} showed that pointer-chasing DMPs that look for an AoP pattern can leak secrets located after the array's bounds, also falling in the \emph{Pref} class.
These two cases highlight the dangers of the recursive behavior of pattern-activated DMPs: when they prefetch too far ahead in the pattern, they may access data out-of-bounds and use it to derive further prefetch targets, thereby leaking it.
In GoFetch's attacks~\cite{chen24gofetch}, the DMP leaks whether an architecturally loaded secret resembles a valid pointer, leaking through \emph{Arch}.
In the Peek-a-Walk DMP attack~\cite{wang25peekawalk}, an attacker deliberately mistrains the Intel DMP into activating during the victim's runtime, causing it to load, scan, and leak a secret located at an attacker-influenced address, \ie \emph{Pref} leakage.

We demonstrate these attacks via the following programs: \underline{IMP Leak} demonstrates a 2-level IMP attack, \underline{Augury PoC} shows the out-of-bounds leakage of AoP DMPs, \underline{GoFetch PoC} replicates GoFetch's proof-of-concept attack on OpenSSL's constant-time conditional swap, and \underline{Peek-a-Walk} mistrains the Intel DMP into activating and leaking a secret.

\subsubsection{Novel attacks exploiting leakage classes} \label{sec:new-dmp-attacks}
We now devise novel attacks that rely on speculative execution and recursive prefetching, leaking through \emph{Trans} and \emph{Pref}, respectively.

\paragraph{DMP-Spectre}
Attackers can exploit DMPs that are capable of leaking through the \emph{Trans} class (see \cref{tab:leakage-classification}) in a manner similar to Spectre attacks.
While typical Spectre attacks require a gadget that transiently loads a secret and subsequently leaks it through a follow-up instruction~\cite{xiong21survey}, DMPs can replace this follow-up instruction.
For example, by deliberately mistraining prediction mechanisms, the attacker can induce the victim to transiently load a secret, which the DMP can then scan and leak.
We demonstrate the potential of these attacks with the \underline{DMP-Spectre} program, which loops through a public array and triggers a transient load of a secret located beyond the array bounds due to branch misprediction.

\paragraph{Recursion and chain prefetching}
In the case of a recursive pointer-chasing DMP, an attacker can embed the address of a secret in user input.
The victim will then load this input, causing the DMP to interpret it as a pointer and fetch the secret data.
Because the DMP is recursive, it reads the prefetched secret, potentially leaking it.
This attack vector was also mentioned but not exploited by Schlüter \etal~\cite{schluter25prefence}.
We demonstrate it for CDP via the \underline{Recursion} test program.
For chain prefetching, the attacker can mistrain the first prefetcher into loading a secret during the victim's execution, causing the chained DMP to activate and read the secret.

\subsubsection{General tests}
Beyond the attacks described above, we run three general leakage tests for the \emph{Arch}, \emph{Sub}, and \emph{Trans} classes.
Unlike the attacks, these tests are not tailored to specific DMPs.
\underline{Load Secret} simply loads a secret, the simplest example of architectural leakage.
\underline{Same Cache Block} accesses public data sharing a cache block with a secret, demonstrating the \emph{Sub} class.
While our test deliberately co-locates secret and public data, this can occur naturally in real programs, \eg when compilers reuse stack space by overwriting unused secrets with public data~\cite{olmos25ctcomp}.
Finally, \underline{Transient Load} shows \emph{Trans} leakage by transiently loading a secret due to misspeculation.

\subsection{Proposed DMP Mitigations}
\label{sec:existing-defenses}

Given our new classification of DMP leakage, we now apply the same concept to mitigations.
First, we briefly introduce defense proposals before classifying them in \cref{sec:defense-classification}.

\subsubsection{Permanent disabling}
The simplest, often proposed defense is disabling the DMP for the entire program's execution~\cite{vicarte22augury,chen24gofetch,wang25peekawalk}, \eg via software-accessible DIT/DOIT bits~\cite{arm25dit,intel25doit}.
While permanent disabling addresses the vulnerabilities of DMPs, it imposes a potentially heavy performance penalty~\cite{vicarte22augury,chen24gofetch,schluter25prefence}.
For example, a program with only a small security-critical code region will lose all the performance benefits of DMPs, including when operating on public data.

\subsubsection{Selective disabling}\label{sec:selective-disabling-defense}
To address the performance cost of permanent disabling, some papers propose to selectively disable DMPs only when computing on secret data~\cite{schluter25prefence,chen24gofetch}.
Apple and Intel similarly recommend selectively setting the DIT and DOIT bits only during security-critical operations~\cite{arm25dit,intel25doit}.
Recently, Schlüter \etal{} implemented selective disabling of all prefetchers in PreFence~\cite{schluter25prefence}, reporting a significant performance improvement over permanent disabling.

\subsubsection{Preventing secret extraction}
Some defenses attempt to prevent attackers from extracting secrets from the microarchitecture.
For instance, PhantomFetch~\cite{zhang25phantomfetch} flushes the prefetcher's state on context switches.
However, as prefetching affects many microarchitectural components, it is impractical to secure all leakage channels without incurring high overhead.
Additionally, this defense is likely insecure in intra-process sandbox settings or when the attacker can continuously monitor the victim from an SMT thread or another core.

\subsubsection{Interleaving}
Pätschke \etal~\cite{patschke25zebrafix} proposed interleaving secret data with specific values in memory as a theoretical defense to prevent secrets from activating the DMP pointer heuristic.
This defense only applies to DMPs with a pointer heuristic and must be adapted to each DMP's specific heuristic~\cite{chen24gofetch,wang25peekawalk,cooksey02dmp}.
Pätschke \etal also note the high performance cost and instead suggest using DMP disabling, which Sharma \etal~\cite{sharma26splittingsecrets} confirmed: their implementation of interleaving incurred a performance overhead of \SI{274}{\percent}.

\subsection{Classification of DMP Mitigations}
\label{sec:defense-classification}

Since flushing the prefetcher state is insecure in our threat model and interleaving is less secure and has much higher overhead than permanent disabling~\cite{patschke25zebrafix,sharma26splittingsecrets}, we focus on permanent and selective disabling.
The theoretical analysis of this section is complemented by an experimental validation and performance measurements on Proteus~\cite{bognar23proteus} in \cref{sec:evaluation}.

Trivially, with \ul{no defense}, no leakage class is mitigated.
Conversely, with \ul{permanent disabling}, DMPs cannot scan any data, mitigating all their leakage.
However, all performance improvements of the DMPs are lost.

\begin{table}
    \centering
    \caption{Leakage classes mitigated by different DMP defenses.}
    \label{tab:defense-classification}
    \begin{tabular}{l c c c c}\toprule
        \multirow{2}{*}{Defense} & \multirow{2}{*}{\emph{Arch}} & \multicolumn{3}{c}{\emph{Ind}}                              \\ \cmidrule{3-5}
                                 &                              & \emph{Sub}                     & \emph{Trans} & \emph{Pref} \\ \midrule
        No Defense               & \xmarkr                      & \xmarkr                        & \xmarkr      & \xmarkr     \\
        Permanent Disabling      & \cmarkg                      & \cmarkg                        & \cmarkg      & \cmarkg     \\
        Selective Disabling      & \cmarkg                      & \xmarkr                        & \xmarkr      & \xmarkr     \\
        \secdmp{}                & \cmarkg                      & \cmarkg                        & \cmarkg      & \cmarkg     \\ \bottomrule
    \end{tabular}
\end{table}

\label{sec:selective-disabling-leakage-classification}
\ul{Selective disabling}~\cite{schluter25prefence,chen24gofetch} consists of disabling the DMP during security-critical code regions (\ie regions where the program architecturally accesses secret data), effectively preventing leakage of the \emph{Arch} class.
However, leakage of the \emph{Ind} classes can still occur in public code regions where DMPs are enabled.
Schlüter \etal~\cite{schluter25prefence} acknowledge that selective disabling does not prevent leakage for recursive DMPs (\emph{Pref}), and advise to disable DMPs when looping through public data if secret data is located after it.
However, this additional disabling only prevents accidental leakage of the \emph{Pref} class caused by out-of-bounds prefetching, without addressing other types of indirect leakage.
Attackers can still deliberately mistrain predictors to trigger \emph{Trans} and \emph{Pref} leakage during public code regions, and \emph{Sub} leakage can occur in non-trivial ways when public and secret data are mixed, e.g., on the stack.
In conclusion, DMPs cannot be safely enabled in public code regions while secret data is accessible in memory.

This highlights a critical gap in DMP defenses: users either need to sacrifice security or suffer performance degradation from permanent DMP disabling.

%% file: sections/35-overview.tex
\section{Overview of \secdmp{}}\label{sec:overview}

To overcome the inefficiency of permanent disabling and the insecurity of selective disabling, we propose \secdmp{}, a general defense strategy for securing DMP designs against data leakage.
\secdmp{} relies on a hardware-software co-design: on the software side, developers mark all data that is secret, while on the hardware side, the DMP is prevented from using secret data for training or prefetching.

Importantly, \secdmp{} prevents all secret leakage through the DMP while imposing minimal performance overhead: it prevents only prefetches that depend on secret data without hindering any prefetches based on public data.

\subsection{Software Changes}

On the software side, application data must be partitioned into public and secret data.
\secdmp{} can be combined with different methods for partitioning~\cite{schwarz20context,daniel23prospect,fustos19spectreguard,yu19oisa}, as long as secret data is correctly labeled.

Our security proof relies on a static memory partition, which we implement in \cref{sec:implementation} using coarse-grained static secret regions, similar to prior academic defenses~\cite{daniel23prospect,yu19oisa,schwarz20context,fustos19spectreguard} and industry patents by Arm~\cite{arm_spec_restricted_memory_2023} and NVIDIA~\cite{nvidia_cacheable_non_speculative_2020}.
While coarse-grained static partitioning introduces several challenges, recent work has shown that it can be practical~\cite{song25securing}; we discuss these challenges and solutions in \cref{sec:practicality-partitioning}.

\subsection{Hardware Changes}
On the hardware side, \secdmp{} ensures that data labeled as secret is never encoded in the DMP state or used for initiating prefetch requests.
Concretely, \secdmp{} requires only two hardware modifications.
The first is the aforementioned mechanism for identifying secret data in memory by associating loaded data with a taint bit.
Second, the DMP itself must be modified to recognize this taint bit and ensure that it ignores tagged data.
Importantly, these changes are local and lightweight.
Unlike existing Spectre defenses based on taint tracking~\cite{daniel23prospect,yu19oisa,schwarz20context,fustos19spectreguard}, \secdmp{} itself does not require secrets to be tracked throughout the core, making adoption more realistic in complex CPU designs---though \secdmp{} can optionally be paired with dynamic secret-tracking schemes as an alternative to static partitioning.

\subsection{Security of SecDMP}
\secdmp{} closes all leakage channels introduced by DMPs for data marked as secret.
We formally prove this property in \cref{sec:proofs} (\cref{thm:secdmpsecurity,thm:secdmpbase-security}) for a static secret partitioning: \secdmp{} provides the same security guarantees as a baseline configuration without a DMP.

Programs may still leak secrets through other microarchitectural channels, such as standard side-channel attacks (\eg{} exploiting caches or port contention) and Spectre attacks.
To achieve comprehensive microarchitectural security, \secdmp{} must therefore be combined with existing defenses against such attacks.
If a program does not leak secrets in the baseline configuration without a DMP, it will also not leak secrets when the baseline is extended with a \secdmp{}-secured DMP.

Importantly, even cryptographic programs that follow the strict speculative constant-time discipline, and are thus secure against standard side-channel and Spectre attacks, can remain vulnerable to DMP-based attacks.
\secdmp{} closes this additional leakage channel and restores the security guarantees offered by speculative constant-time defenses.
For instance, \secdmp{} can be combined with existing hardware-based~\cite{daniel23prospect,davoli26dfence,yu19oisa,schwarz20context,fustos19spectreguard} or software-based~\cite{DBLP:conf/uss/ZhangBCSY23,olmos25protecting} Spectre defenses; we show this formally in \cref{sec:secdmp-combination-secure-speculation}.

%% file: sections/40-formalization.tex
\section{Formalization of DMPs and \secdmp{}}\label{sec:formalization}
To formally reason about DMP security, we propose a formal semantics capturing DMP behavior.
We start by introducing a baseline out-of-order semantics (\(\protect\basesem{}\), \cref{sec:baseline-semantics}), which we extend to model DMP behavior (\(\protect\dmpsem{}\), \cref{sec:formalization-existing-dmps}).
Finally, we formalize \secdmp{} as an extension of an arbitrary base semantics (\(\protect\secdmpbasesem{}\), \cref{sec:secdmp-main}).
\if\extendedVersion0
We keep the formalization high-level; the full formalization can be found in the extended version~\cite{sprengers27secdmpextended}.
\fi

\subsection[Baseline Speculative Semantics]{Baseline Speculative Semantics (\(\protect\basesem{}\))}\label{sec:baseline-semantics}

\subsubsection{Language}
The language in \Cref{fig:syntax} defines our ISA. It contains instructions for moving data between registers and memory, as well as jumps and conditional branches. Direct assignments to the program counter are not allowed.

\subsubsection{Configurations}
A hardware configuration \(\mconf{}\) is a tuple \(\mconffull{}\) where
\(\mem: \Vals \to \Vals\) is a memory mapping addresses to values;
\(\reg: \Regs \to \Vals\) is a register map of registers to values;
$\buf$ is a reorder buffer (ROB), a sequence of (possibly
transient) instructions;
and $\micro$ is the \emph{microarchitectural context}.
Similar to existing formal semantics~\cite{daniel23prospect}, we assume that \(\micro\) is the \emph{observable part of the microarchitectural state}: attackers can observe \(\micro\) at each execution step.
Finally, we also refer to \(\conf{\mem, \reg}\) as the \emph{architectural configuration} and denote it \(\aconf\).

During the execution, we (partially) apply changes that occur in the ROB using a function \(\apl(\buf, \reg)\), which applies all effects of instructions in \(\buf\) to
\if\extendedVersion1
\(\reg\) (\cref{def:apl} in \cref{app:baseline-semantics}).
\else
\(\reg\)~\cite{sprengers27secdmpextended}.
\fi
When the value of an assignment in \(\buf\) is not resolved yet, the corresponding register is mapped to a special symbol \(\botval\), meaning that it is undefined.
Thus, \(\apl\) generates a new register map where the value of some registers is undefined.

\subsubsection{Attacker strategies}
Our semantics is parameterized by a set of functions \(\funcs = \{\nextdir, \predict, \update\}\), which model the next directive to execute, control flow predictions, and microarchitectural state updates, respectively.
These functions can be adversarially chosen to model different attacker strategies. We call an \emph{attacker strategy} \(\attacker = (\micro_0, \funcs)\) a choice of initial microarchitectural state \(\micro_0\) and functions \(\funcs\).
Given a program \(\prog{}\) and an attacker strategy \(\attacker = (\micro_0, \funcs)\), we let \(\bigmicro{}{\prog{}}{\aconf{}}{\attacker}\) denote the sequence of observable microarchitectural states \(\micro_{0} \dots \micro_{n}\) produced by the execution of \(\prog{}\) from the initial configuration \(\conf{\aconf{}, \varepsilon, \micro_0}\) in a semantics parameterized by functions \(\funcs\).
For an \(n\)-step execution \(\conf{\aconf{}, \varepsilon, \micro_0} \relmicro{}{} \dots \relmicro{}{} \conf{\aconf{}_n, \buf_n, \micro_n}\), \(\bigmicro{}{\prog{}}{\aconf{}}{\attacker} = \micro_{0} \dots \micro_{n}\).

\input{semantics/language}

\subsubsection{Leakage model}
The function \(\update\) models information leakage by updating the microarchitectural context \(\micro\).
Intuitively, all arguments of the update function are leaked and become attacker-observable.
Following the common constant-time leakage model, \(\protect\basesem{}\) leaks the \emph{program counter} of fetched and executed instructions and the \emph{memory addresses} accessed by (speculative) loads and stores.

\subsubsection{Semantics}
Our base semantics \(\protect\basesem{}\) models the execution of assembly programs at the microarchitectural level with a deterministic relation $\smicro{}{}{\mconf_{1}}{\mconf_{2}}{}$, meaning that a step in the configuration \(\mconf_{1}\) produces a unique configuration \(\mconf_{2}\).
It captures direct and indirect control flow speculation (Spectre-v1 and Spectre-v2), similar to existing microarchitectural semantics~\cite{daniel23prospect, guarnieri21hardwaresoftware}.
It is unknown how existing DMPs interact with other speculation mechanisms like memory disambiguation~\cite{horn18spectrev4}, load value~\cite{kim25flop}, or load address~\cite{kim25slap} prediction.

The semantics uses \emph{directives} to model the different stages of an out-of-order pipeline: $\fetch{}$, $\execute{i}$ for \(i \in \mathbb{N}\), and $\retire{}$.
At each step, the semantics selects the next directive \(\directive{}\) to execute using the function \(\directive{} = \nextdir(\micro)\) and applies the corresponding rule for that directive, modeled as an auxiliary relation \(\smicro{\directive{}}{}{\mconf_{1}}{\mconf_{2}}{}\).
\if\extendedVersion1
Semantics rules are standard, and we leave them to \cref{app:baseline-semantics}.
\else
Semantics rules are standard; we give an excerpt of the complete rules~\cite{sprengers27secdmpextended} in \cref{app:short-semantics}.
\fi

\subsection[DMP Extensions]{DMP Extensions (\(\protect{\dmpsem{}}\) and \(\protect{\dmpbasesem{}}\))}\label{sec:formalization-existing-dmps}

To model the behavior of DMPs, we introduce two functions, \(\train(\micro, \loc, \mem, \val{a})\) and \(\prefetch(\micro, \loc, \mem, \val{a})\), modeling training and prefetching operations.
Both functions take as input the current microarchitectural state \(\micro\), the location \(\loc\) of the instruction triggering the DMP, the address \(\val{a}\) loaded by the instruction, and the memory \(\mem\).
For generality, we give the DMP access to the whole memory: this allows us to model DMPs that perform recursive prefetches based on loaded values.
Both functions return a new microarchitectural state \(\micro'\) modeling the updates performed by the DMP.
In general, we leave these functions undefined: they can be instantiated to model different DMP implementations.
This also means that an attacker strategy \(\attacker\) includes \(\train\) and \(\prefetch\) functions.

\subsubsection[Modeling existing implementations]{Modeling existing implementations (\(\protect{\dmpsem{}}\))}\label{sec:apple-model}
In existing DMP implementations, the DMP is trained and activated on memory loads.
To model this behavior, we extend the \textsc{load} rule of \(\basesem{}\) with the \(\train\) and \(\prefetch\) functions, obtaining the \rulename{execute-load-dmp}{fig:execute-load-dmp} rule defined in \Cref{fig:execute-load-dmp}.
To show that our model captures existing DMP designs, we now present a high-level model of the Apple DMP and include models of the Intel DMP, CDP, and IMP in \cref{app:extraDMP}.

\begin{figure}[t]
    \centering
    \begin{mathpar}
        \updatedLoad{}
    \end{mathpar}
    \caption{Load evaluation for \(\protect\bigmicro{\dmp{}}{}{}{}{}\). Boxes mark differences with \(\protect\bigmicro{\base{}}{}{}{}{}\).}\label{fig:execute-load-dmp}
\end{figure}

\Cref{fig:sec-apple-dmp} contains the model of the Apple DMP~\cite{chen24gofetch} described in \cref{sec:dmp-models-apple} (the gray box will be explained later and can be ignored for now).
Since this DMP does not have a training phase, we only define the \(\prefetch\) function.
This function applies Apple's pointer heuristic to all words (\ie candidate pointers) in the loaded cache block (lines~2--7).
If the heuristic passes, it dereferences the candidate pointer (line~8).
To model the DMP's next-line prefetching behavior, the function also prefetches the cache block following the pointer's target in memory (line~8).
With this behavior, our model captures leakage through both the \emph{Arch} and \emph{Sub} classes.
It also captures leakage through the \emph{Trans} class, as the semantics allows \(\prefetch\) to be called on speculative loads.
However, since the specific mechanism responsible for the Apple DMP's \emph{Pref} leakage is unknown (as mentioned in \cref{sec:dmp-models-apple}), we do not capture this leakage class in our model.

\begin{figure}[t]
    \centering
    \secureAppleDmp
    \caption{Model of the Apple DMP where \(n\) represents the number of words in a cache block.}
    \label{fig:sec-apple-dmp}
\end{figure}

\subsubsection[Generalization]{Generalization (\(\protect{\dmpbasesem{}}\))}
More generally, given an arbitrary hardware semantics \absbase{}, we define a \dmp{}-enhanced semantics
\(\dmpbasesem{}\).
We assume \(\absbase{}\) is equipped with two functions \(\getpc(\mconf{}, \directive)\) and \(\loadaddr(\mconf{}, \directive)\), which, given the current configuration and directive, respectively return the current program counter and the address being loaded (\(\loadaddr\) returns \(\bot\) if no address is accessed).
We also add a predicate \(\dmpactivate(\micro{}, \loc, \directive)\), which returns whether the DMP is activated at the current location \(\loc\) and directive \(\directive\).
These functions are defined by the semantics, not attacker-controlled.

The semantics of \(\dmpbasesem\) is defined in \Cref{fig:dmpbase}.
Each rule follows the base semantics $\absbase{}$ to compute the next configuration.
Additionally, they update the microarchitectural state to reflect whether the DMP was activated.
If this is the case (\ie $\dmpactivate{}$ is true), we trigger the DMP using the $\train{}$ and $\prefetch{}$ functions.
To leave the order in which microarchitectural effects are combined unspecified, we pass the incoming microarchitectural state together with every computed state to \(\update\), which produces the final microarchitectural state.
Note that, compared to the base semantics, the \textsc{dmp-activate} rule only modifies the microarchitectural state; the architectural state and reorder buffer follow the base semantics.

\begin{figure}[t]
    \centering
    \begin{mathpar}
        \abstractDMPactivate{}
        \and
        \abstractDMPnoactivate{}
    \end{mathpar}
    \caption{Semantics of \(\protect\dmpbasesem{}\), with $\val{a}^{\bot} \in \Vals \cup \{\bot\}$.}\label{fig:dmpbase}
\end{figure}

\subsection[Formalization of SecDMP]{Formalization of SecDMP (\(\protect\secdmpsem{}\) and \(\protect\secdmpbasesem{}\))}\label{sec:secdmp-main}

\subsubsection{Memory partitioning} \label{sec:security-levels}
\secdmp{} assumes that memory is partitioned between public (low, \(\botsec\)) and secret (high, \(\topsec\)) data.
We model this assumption using a memory partition \(\memsec\) that maps memory addresses to their security level.
Similar to prior work~\cite{daniel23prospect}, we require this partition to be fixed (\ie it cannot change over time), meaning that any secret data written to the public memory region is \emph{declassified}.
Our semantics does not make a distinction between declassified values and leakage; both are simply leaked to the microarchitectural state via $\update$.
We capture this with the rule
\if\extendedVersion1
\rulename{retire-store-low}{fig:retire-rules} (\cf \Cref{fig:retire-rules} in \cref{app:baseline-semantics}),
\else
\rulename{retire-store-low}{fig:relevant-rules} (\cf \Cref{fig:relevant-rules} in \cref{app:short-semantics}),
\fi
which leaks the value \(\val{v}\) when storing it to a public address (\ie{} \(\memsec(\val{a}) = \botsec\)).
Note that declassification never occurs speculatively.

Given such a partition, we define the \emph{public projection} of a memory \(\mem\) as the memory \(\sproj{\mem}\) such that for all addresses \(a\):
\[
    \sproj{\mem}(a) = \begin{cases}
        \mem(a) & \text{if } \memsec(a) = \botsec \\
        \bot    & \text{otherwise}
    \end{cases}
\]

\subsubsection{Securing existing DMPs with \secdmp{}}\label{sec:secdmp-applied}
We first show how \secdmp{} can be applied to existing DMP designs by securing our model of the Apple DMP (\(\protect{\dmpsem{}}\)).
We modify the \(\prefetch\) function as shown by the gray box in \Cref{fig:sec-apple-dmp}: for any address \(\val{address}\) analyzed by the DMP, \secdmp{} first checks whether this address corresponds to public data (\(\memsec(\val{address}) = \botsec\)).
If so, we proceed as in the insecure model and prefetch the corresponding target address if the pointer heuristic passes.
Otherwise, we do not analyze the address as it contains secret data.
\cref{app:extraDMP} describes the secure models of the Intel DMP, CDP, and IMP.

\subsubsection{Generalization}\label{sec:secdmp-general}
More generally, \secdmp{} can secure any DMP by restricting its behavior ($\train$ and $\prefetch$) to public memory regions defined by $\sproj{\mem}$:
\begin{gather*}
    \micro_\text{train} \mydef \train(\micro, \loc, \rulebox{\sproj{\mem}}, \val{a}) \\
    \micro_\text{pref} \mydef \prefetch(\micro, \loc, \rulebox{\sproj{\mem}}, \val{a}) \\
    \micro' \mydef \update(\micro, \mathsf{activate}, \micro_\text{train},
      \micro_\text{pref}, \val{a})
\end{gather*}
\secdmp{} can therefore use the entire public memory for training and prefetching.

The individual \secdmp{} models for the Apple DMP, and those for the Intel DMP, CDP, and IMP in \cref{app:extraDMP}, are naturally subsumed by a more abstract maximally-leaking model where $\train$ and $\prefetch$ in \Cref{fig:execute-load-dmp} are kept abstract and restricted to the public memory projection $\sproj{\mem}$.
We call this model $\secdmpsem{}$.
Additionally, given an arbitrary hardware semantics $\absbase{}$, we define a $\secdmp{}{}$-enhanced semantics, $\secdmpbasesem{}$, as in \Cref{fig:dmpbase}, but with $\train$ and $\prefetch$ restricted to the public memory projection $\sproj{\mem}$.

%% file: semantics/language.tex
\begin{figure}[t]

  \begin{tabularx}{\linewidth}{r@{~}r@{~}X }
  \multicolumn{3}{l}{
    (Values) $\val{a}$, $\val{}, \loc{} \in \Vals \quad$%
    (Registers) $\register{x}, \register{r} \in \Regs \quad$%
  }\\
  \(e\)  & \(::=\) & \(\val{v}\) \(\vert\) \(\register{x}\) \(\vert\) \(e_1 \otimes e_2\) \\
  \(inst\) & \(::=\) &
  \(\mov{\register{r}}{e}\)  \(\vert\)
  \(\load{\register{r}}{e}\) \(\vert\)
  \(\store{e_a}{e_v}\)
  \(\vert\)
  \(\beqz{e}{\loc}\) \(\vert\)
  \(\jmp{e}\)
  \end{tabularx}

  \caption{Syntax of assembly instructions: \(\otimes\) is an abstract binary operator, \(\register{r} \neq \pc{}\). A program \(\protect\prog{}\) is a partial mapping from locations to instructions and \(\protect\prog{\loc{}}\) denotes the instruction at location \(\protect\loc{}\).
  }\label{fig:syntax}%
\end{figure}

%% file: sections/50-secdmp.tex
\section{Security Theorems and Proofs}\label{sec:proofs}
In this section, we prove that \secdmp{} instantiations do not leak more than what is already leaked (or declassified) by the baseline semantics they extend.
\if\extendedVersion0
The full formalization and proofs can be found in the extended version of our paper~\cite{sprengers27secdmpextended}.
\fi

\subsection{Definitions}

\subsubsection{Public equivalence}
We extend the memory partition to registers by defining a partition \(\parsec = (\memsec, \regsec)\), where \(\memsec\) maps memory addresses to their security level and \(\regsec\) maps registers to their security level.
Two architectural configurations $\aconf, \aconf'$ are \emph{public-equivalent} (or \(\parsec\)-equivalent) according to a partition \(\parsec\), written \(\aconf{} \loweq \aconf'\), if they agree on their program counter, public memory, and public registers.
\if\extendedVersion1
The formal definition is given in \cref{app:secdmp-proofs}.
\fi

\subsubsection{Microarchitectural security}
Given a hardware semantics \(\bigmicro{}{}{}{}{}\), a program \(\prog{}\), a security partition \(\parsec{}\), and two initial architectural configurations \(\aconf\), \(\aconf'\), we say that \(\prog{}\) leaks on \(\aconf, \aconf'\) if \(\aconf, \aconf'\) are \(\parsec\)-equivalent and there exists an attacker strategy \(\attacker\) that produces distinct observations from them.
Formally:
\begin{definition}
    \label{def:hwleak}
    Given a hardware semantics \(\bigmicro{}{}{}{}{}\), a program \(\prog{}\) leaks on \(\aconf, \aconf'\) under partition \(\parsec\), written
    \(\hwleak{\prog{}}{\parsec{}}{\aconf}{\aconf'}{\bigmicro{}{}{}{}{}}\), if:
    \begin{equation*}
        \aconf{} \loweq \aconf{'} \wedge
        \exists \attacker.\
        \bigmicro{}{\prog{}}{\aconf{}}{\attacker} \neq \bigmicro{}{\prog{}}{\aconf{}'}{\attacker}.
    \end{equation*}
\end{definition}

\noindent
Conversely, if no pair of \(\parsec\)-equivalent configurations leaks, we say that \(\prog{}\) is \emph{secure} for security partition \(\parsec{}\) on \(\bigmicro{}{}{}{}{}\).

\subsubsection{Comparing the security of hardware semantics}
In the following, we compare the security guarantees of different hardware semantics (\eg{} with and without DMPs).
To do so, we define a relation \(\moresec\) where \(\optsem{} \moresec \absbasesem{}\) means that \(\optsem{}\) does not leak any secret information that is not already leaked by \(\absbasesem{}\). Formally:

\begin{definition}[\(\protect\optsemnocolor{} \moresec{}  \protect\absbasesemnocolor{}\)]\label{def:security}
    A hardware semantics \(\optsem{}\) is at least as secure as a hardware semantics \(\absbasesem{}\), written \(\optsem{} \moresec \absbasesem{}\), if for every program \(\prog{}\), security partition \(\parsec{}\), and pair of initial architectural configurations \(\aconf,\aconf'\):
    \[
        \hwleak{\prog{}}{\parsec{}}{\aconf}{\aconf'}{\optsem{}}
        \implies
        \hwleak{\prog{}}{\parsec{}}{\aconf}{\aconf'}{\absbasesem{}}.
    \]
\end{definition}

For example, a sequential hardware semantics \(\bigmicro{\textsc{seq}}{}{}{}\) is (typically) more secure than a speculative hardware semantics \(\bigmicro{\textsc{spec}}{}{}{}\) (\(\bigmicro{\textsc{seq}}{}{}{} \moresec \bigmicro{\textsc{spec}}{}{}{}\)) since sequential execution does not leak through speculation.
Note that if \(\optsem{} \moresec \absbasesem{}\), then all programs that are secure on \(\absbasesem{}\) are also secure on \(\optsem{}\).

\subsection{Security Proof}

\subsubsection[Security of SecDMP]{Security of $\protect\secdmpsem{}$}
We start by showing that our speculative semantics equipped with \secdmp{} does not leak more than \(\basesem\).
Formally:

\begin{restatable}[]{theorem}{secdmpthm}\label{thm:secdmpsecurity}
    \(\protect\secdmpsem{} \moresec \protect\basesem{}\)
\end{restatable}

\begin{proof}
    \if\extendedVersion1
    We leave the full proof to \cref{app:secdmp-proofs} and sketch the main idea here.
    \else
    We leave the full proof to our extended version~\cite{sprengers27secdmpextended} and sketch the main idea here.
    \fi
    Fix two publicly equivalent initial configurations and assume an attacker \(\secdmp{\attacker}\) distinguishes them in \(\secdmpsem{}\); let \(n+1\) be the first step where the resulting traces diverge (\ie{} the microarchitectural states are different).
    Without loss of generality, we assume this attacker to be equipped with a maximally-leaking $\update$ function that leaks all its arguments to the microarchitectural state---\ie{} it represents the strongest possible leakage allowed by the semantics.\footnote{In particular, this attacker captures, in a non-destructive way, all leakage piled up by our semantics rules. This assumption does not restrict the generality of the proof since a divergence produced by any other $\update$ function is also captured by such a maximally-leaking $\update$ function.}
    We then construct a baseline attacker \(\base{\attacker}\) that distinguishes the same two configurations in \(\basesem{}\) at the same step by simulating the $\secdmpsem{}$ execution.
    By induction over the small step semantics of \(\secdmpsem\), we show that we can construct pairs of executions in \(\basesem{}\) that (1) match the architectural state and reorder buffer of executions in \(\secdmpsem{}\) at every step, and (2) ensure that whenever the \(\secdmpsem{}\) executions diverge, the baseline executions also diverge.
    Intuitively, because \(\secdmpsem{}\) restricts \(\train\) and \(\prefetch\) to the public projection \(\sproj{\mem}\), any divergence caused by \secdmp{} must arise from data already leaked in the baseline model (program counters, accessed addresses, and public memory); thus \(\basesem{}\) can always produce the same leakage as \(\secdmpsem{}\).
\end{proof}

\subsubsection{Instantiating SecDMP implementations}
It is easy to see that the \secdmp{} implementations proposed in \cref{sec:secdmp-applied} and \cref{app:extraDMP} are instances of $\secdmpsem{}$: they only use public memory for training and prefetching.
Therefore, it directly follows from \cref{thm:secdmpsecurity} that these \secdmp{} implementations are as secure as the baseline speculative semantics.

\subsection{Generalization}\label{sec:secdmp-combination-secure-speculation}
We generalize our results for an implementation of $\secdmp{}$ combined with an arbitrary baseline $\absbase{}$.
Specifically, we show that under natural assumptions on the base semantics \(\absbasesem{}\), the \secdmp{}-enhanced semantics \(\secdmpbasesem{}\) is at least as secure as the base semantics \(\absbasesem{}\).

\subsubsection[Assumptions over B]{Assumptions over $\protect\absbasesem{}$}
\if\extendedVersion1
We give a high-level summary and leave formal definitions to \cref{app:secdmp-proofs}.
\else
We give a high-level summary and leave formal definitions to the extended version~\cite{sprengers27secdmpextended}.
\fi
\begin{itemize}
\item[\if\extendedVersion1\cref{hyp:micro}\else\textbf{A1}\fi]
The microarchitectural state \(\micro\) is updated only through \(\update\), and never directly from the architectural state or reorder buffer.
\item[\if\extendedVersion1\cref{hyp:corr}\else\textbf{A2}\fi]
The semantics cannot encode information from \(\micro\) into the architectural state or reorder buffer without using \(\predict\) or \(\nextdir\).
\item[\if\extendedVersion1\cref{hyp:ct}\else\textbf{A3}\fi]
The base semantics respects a constant-time leakage model, exposing at least the program counter and accessed memory addresses, and it declassifies values written to public memory.
\end{itemize}
\if\extendedVersion1\cref{hyp:micro} \else\textbf{A1} \fi and \if\extendedVersion1\cref{hyp:corr} \else\textbf{A2} \fi ensure that
leakage and attacker-controlled speculation are modeled by dedicated interfaces in the semantics---\ie{} the attacker can neither observe nor influence execution outside the behavior exposed by these interfaces.

\subsubsection[Security of B+SecDMP]{Security of $\secdmpbasesem{}$}
\noindent
Under these assumptions, we can show that \(\secdmpbasesem{}\) is at least as secure as (or does not leak more than) the base speculative semantics \(\absbasesem{}\). Formally:

\begin{restatable}[]{theorem}{secdmpbasethm}\label{thm:secdmpbase-security}
    Given a base speculative semantics \(\absbasesem{}\) satisfying
    \if\extendedVersion1\cref{hyp:micro,hyp:ct,hyp:corr}\else\textbf{A1--3}\fi,
    $\secdmpbasesem{} \moresec \absbasesem{}$
\end{restatable}

\begin{proof}
    \if\extendedVersion1
    We leave the full proof to \cref{app:secdmp-proofs} and sketch the main idea here.
    \else
    We leave the full proof to the extended version~\cite{sprengers27secdmpextended} and sketch the main idea here.
    \fi
    The proof follows the same structure as that of \cref{thm:secdmpsecurity}.
    Fix two publicly equivalent initial configurations and assume an attacker \(\secdmpbase{\attacker}\) distinguishes them in \(\secdmpbasesem{}\).
    We then construct a baseline attacker \(\base{\attacker}\) that simulates the $\secdmpbasesem{}$ execution in \(\absbasesem{}\) and distinguishes the same two configurations.
    We use
    \if\extendedVersion1\cref{hyp:corr} \else\textbf{A2} \fi
    to match the architectural states and reorder buffers of the corresponding
    executions in \(\absbasesem{}\) and \(\secdmpbasesem{}\) by matching
    their predictions and directives.
    We use
    \if\extendedVersion1\cref{hyp:micro} \else\textbf{A1} \fi
    to ensure that all baseline microarchitectural changes pass through
    \(\update\), allowing us to reproduce any divergence originating in the underlying baseline step by choosing an update
    function for \(\absbasesem{}\) that produces different results in the two executions.
    Finally, \if\extendedVersion1\cref{hyp:ct} \else\textbf{A3} \fi
    ensures that public memories remain equal before the divergences and that differences in the program counters or loaded addresses supplied to \secdmp{} are already observable in the baseline.
    Thus, we can always construct an attacker for \(\absbasesem{}\) that distinguishes the same pair of initial configurations as \(\secdmpbasesem{}\).
\end{proof}

\subsubsection{Instantiation with sequential CT}
As a corollary of \cref{thm:secdmpbase-security}, we can instantiate \absbase{} with a sequential semantics in the constant-time leakage model, denoted \(\Sem{\basecolor{\textsc{SeqCT}}}\).
This semantics corresponds to an ISA semantics that executes instructions sequentially without speculation while leaking only the program counter, accessed memory addresses, and declassifies values written to public memory.
It is easy to see that \(\Sem{\basecolor{\textsc{SeqCT}}}\) satisfies
\if\extendedVersion1\cref{hyp:micro,hyp:ct,hyp:corr}\else\textbf{A1--3}\fi.
Therefore, we can conclude that \(\Sem{\textsc{\basecolor{SeqCT}+\secdmpname}}\) is at least as secure as \(\Sem{\basecolor{\textsc{SeqCT}}}\).

\subsubsection{Instantiation with \prospect{}}
As a final step, we combine \secdmp{} with \prospect~\cite{daniel23prospect}, a state-of-the-art speculative-constant-time defense, to obtain a system that provides strong security guarantees against both speculative execution attacks and DMP attacks for constant-time programs.
\prospect{} satisfies
\if\extendedVersion1\cref{hyp:micro,hyp:ct,hyp:corr}\else\textbf{A1--3}\fi,
therefore, by \cref{thm:secdmpbase-security}, we can conclude that \(\Sem{\textsc{\basecolor{\prospect{}}+\secdmpname}}\) does not leak more than \(\Sem{\textsc{\basecolor{\prospect{}}}}\).
Furthermore, from~\cite{daniel23prospect}, we know that \(\Sem{\basecolor{\textsc{\prospect{}}}}\) is at least as secure as a baseline sequential semantics in the constant-time leakage model (\ie{} \(\Sem{\basecolor{\textsc{\prospect{}}}} \moresec \Sem{\basecolor{\textsc{SeqCT}}}\)).
Combining these, we obtain that \(\Sem{\textsc{\basecolor{\prospect{}}+\secdmpname}} \moresec \Sem{\basecolor{\textsc{SeqCT}}}\).

In conclusion, \(\Sem{\textsc{\basecolor{\prospect{}}+\secdmpname}}\) provides strong end-to-end security guarantees against both speculative execution and DMP attacks for constant-time programs.

%% file: sections/60-implementation-evaluation.tex
\section{Implementation and Evaluation} \label{sec:implementation-evaluation}
In this section, we empirically evaluate the security, performance, and hardware cost of \secdmp{}, permanent disabling, and selective disabling.
We also consider the combination of \secdmp{} and \prospect{}~\cite{daniel23prospect} to evaluate the cost of integrating \secdmp{} into existing Spectre defenses that rely on secret partitioning, protecting constant-time code in the presence of both speculative execution and DMPs.

\subsection{Implementation} \label{sec:implementation}
Our implementation extends the Proteus \riscv processor~\cite{bognar23proteus}, on which \prospect{}~\cite{daniel23prospect} was also implemented.
Proteus is a configurable research processor that includes an out-of-order pipeline and advanced microarchitectural features such as caches, speculative execution, branch prediction, and a prefetching module.
We also build on the Proteus ecosystem~\cite{ecosystemgithub}, containing benchmarks and tools for evaluation.

\subsubsection{Proteus+DMP}
First, we implement the four DMPs described in \cref{sec:dmp-models}, creating, to the best of our knowledge, the first open-source DMP implementations.
These will act as the four baselines for our evaluation in \cref{sec:evaluation}; we will repeat the evaluation for each baseline and refer to them as Proteus+DMP.

We model the DMPs as faithfully as possible after the designs described in the literature, resulting in the leakage classes from \cref{tab:leakage-classification} (including all \emph{Trans} leakage).
The only exception is the Apple DMP: since the mechanism responsible for its \emph{Pref} leakage is unknown, we do not model it.
As a consequence, our models of the Apple DMP and CDP are very similar, the main difference being that CDP is recursive and the Apple DMP is not.
This enables sharing most of their implementation code.
For IMP, we implement the default 2-level variant \val{A[B[i]]}.
Since CDP and IMP feature several configurable parameters, we initialize them with the recommended values from the original papers~\cite{cooksey02dmp,yu15imp}, \eg a maximum recursion depth of 3 for CDP and a prefetch lookup table with 16 entries for IMP.

\subsubsection{Permanent and selective disabling}
Next, we implement prototypes of the defenses on Proteus+DMP.
To model permanent disabling, we use the original version of Proteus with no DMP.
For selective disabling, we add a control and status register (CSR) to toggle the DMP from software.
To prevent instructions from executing speculatively with incorrect CSR values, each CSR write in software must be followed by a fence instruction.

\subsubsection{SecDMP}
Our \secdmp{} implementation partitions the memory into public and secret regions, similar to \prospect{}~\cite{daniel23prospect}, and prevents the DMP from scanning data marked as secret.
The secret regions are determined by CSRs (four CSRs defining two regions in our implementation).
To match the static partitioning of \secdmp{}'s security proof, the CSRs must be configured at program start and remain unchanged throughout the program.
Notably, \secdmp{} requires only the DMP to distinguish public and secret data, requiring no secret tracking in the processor pipeline.

\subsubsection{\prospect{}+SecDMP}
For the combination of \prospect{} and \secdmp{}, we add the \secdmp{} extension to the upstream \prospect{} implementation.

\subsection{Evaluation}\label{sec:evaluation}
We run the evaluation on each prototype for all four DMPs,
configuring the Proteus \riscv processor with a 4-way 16-set cache, a 64-cycle miss penalty, a 32-entry reorder buffer, 8 arithmetic units, 3 load units, and 2 multiplication/division units.
The experiments are run in Verilator~\cite{verilator} simulations.
Numerical results are available in \extendedElse{\cref{app:evaluation-results}}{the extended paper~\cite{sprengers27secdmpextended}} and in \artifactrepo{}.

\subsubsection{Security}\label{sec:evaluation:sec}
\begin{table}[t]
    \centering
    \caption{Security evaluation results: (1)~IMP Leak, (2)~Augury PoC, (3)~GoFetch PoC, (4)~Peek-a-Walk, (5)~DMP-Spectre, (6)~Recursion, (7)~Load Secret, (8)~Same Cache Block, (9)~Transient Load, (A)~PHT-Load, (B)~PHT-Store, (C)~PHT-Branch, (D)~PHT-Jump, (E)~PHT-Division. Leakage on Apple DMP~(\protect\dmpleak{1}{0}{0}{0}), Intel DMP~(\protect\dmpleak{0}{1}{0}{0}), CDP~(\protect\dmpleak{0}{0}{1}{0}), IMP~(\protect\dmpleak{0}{0}{0}{1}), or none~(\cmarkg).}
    \label{tab:security-evaluation}
    \setlength{\tabcolsep}{3.25pt}
    \begin{tabular}{l c c c c c c c c c c c c c c}\toprule
        \multirow{2}{*}{Defense} & \multicolumn{9}{c}{DMP} & \multicolumn{5}{c}{Spectre}                                                                                                                                                                                                                                    \\ \cmidrule(lr){2-10} \cmidrule(lr){11-15}
                                 & 1                       & 2                           & 3                    & 4                    & 5                    & 6                    & 7             & 8                    & 9             & A             & B             & C             & D             & E             \\ \midrule
        Proteus+DMP              & \dmpleak{0}{0}{0}{1}    & \dmpleak{0}{1}{0}{0}        & \dmpleak{1}{0}{1}{0} & \dmpleak{0}{1}{0}{0} & \dmpleak{1}{0}{1}{0} & \dmpleak{0}{0}{1}{0} & \dmpleakall{} & \dmpleak{1}{0}{1}{0} & \dmpleakall{} & \dmpleakall{} & \dmpleakall{} & \dmpleakall{} & \dmpleakall{} & \dmpleakall{} \\
        Permanent Disabling      & \cmarkg                 & \cmarkg                     & \cmarkg              & \cmarkg              & \cmarkg              & \cmarkg              & \cmarkg       & \cmarkg              & \cmarkg       & \dmpleakall{} & \dmpleakall{} & \dmpleakall{} & \dmpleakall{} & \dmpleakall{} \\
        Selective Disabling      & \dmpleak{0}{0}{0}{1}    & \dmpleak{0}{1}{0}{0}        & \cmarkg              & \dmpleak{0}{1}{0}{0} & \dmpleak{1}{0}{1}{0} & \dmpleak{0}{0}{1}{0} & \cmarkg       & \dmpleak{1}{0}{1}{0} & \dmpleakall{} & \dmpleakall{} & \dmpleakall{} & \dmpleakall{} & \dmpleakall{} & \dmpleakall{} \\
        \secdmp{}                & \cmarkg                 & \cmarkg                     & \cmarkg              & \cmarkg              & \cmarkg              & \cmarkg              & \cmarkg       & \cmarkg              & \cmarkg       & \dmpleakall{} & \dmpleakall{} & \dmpleakall{} & \dmpleakall{} & \dmpleakall{} \\
        \prospect                & \dmpleak{0}{0}{0}{1}    & \dmpleak{0}{1}{0}{0}        & \dmpleak{1}{0}{1}{0} & \dmpleak{0}{1}{0}{0} & \dmpleak{1}{0}{1}{0} & \dmpleak{0}{0}{1}{0} & \dmpleakall{} & \dmpleak{1}{0}{1}{0} & \dmpleakall{} & \dmpleakall{} & \dmpleakall{} & \dmpleakall{} & \dmpleakall{} & \dmpleakall{} \\
        \prospect{}+\secdmp{}    & \cmarkg                 & \cmarkg                     & \cmarkg              & \cmarkg              & \cmarkg              & \cmarkg              & \cmarkg       & \cmarkg              & \cmarkg       & \cmarkg       & \cmarkg       & \cmarkg       & \cmarkg       & \cmarkg       \\ \bottomrule
    \end{tabular}
\end{table}
We evaluate DMP security by running the leakage tests and attacks from \cref{tab:programs-classification}.
Since \prospect{}'s original implementation has been evaluated against Spectre-PHT~\cite{daniel23prospect}, we additionally run five Spectre-PHT tests, each with a different leakage source (\ie load, store, branch, jump, and division).
For selective disabling, we modify the programs to disable DMPs before each security-critical code region, and re-enable them at the end of each region.
For \secdmp{} and \prospect{}, we annotate all secrets, which are then moved to a separate memory region by the compiler and linker.
We write the region's boundaries to the CSRs at program start.

Our security evaluation methodology is based on non-interference testing, similar to prior work~\cite{daniel23prospect,winderix24libra,winderix24ami}.
We craft two variants of each program, \(\prog{}_1, \prog{}_2\), that only differ in their secret values, and run them on each prototype.
If \(\prog{}_1, \prog{}_2\) exhibit a difference in any observable signal, we consider the secret as leaked.
For the DMP attacks (1--6), we monitor all signals observable to an attacker, \eg{} memory bus addresses, program counter, tag bits of instruction and data caches, and branch predictor state.
For the three general DMP leakage tests and the Spectre tests, we additionally monitor internal DMP signals, considering data as leaked once the DMP has analyzed it.
This keeps the leakage tests independent of any specific DMP's activation criteria.
Our results in \cref{tab:security-evaluation} confirm that:
\begin{itemize}
    \item Our DMP tests leak secrets exclusively through the DMPs: each program leaks on at least one Proteus+DMP baseline,\footnote{All DMP leakage tests and attacks leak through their targeted DMPs. While the Augury PoC was originally demonstrated on the Apple DMP~\cite{vicarte22augury}, it only leaks on the Intel DMP in our evaluation because we do not model the Apple DMP's unknown \emph{Pref} leakage mechanism.} while permanent disabling prevents all leakage.
    \item Selective disabling only prevents architectural leakage: it is secure for the \underline{GoFetch PoC} and \underline{Load Secret} programs, which only leak through the \emph{Arch} class, and is insecure for programs that leak through other classes.
    \item \secdmp{} correctly mitigates all DMP secret leakage.
    \item \prospect{} alone leaks through all our programs: while it prevents Spectre-PHT leakage, each Spectre test still architecturally loads a secret before the speculative access, which leaks the secret through the DMPs.
    \item \prospect{}+\secdmp{} prevents all secret leakage from both Spectre-PHT and DMPs.
\end{itemize}

\subsubsection{Performance}\label{sec:evaluation:performance}
A meaningful benchmark for evaluating the performance overhead of the defenses should
\begin{enumerate*}[label=(\arabic*)]
    \item compute on secret data and
    \item see a significant performance increase from DMPs.
\end{enumerate*}
The results then show for each defense whether deploying it to strengthen security preserves the performance gained from data-dependent prefetching.
For this purpose, we create \emph{DMP-Bench}, inspired by SpectreGuard~\cite{fustos19spectreguard}, which combines cryptographic and general-purpose benchmarks.

\paragraph{Cryptographic benchmarks}
\begin{figure}[!t]
    \centering
    \includegraphics[width=0.9\linewidth]{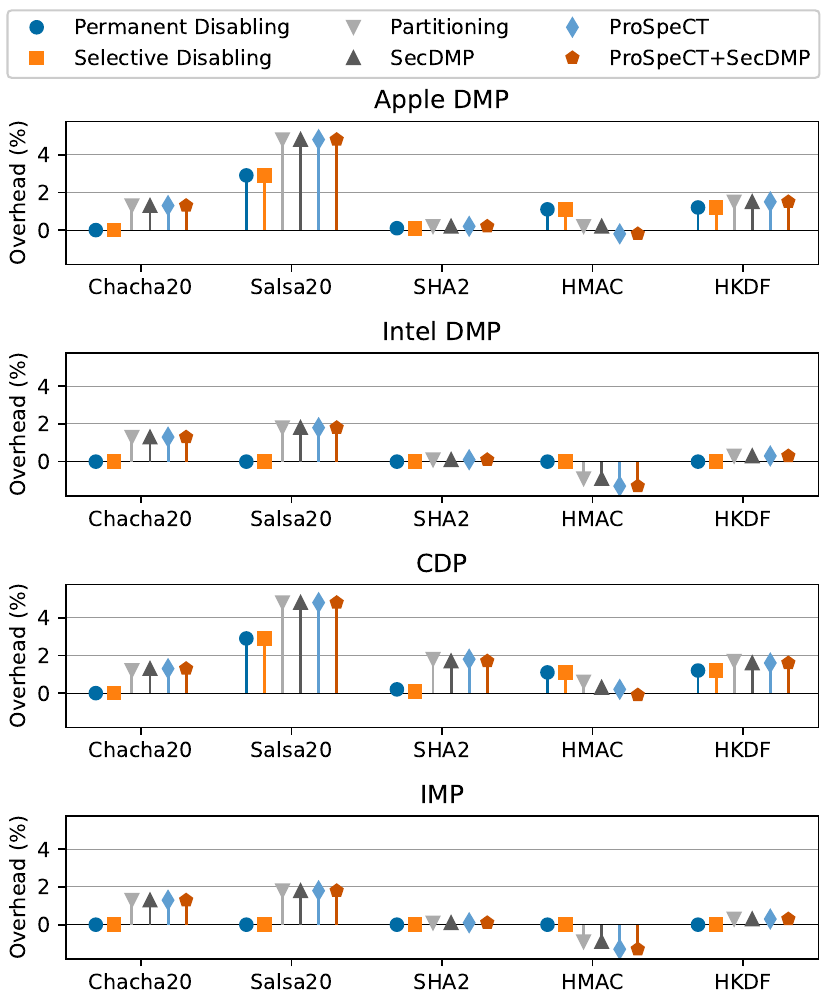}
    \caption{HACL* performance overhead per DMP, relative to Proteus+DMP.}
    \label{fig:hacl_results}
\end{figure}
First, we evaluate the performance of several primitives from the HACL* cryptographic library~\cite{zinzindohoue17hacl}.
As \secdmp{} and \prospect{} require secret partitioning, we manually annotate secrets as in \cref{sec:evaluation:sec}.
We partition five cryptographic primitives: Chacha20, Salsa20, SHA2 (256-bit variant), HMAC, and HKDF.\footnote{We verify the partition correctness through a Proteus plugin. For Salsa20, we observe unnecessary secret spills introduced by the compiler, which we remove by patching the assembly.}
In \cref{sec:practicality-partitioning}, we discuss how future work can overcome the partitioning challenges that arise for more complex programs.
To isolate the performance overhead of the partitioning itself, we evaluate the partitioned programs on the baseline hardware configurations without any defense enabled.
For selective disabling, we disable DMPs before calling each cryptographic function and re-enable them upon return.

The results in \Cref{fig:hacl_results} show that the defenses incur only a small overhead on the cryptographic benchmarks.
The highest overhead of permanent disabling is only \SI{2.9}{\percent} (on Salsa20 for both the Apple DMP and CDP), indicating that these benchmarks barely benefit from DMPs in the first place.
Selective disabling has the same overhead as permanent disabling, since it cannot re-enable DMPs during cryptographic functions.
Additionally, the cost of secret partitioning ranges from \SI{-0.9}{\percent} to \SI{+4.8}{\percent}, with an average of \SI{+1.15}{\percent} over all DMPs and benchmarks, and varies by DMP: the Apple DMP and CDP perform worse on the partitioned programs.
Finally, the overheads of \secdmp{} and \prospect{} stem solely from the partitioning method used; they add no additional cost, as \secdmp{} only blocks secret-dependent prefetches and \prospect{} only stalls speculative execution on secret data.

\paragraph{Work functions}
\begin{figure}[!t]
    \centering
    \includegraphics[width=0.95\linewidth]{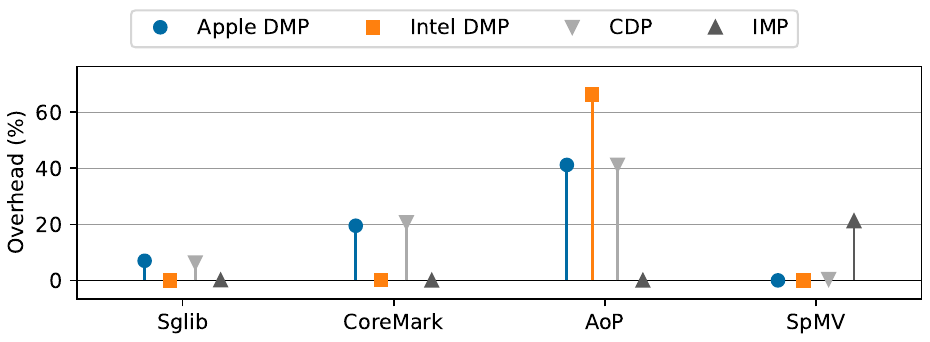}
    \caption{Permanent disabling performance overhead on the different work benchmarks per DMP, relative to Proteus+DMP.}
    \label{fig:work_results_disabling}
\end{figure}
The second component of DMP-Bench consists of public workloads, where no secrets are annotated for \secdmp{} and \prospect{}, and selective disabling never disables the DMP.
In the following, we focus only on the results for permanent disabling, summarized in \Cref{fig:work_results_disabling}, because the other defenses do not affect performance for public workloads (see the full results \extendedElse{in \cref{app:work-results}}{in the extended version~\cite{sprengers27secdmpextended}}).

First, we evaluate the performance of the Embench benchmark suite~\cite{embenchiot} and CoreMark~\cite{coremark}.
From Embench, only \emph{sglib-combined} (Sglib in \Cref{fig:work_results_disabling}) benefits from our DMPs: permanent disabling incurs a performance overhead of \SI{7.0}{\percent} for the Apple DMP and \SI{6.3}{\percent} for CDP.
The other Embench benchmarks gain no performance from the DMPs (\extendedElse{\cf{} \cref{app:work-results}}{\cf{}~\cite{sprengers27secdmpextended}}).
For CoreMark, we observe a more significant performance impact: permanent disabling has an overhead of \SI{20.5}{\percent} for the Apple DMP and \SI{20.7}{\percent} for CDP.
Neither \emph{sglib-combined} nor CoreMark triggers the Intel DMP or IMP, as neither contains their required activation patterns.
We therefore evaluate two additional workloads.
For IMP, we take the SpMV benchmark from HPCG~\cite{dongarra16hpcg}, which was evaluated in the IMP paper~\cite{yu15imp}, and port it to Proteus.
For the Intel DMP, we create a custom benchmark, AoP, that contains the array-of-pointers activation pattern.
\Cref{fig:work_results_disabling} shows that both gain a significant performance increase from their respective DMPs: permanent disabling incurs an overhead of \SI{21.2}{\percent} on SpMV for IMP, and an overhead of \SI{66.4}{\percent} on AoP for the Intel DMP and \SI{41.2}{\percent} for both the Apple DMP and CDP.

\paragraph{DMP-Bench}
\begin{figure*}[!t]
    \centering
    \includegraphics[width=\linewidth]{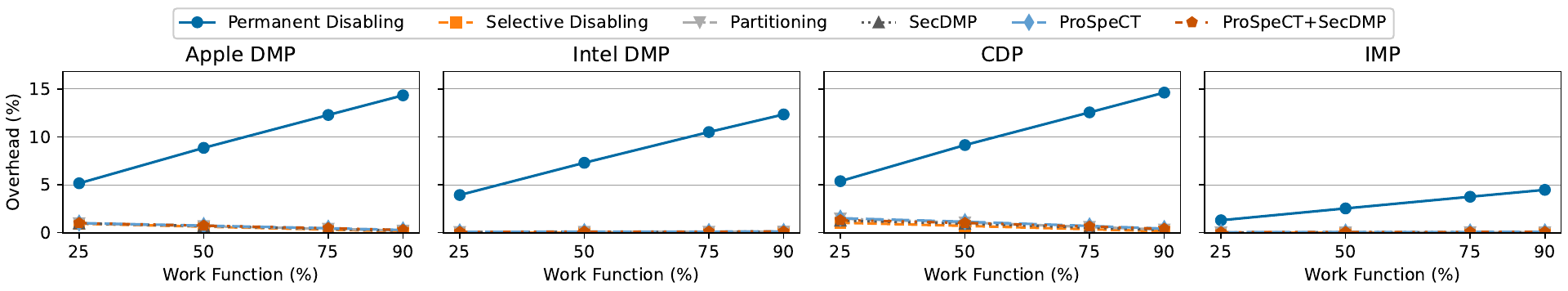}
    \caption{Performance overhead on DMP-Bench, geometric mean over all work/crypto combinations per DMP, relative to Proteus+DMP.}
    \label{fig:dmp_bench_results}
\end{figure*}
Finally, we combine each work benchmark with each cryptographic primitive to obtain \emph{DMP-Bench}, resulting in 20 work/crypto combinations.
We run each combination in four configurations, varying the fraction of time spent in the work function at \SI{25}{\percent}, \SI{50}{\percent}, \SI{75}{\percent}, and \SI{90}{\percent}, for a total of 80 programs evaluated on each defense prototype for each DMP.
The results are shown in \Cref{fig:dmp_bench_results}, reporting the geometric mean over all work/crypto combinations.
The full results are available in \extendedElse{\cref{tab:dmp-bench-AppleDmp,tab:dmp-bench-IntelDmp,tab:dmp-bench-Cdp,tab:dmp-bench-Imp}, in \cref{app:dmp-bench-results}}{the extended version~\cite{sprengers27secdmpextended}}.

Permanent disabling incurs a smaller overhead when most of the time is spent encrypting; this overhead increases as the work function's share grows, with a maximum of \SI{59.9}{\percent} overhead at \SI{90}{\percent} work on AoP-Salsa20 for the Intel DMP.
Neither selective disabling nor \secdmp{} impacts DMP performance for the public work functions, solving the performance drawbacks of permanent disabling.
Selective disabling has an average overhead of \SI{0.26}{\percent} over all DMPs and benchmarks.
Its overhead is small at \SI{25}{\percent} work (at most \SI{2.0}{\percent} on AoP-Salsa20 for CDP) and decreases to near zero at higher work percentages, as it only disables DMPs during the crypto functions.
\secdmp{}'s average overhead is \SI{0.36}{\percent}; the overhead is highest at \SI{25}{\percent} work (at most \SI{3.3}{\percent} on CoreMark-Salsa20 for CDP) and also approaches zero as the work percentage increases.
Again, the overhead stems solely from the secret partitioning method used; \secdmp{} itself adds no overhead as it does not hinder prefetching on public data.
Finally, \prospect{} adds at most \SI{0.5}{\percent} over the partitioning cost.

\subsubsection{Hardware cost}\label{sec:evaluation:hardware}
We use EVAL\nobreakdashes-HD~\cite{demeulemeester25hardwarecost} to measure the critical path and area overhead of the original version of Proteus, Proteus extended with the DMPs, and the hardware defenses.
The critical path remains constant at \SI{12.43}{\nano\second} across all prototypes.
The results in \cref{tab:hardware-evaluation} show that:
\begin{itemize}
    \item IMP introduces a substantial area increase of \SI{29.7}{\percent} caused by its large prefetch lookup table of 16 entries, while the Apple DMP, Intel DMP, and CDP introduce a \SI{2.0}{\percent}, \SI{5.1}{\percent}, and \SI{2.7}{\percent} area increase, respectively.
    \item Selective disabling and \secdmp{} introduce a negligible area overhead of at most \SI{0.2}{\percent} and \SI{0.5}{\percent}, respectively.
    \item \prospect{} alone has a small area overhead (highest is \SI{2.4}{\percent} for the Apple DMP). Extending \prospect{} with \secdmp{} increases area by at most \SI{0.3}{\percent}, since \prospect{} itself already relies on secret partitioning.
\end{itemize}

\input{tables/hardware_table.tex}

%% file: tables/hardware_table.tex
\begin{table}[t]
  \centering
  \caption{Hardware area cost evaluation results.}
  \label{tab:hardware-evaluation}
  \begin{tabular}{l c c c c}\toprule
    \multirow{2}{*}{Defense} & \multicolumn{4}{c}{Area (\si{\milli\meter\squared})} \\
  \cmidrule(lr){2-5}
     & Apple DMP & Intel DMP & CDP & IMP \\ \midrule
    Proteus & 0.6723 & 0.6723 & 0.6723 & 0.6723 \\
    Proteus+DMP & 0.6858 & 0.7064 & 0.6906 & 0.8717 \\
    Selective Disabling & 0.6861 & 0.7079 & 0.6910 & 0.8735 \\
    \secdmp{} & 0.6880 & 0.7102 & 0.6929 & 0.8758 \\
    \prospect{} & 0.7022 & 0.7228 & 0.7070 & 0.8881 \\
    \prospect{}+\secdmp{} & 0.7031 & 0.7253 & 0.7080 & 0.8908 \\
    \bottomrule
  \end{tabular}
\end{table}

%% file: sections/70-discussion-related-future.tex
\section{Discussion and Future Work} \label{sec:discussion}

\subsection{Practicality of Secret Partitioning}\label{sec:practicality-partitioning}
Our \secdmp{} implementation uses a coarse-grained static memory partition, which comes with two main challenges.
Importantly, these challenges stem from the concrete partitioning method used, not from \secdmp{}'s design itself.

The first challenge is handling dynamic allocation of secret data on the heap, which can be solved by splitting the heap into a public and a secret region and differentiating between them at allocation time (either inferred from annotations or by providing a secure heap allocation function)~\cite{schwarz20context,song25securing}.

The second challenge is partitioning the stack.
Coarse-grained partitioning cannot precisely separate public and secret data on the stack, often requiring the entire stack to be marked as secret~\cite{song25securing}.
In \cref{sec:evaluation}, we avoid stack-resident secrets entirely by relocating them to a global memory region.
While this approach works for simple programs, more complex programs risk the compiler spilling secret registers on the stack~\cite{daniel23prospect,schwarz20context,song25securing}.
In addition, it may not be trivial to modify programs to change stack variables into global variables while maintaining functional correctness~\cite{song25securing}.
To address these issues, Song \etal recently proposed SecSep~\cite{song25securing}, an assembly type system that tracks secrets and automatically rewrites programs to partition secret and public data on the stack.
While initially designed to support secure speculation defenses such as \prospect{}~\cite{daniel23prospect}, this technique directly applies to \secdmp{} as well.
Future work could integrate SecSep into compilers, potentially building on existing compiler-based secret-tracking mechanisms~\cite{philippe25gnuzero,zhang25ctllvm,soares23partialcfl}.

\subsection{Marking Secrets}
Similar to other microarchitectural defenses~\cite{yu19oisa,daniel23prospect,fustos19spectreguard,schwarz20context}, we achieve a good security/performance trade-off by requiring software to inform hardware about where secrets are located in memory.
Hence, our work can be seen as additional evidence that extending the hardware-software interface with mechanisms to mark secrets in memory is a powerful enabler for secure and efficient defenses.
Furthermore, we believe that our design could be extended to secure other optimizations in the memory subsystem.
For instance, silent store suppression~\cite{lepak00silentstores}, which leaks whether a written value matches the one already present in memory~\cite{vicarte21pandora}, could be secured by never suppressing writes to addresses holding secret data.

\subsection{Relation to \prospect}

\prospect is a Spectre countermeasure with a formal proof of security~\cite{daniel23prospect}.
Our design and formalization are inspired by \prospect: both mitigations rely on secrecy tracking.
However, while \prospect guarantees security for constant-time programs in the presence of speculation, it fails to protect constant-time programs against leakage from DMPs.

\secdmp{} is motivated by this gap: we show that leakage induced by DMPs can be generally mitigated by reusing the same secrecy-tracking mechanism.
Our formal proof goes beyond that of \prospect{}: we not only show that the combination of \secdmp{} and \prospect{} mitigates both leakage channels, but also that \secdmp{} eliminates all secret leakage through DMPs for an arbitrary baseline.

Our implementation shares the secret-tracking mechanism used by \prospect{} and reuses the existing implementation for the combined evaluation of \prospect{}+\secdmp{}.
An important difference is that, contrary to \prospect{}, which needs precise taint tracking through the pipeline, \secdmp{} only requires taint tracking in memory and the DMP, making it more realistic for complex implementations.

\subsection{Reverse-Engineering of Commercial DMPs}
The concrete DMP models in \cref{sec:formalization-existing-dmps} are based on prior reverse-engineering efforts and are only as precise as these efforts.
For instance, the specific mechanisms responsible for the Apple DMP's \emph{Pref} leakage reported by GoFetch~\cite{chen24gofetch} are unknown, so our model does not capture this leakage class.
Similarly, it is unclear whether the Intel DMP leaks data during speculation (\emph{Trans} class) or only architecturally.
Importantly, our security proofs hold even if the current models are imprecise, as \secdmp{}'s guarantees are established over an abstract DMP model that we believe can accurately capture any commercial design.
Leveraging model-based testing~\cite{oleksenko22revizor} could be a promising approach to refining our models.

Another direction for future work is to reverse-engineer the commercial Arm DMP and analyze its security implications in terms of our leakage classes.
This is of particular interest because Arm's security report~\cite{arm25dmp} indicates that the DMP is recursive, potentially enabling the recursive attack from \cref{sec:new-dmp-attacks} on real hardware.

\subsection{Secure Selective Disabling}
As discussed in \cref{sec:selective-disabling-leakage-classification}, it is fundamentally difficult to predict when public code regions may trigger DMPs into scanning secrets, especially for the \emph{Pref} and \emph{Trans} classes.
Consequently, secure selective disabling is only possible if all secrets are removed from memory before re-enabling DMPs.
Concurrent work by Häuser \etal~\cite{häuser25secrettypes} makes a similar observation and, as a solution, suggests unmapping all secret data when exiting security-critical code regions.
While this approach secures selective disabling across all leakage classes by making secrets inaccessible to DMPs, the authors acknowledge that it can introduce significant performance overhead.

%% file: sections/80-conclusion.tex
\section{Conclusion}
\secdmp{} is a powerful and provably secure defense against side-channel leakage of DMPs.
By requiring secrets to be tagged in memory, \secdmp{} ensures that no secrets leak through the DMP---without compromising the performance gained from prefetching based on public data.
We propose a leakage classification to categorize DMP leakage and show that \secdmp{} is the only defense that mitigates all leakage classes across all DMP designs while introducing minimal performance overhead.
Additionally, our open-source DMP implementations enable empirical validation of our results and can be used for future research on DMP security.

%% file: sections/100-appendix.tex
\if\extendedVersion0
    \input{sections/104-short-semantics.tex}
\fi

\section{Additional DMP Formalizations} \label{app:extraDMP}

\subsection{Formal Model of the Intel DMP}
\Cref{fig:sec-intel-dmp} contains our Intel DMP model, which is informed by the prior reverse-engineering work~\cite{wang25peekawalk} described in \cref{sec:dmp-models-intel}.
As explained there, we assume the Intel DMP can be trained and activated on speculative loads (\ie the \emph{Trans} class), which our speculative semantics captures by default.
The gray boxes in \Cref{fig:sec-intel-dmp} depict the \secdmp{} defense, and are ignored for the insecure model.
\begin{figure}[t]
    \centering
    \secureIntelDmp
    \caption{Model of the Intel DMP. \(\getLoadHistory{}\) extracts the history of load addresses from \(\micro\) that was appended with \(\appendToLoadHistory{}\). \(\predictStride{}\) returns the stride learned during the training phase, or zero if no activation pattern is detected.}
    \label{fig:sec-intel-dmp}
\end{figure}

\subsubsection{Insecure model}
The \(\train\) function in \Cref{fig:sec-intel-dmp} models the Intel DMP's training phase, which looks for an array-of-pointers dereferencing pattern.
For each load, it remembers the address by appending it to \(\micro\) (line~6) and later checks if subsequent loads dereference this address's value (lines~2--5).
The training phase thus checks whether an architecturally loaded value equals a later load address, leaking through the \emph{Arch} class.
Since the pattern the DMP looks for occurs in a loop, it additionally remembers the location of the instruction that triggered the initial load, and internally matches this with the corresponding addresses to derive the stride value.

The \(\prefetch\) function represents the activation of the DMP after a stride has been learned (lines~9--10).
Using the stride, it calculates an address ahead in the pattern (line~11), performs the check that limits the maximum stride value (line~12), loads the candidate pointer at this address (line~14), applies the pointer heuristic (line~15), and finally dereferences the candidate pointer if the heuristic passes (line~16).
It thus loads and subsequently leaks data located ahead in the pattern, \ie through the \emph{Pref} class.

Intel's pointer heuristic (lines~18--22) succeeds only if the upper 17 bits of the candidate pointer are all 0, restricting prefetching to user memory~\cite{wang25peekawalk}.
Additionally, if linear address masking (LAM)\footnote{LAM is a recently added feature allowing software to embed metadata in the highest-order bits of virtual addresses.} is enabled, the heuristic only checks if the bits at indices 63 and 47 are 0.

\subsubsection{Applying SecDMP}
To secure the Intel DMP with \secdmp{}, we modify it as depicted by the gray boxes in \Cref{fig:sec-intel-dmp}.
We modify both the \(\train\) and \(\prefetch\) functions, ensuring the DMP neither trains on nor dereferences secrets.

\subsection{Formal Model of CDP}
\Cref{fig:sec-cdp} contains the model of CDP~\cite{cooksey02dmp} described in \cref{sec:dmp-models-cdp}.
The gray box depicts the \secdmp{} defense, which is ignored for the insecure model.

\subsubsection{Insecure model}
Since CDP has no training phase, we only require the \(\prefetch\) function.
To model CDP's recursive behavior, we define \(\prefetchrec\), which takes the current recursion depth \(depth\) as an argument.
If \(depth\) exceeds CDP's configurable maximum \(maxRecursionDepth\), the function returns (lines~4--5).
Otherwise, it applies CDP's pointer heuristic to every word in the loaded cache block (lines~6--11).
If the heuristic passes, it prefetches the target and the next \(nextLinePrefetches\) adjacent cache blocks, thereby leaking the target (lines~12--13).
Because every word in the loaded cache block is scanned, the model captures CDP's \emph{Arch}, \emph{Sub}, and \emph{Trans} leakage (the latter because, in our semantics, speculative loads trigger prefetchers by default).
Finally, the function calls \(\prefetchrec\) recursively on each prefetched target (line~14), capturing the \emph{Pref} leakage class.

\begin{figure}[t]
    \centering
    \secureCdp
    \caption{Model of CDP where \(n\) represents the number of words in a cache block. CDP has a configurable \(maxRecursionDepth\), number of \(nextLinePrefetches\) and number of \(compareBits\) for the pointer heuristic.}
    \label{fig:sec-cdp}
\end{figure}

\subsubsection{Applying SecDMP}
To secure CDP with \secdmp{}, we modify the model to never scan secret data (gray box in \Cref{fig:sec-cdp}): at line~9, it checks whether the address holds public data and only proceeds with the prefetch if it does.

\subsection{Formal Model of IMP}
\Cref{fig:sec-imp} depicts the model of IMP (\cf \cref{sec:dmp-models-imp}) as described in \cite{yu15imp}.
The gray boxes depict the \secdmp{} defense, and are ignored for the insecure model.

\subsubsection{Insecure model}
The \(\train\) function models IMP's training phase, which looks for the indirect array access pattern: indexing one array (\ie the indirect array) using the values of another array (\ie the index array).
First, IMP looks for a regular stride access pattern that loads the index values of the index array.
For this purpose, line~2 stores the current program counter and address.

Next, lines~3--9 detect the indirect pattern itself.
Since indirect accesses are not part of the index array's stride pattern, IMP first checks that the current load matches no regular stride pattern (line~3).
If so, it gets two earlier addresses \(a_1\) and \(a_2\) that are part of a common stride pattern (line~4), \ie potentially two different accesses into the index array.
IMP then checks whether the current load of address \(a\) is an access into the indirect array that uses the value stored at \(a_2\) as index.
To do so, it computes one candidate indirect base address by subtracting this index from \(a\), and another by subtracting the value stored at \(a_1\) (presumably an earlier index access) from a previously loaded address \(a_c\) that also did not match a regular stride pattern (lines~7--8).
If both candidates match, \(a\) and \(a_c\) are likely accesses belonging to the same indirect pattern, and IMP registers the pattern at line~9.
Since this state update depends on the previously loaded values at \(a_1\) and \(a_2\), it causes leakage through the \emph{Arch} class.
As the size of the indirect array's elements is unknown, IMP considers several shift values \(s\) when computing the candidates (line~6).

The following lines, \ie lines~10--13, verify detected indirect patterns: for a load matching no regular stride pattern (line~3), IMP checks, for each learned pattern (with indirect base address \(b\) and shift \(s\)) and for a recent index access \(a_i\) (with program counter \(\loc'\)) into that pattern's index array, whether the current address \(a\) equals the indirect base address \(b\) plus the index \(\mem(a_i)\) shifted by \(s\) (lines~10--12).
If so, \(a\) is likely an access into the corresponding indirect pattern, and IMP increases the pattern's confidence (line~13), again leaking through \emph{Arch}.

Finally, line~14 records the current address \(a\) as a candidate indirect address, since it matched no regular stride.

The \(\prefetch\) function activates IMP upon loading an entry of the index array of a learned, sufficiently confident indirect pattern, identified by the program counter of the index loads (line~17).
It then loads an entry further ahead in the index array, shifts it, adds it to the indirect base address, and prefetches the result (lines~18--22), hence capturing IMP's \emph{Pref} leakage.
Since our semantics allows speculative loads by default, the model also captures the \emph{Trans} leakage class.

\subsubsection{Applying SecDMP}
The gray boxes in \Cref{fig:sec-imp} depict the modifications required to secure IMP with \secdmp{}: they ensure that \(\train\) never analyzes secret data and that \(\prefetch\) never derives prefetch targets from secrets.

\begin{figure}[t]
    \centering
    \secureImp
    \caption{Model of IMP.
        \(\isStrideAccess{}\) reports whether a load belongs to a regular stride pattern, and \(\getIndexAddresses{}\) returns the earlier address pairs of such a pattern along with the associated program counter.
        \(\appendCandidate{}\) records loads matching no stride pattern, retrieved by \(\getCandidates{}\).
        \(\setIndirectPattern{}\) registers a detected indirect pattern, accessed via \(\getIndirectPatterns{}\) and \(\getIndirectPattern{}\).
        \(\increaseConfidence{}\) and \(\isConfident{}\) track whether a pattern was observed often enough to activate, and \(\getDistance{}\) returns how far ahead in the index array to prefetch.
        IMP is parameterized by its shift values (\(\shiftvalues{}\)); its activation confidence and maximum prefetch distance are abstracted by \(\isConfident{}\) and \(\getDistance{}\).
    }
    \label{fig:sec-imp}
\end{figure}

\if\extendedVersion1
    \section{Full Semantics}
    \subsection{Baseline Speculative Semantics}\label{app:baseline-semantics}
    \input{semantics/baseline.tex}

    \subsection{Semantics of SecDMP}\label{app:secdmp-semantics}
    \input{semantics/secdmp.tex}

    \section{Proofs}\label{app:secdmp-proofs}
    \input{sections/101-proof.tex}

    \section{Evaluation Results} \label{app:evaluation-results}
    \subsection{HACL* Cryptographic Benchmarks} \label{app:hacl-results}
    \cref{tab:hacl-results} contains the numerical results of the performance evaluation for the HACL*~\cite{zinzindohoue17hacl} cryptographic benchmarks.
    The results are given for each defense prototype and for each DMP, relative to the corresponding Proteus+DMP baseline.

    \subsection{Work Benchmarks} \label{app:work-results}
    \Cref{fig:work_results_full} shows the performance overhead on the work benchmarks for each defense prototype and for each DMP, relative to the corresponding Proteus+DMP baseline.
    The work benchmarks consist of Sglib (\ie the \emph{sglib-combined} benchmark from the Embench benchmark suite~\cite{embenchiot}), CoreMark~\cite{coremark}, the custom Array-of-Pointers (AoP) benchmark, and the SpMV benchmark from HPCG~\cite{dongarra16hpcg}.
    \Cref{fig:embench_results} contains the full results of the Embench~\cite{embenchiot} performance evaluation for all 19 benchmarks (including \emph{sglib-combined}) and for the geometric mean over all benchmarks.

    \subsection{DMP-Bench} \label{app:dmp-bench-results}
    \cref{tab:dmp-bench-AppleDmp,tab:dmp-bench-IntelDmp,tab:dmp-bench-Cdp,tab:dmp-bench-Imp} contain the full numerical results of the DMP-Bench performance evaluation for the Apple DMP, Intel DMP, CDP, and IMP, respectively.
    Each table's results are relative to the corresponding Proteus+DMP baseline.
\fi

\input{sections/102-artifact-appendix.tex}

\if\extendedVersion1
    \input{tables/hacl_table.tex}

    \begin{figure*}[!t]
        \centering
        \includegraphics[width=\linewidth]{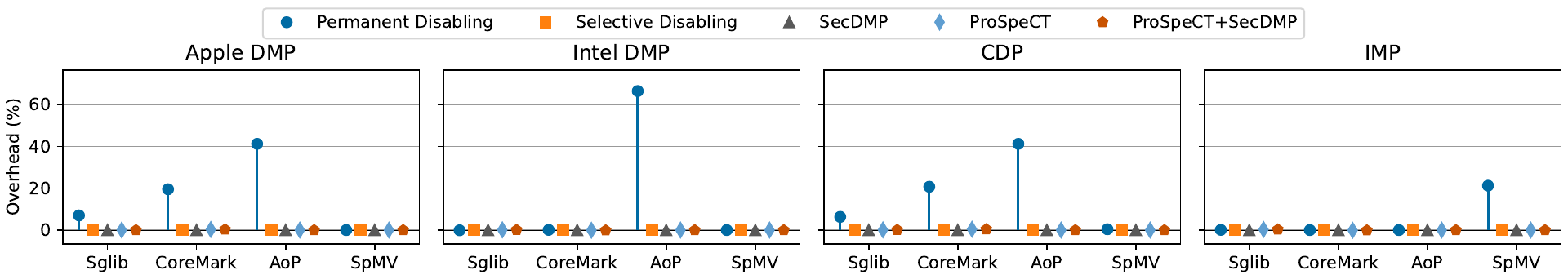}
        \caption{Performance overhead on the different work benchmarks per DMP, relative to Proteus+DMP.}
        \label{fig:work_results_full}
    \end{figure*}

    \begin{figure*}[!t]
        \centering
        \includegraphics[width=\linewidth]{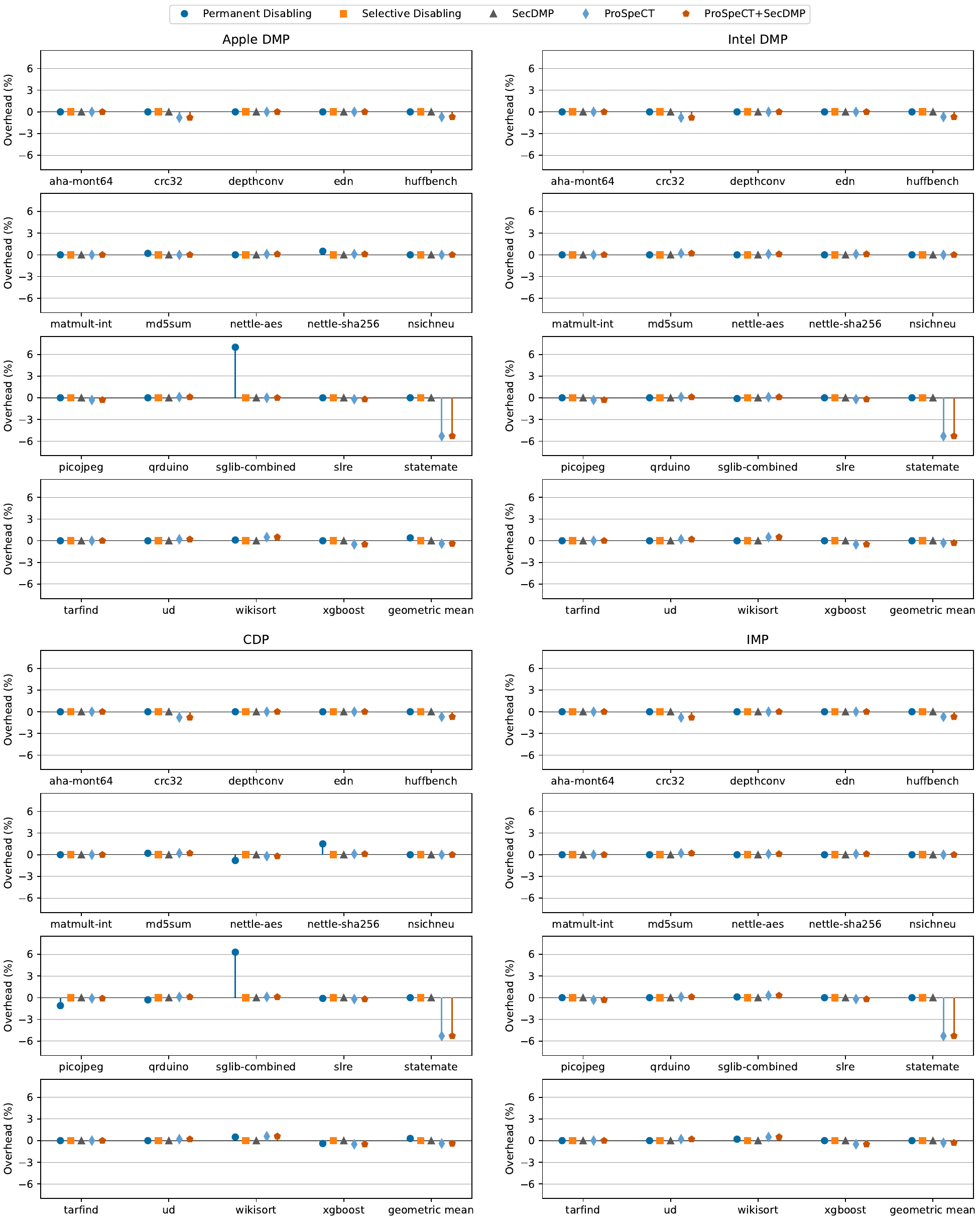}
        \caption{Performance overhead on the Embench benchmark suite per DMP, relative to Proteus+DMP.}
        \label{fig:embench_results}
    \end{figure*}

    \input{tables/dmp_bench_table.tex}

\fi

%% file: sections/104-short-semantics.tex
\section{Semantics of \texorpdfstring{\(\basesem{}\)}{BASE} (excerpt)}\label{app:short-semantics}
We define a representative subset of rules in \Cref{fig:relevant-rules} and leave the full semantics to the extended version~\cite{sprengers27secdmpextended}.

Following the behavior of out-of-order processors, our semantics fetches program instructions in order with the $\fetch{}$ directive and places them in the ROB.
Instructions in the ROB are marked with a tag (denoted \(inst @ T\)) that tracks speculations:
when an instruction $inst$ contains a predicted value, the tag records the location \(\val{l}\) of the instruction that produced the prediction; otherwise, the tag is \(\varepsilon\).
For instance, after fetching a branch or jump instruction, the rule \rulename{fetch-predict-branch-jmp}{fig:relevant-rules} predicts the next location \(\loc'\) and updates \pc{} in the ROB accordingly.

The directive \execute{i} executes the \(i\)-th instruction in the ROB, effectively capturing out-of-order execution.
For instance, the rule \rulename{jmp-commit}{fig:relevant-rules} models the execution of an indirect jump when the prediction is correct.
It evaluates the actual jump target \(\loc{}'\), ensures that it matches the predicted target \(\loc\), and replaces the tag \(\loc_{0}\) with \(\varepsilon\) to mark the instruction as resolved.
The rule \rulename{execute-load}{fig:relevant-rules} executes a \load{}{} instruction: it computes the address \(\val{a}\), retrieves the corresponding value \(\mem(\val{a})\) from the memory, and updates the ROB.
Note that our semantics allows loads to be executed speculatively (\ie{} after a control-flow prediction), however, it does not allow for speculation about store-to-load dependencies nor store-to-load forwarding (\ie{} the rule ensures that there is no preceding store instruction in the ROB).

Instructions are retired in order using the $\retire{}$ directive. Here, the effects of the instruction are committed to the architectural state (\(\reg\) and \(\mem\)).
For instance, the rule \rulename{retire-store-low}{fig:relevant-rules} retires a store instruction on top of the ROB and updates the memory.

\begin{figure}[t]
    \centering
    \begin{mathpar}
        \fetchPredictBranchJump

        \jmpCommit


        \executeLoad

        \retireStoreLow
    \end{mathpar}
    \caption{Evaluation rules for \(\protect\basesem{}\) (excerpt).}
    \label{fig:relevant-rules}
\end{figure}

%% file: semantics/baseline.tex
This section defines the baseline speculative semantics \(\basesem{}\) used in this paper.
It is an out-of-order speculative semantics with direct and indirect control flow speculation (capturing Spectre-v1 and Spectre-v2). It builds on similar existing microarchitectural semantics~\cite{daniel23prospect, guarnieri21hardwaresoftware}.

\paragraph{\(\apl\) function}

The \(\apl\) function applies the effects of the instructions in the reorder buffer (ROB) to the register map.
When the value of an assignment in the reorder buffer is not resolved yet, the corresponding register is mapped to a special symbol \(\botval\), meaning that it is undefined.

\begin{definition}[Apply function $\apl$]\label{def:apl}
    For all register maps \(\reg\) and reorder buffers \(\buf\):
    \begin{equation*}
        \begin{array}{r@{\ }l@{\ }l@{}l@{}l}
            \apl(\varepsilon, \reg)                        & = & \reg                                                                      \\
            \apl(\mov{r}{v} @T\cdot \buf, \reg)            & = & \apl(\buf, \reg\change{\register{r}}{v}) {\sf ~if~} v \in \Vals           \\
            \apl(\mov{r}{e} @T\cdot \buf, \reg)            & = & \apl(\buf, \reg\change{\register{r}}{\botval}) {\sf ~if~} e \not\in \Vals \\
            \apl(\load{x}{e} @T \cdot \buf, \reg)          & = &
            \apl(\buf, \reg\change{\register{x}}{\botval})                                                                                 \\
            \apl(\store{e_{a}}{e_{v}} @T \cdot \buf, \reg) & = & \apl(\buf, \reg)
        \end{array}
    \end{equation*}
\end{definition}

\paragraph{Expression evaluation}
The evaluation of an expression \(e\) with a register map \(\reg\), denoted \(\semexpr{e}{\reg}\), is given in \Cref{fig:eval_expr}.
It is a partial function from expressions to values in \(\Vals\): it is undefined if one of the sub-expressions is undefined (\(\reg(\register{r}) = \botval\)).

\begin{figure}[h]
    \centering
    \begin{equation*}
        \semexpr{e}{\reg} = \begin{cases}
            \val{v}                     & \text{if } e = \val{v} \in \Vals{}                                                                            \\
            \reg(\register{r})          & \text{if } e = \register{r} \in \Regs{} \text{ and } \reg(\register{r}) \in \Vals                             \\
            \val{v_1} \otimes \val{v_2} & \text{if } e = e_1 \otimes e_2, \semexpr{e_1}{\reg} = \val{v_1}, \text{ and } \semexpr{e_2}{\reg} = \val{v_2}
        \end{cases}
    \end{equation*}
    \caption{Evaluation of expressions.}\label{fig:eval_expr}
\end{figure}

\paragraph{Step rule}
The \rulename{step}{fig:step-rule} rule given in \Cref{fig:step-rule} models steps in the hardware semantics \(\basesem{}\).
It determines the next directive to execute using the \(\nextdir\) function and makes a step accordingly using the small-step semantics for directives \(\smicro{\directive}{}{}{}{}\).

\begin{figure}[ht]
    \begin{mathpar}
        \centering
        \stepRule{}
    \end{mathpar}
    \caption{Evaluation rule for \(\protect\basesem{}\).}\label{fig:step-rule}
\end{figure}

\paragraph{Fetch directive}
\begin{itemize}
    \item \rulename{fetch-predict-branch-jmp}{fig:fetch-rules} %
          is the evaluation of a $\fetch{}$ directive when the instruction to fetch is a branch.
          The rule predicts the next location \(\loc'\) and updates $\pc{}$ in the ROB accordingly.
    \item \rulename{fetch-others}{fig:fetch-rules} %
          is the evaluation of a $\fetch{}$ directive when the instruction to fetch is not a branch.
          The rule adds the current instruction to the ROB, increments $\pc{}$,\footnote{For simplicity, we use \(\pc + 1\) to denote the increment of the program counter, but the actual expression depends on the architecture.} and updates the
          ROB accordingly.
    \item In both cases, the current location \(\loc\) is leaked to the microarchitectural context to model updates to model leakage of program control flow.
\end{itemize}

\begin{figure}[ht]
    \begin{mathpar}
        \centering
        \fetchPredictBranchJump{}
        \and
        \fetchOthers{}
    \end{mathpar}
    \caption{Evaluation rules for the \textsc{fetch} directive.}\label{fig:fetch-rules}
\end{figure}

\paragraph{Execute directive for control-flow instructions}
\begin{itemize}
    \item \rulename{branch-commit}{fig:execute-jmp-rules} (resp.\ \rulename{jmp-commit}{fig:execute-jmp-rules}) is the execution of a conditional branch (resp.\ indirect jump) in the ROB. It evaluates the condition (resp.\ jump target), ensures that the predicted target location \(\loc\) corresponds to the actual target, and replaces the tag \(\loc_{0}\) with \(\varepsilon\) to mark the instruction as resolved.
    \item \rulename{branch-rollback}{fig:execute-jmp-rules} (resp.\ \rulename{jmp-rollback}{fig:execute-jmp-rules}) is the execution of a conditional branch (resp.\ indirect jump) in the ROB when the target is incorrectly predicted. The rule discards the instruction that has been mispredicted from the ROB, together with more recent ROB instructions.
    \item In all cases, the actual target location and the location of the branch are leaked to the microarchitectural context to model updates to the branch predictor.
\end{itemize}

\begin{figure}[ht]
    \centering
    \begin{mathpar}
        \branchCommit{}
        \and
        \branchRollback{}
        \and
        \jmpCommit{}
        \and
        \jmpRollback{}
    \end{mathpar}
    \caption{Evaluation rules for the \textsc{execute} directive (for control-flow
        instructions).}\label{fig:execute-jmp-rules}
\end{figure}

\paragraph{Execute directive for memory and assignments}
\begin{itemize}
    \item \rulename{execute-assign}{fig:execute-mem-rules} executes an assignment of an expression \(e\) to a register \register{r}. The rule evaluates \(e\) and updates the value of \(\register{r}\) in the ROB accordingly.
    \item \rulename{execute-load}{fig:execute-mem-rules} executes a \load{}{} instruction in the ROB. The rule computes the address \(\val{a}\) of the load, retrieves the corresponding value \(\mem(\val{a})\) from the memory and updates the ROB.
          Notice that the rule also checks that there is no preceding store instruction in the ROB; it does not allow for speculation about store-to-load dependencies nor store-to-load forwarding.
          The rule also leaks the address to the microarchitectural context to model cache behavior.
    \item \rulename{execute-store}{fig:execute-mem-rules} evaluates both operands of a store instruction and adds the result to the reorder buffer. Notice that the value of the store is not yet committed to the memory; this will be done when the store is retired. The rule also leaks the address of the store to the microarchitectural context to model cache behavior.
\end{itemize}

\begin{figure}[ht]
    \centering
    \begin{mathpar}
        \executeAssign{}
        \and
        \executeLoad{}
        \and
        \executeStore{}
    \end{mathpar}
    \caption{Evaluation rules for the \textsc{execute} directive (assignment and memory instructions).}\label{fig:execute-mem-rules}
\end{figure}

\paragraph{Retire directive}
\begin{itemize}
    \item \rulename{retire-assign}{fig:retire-rules} retires an assignment on top of the ROB and updates the register map accordingly;
    \item \rulename{retire-store-low}{fig:retire-rules} and \rulename{retire-store-high}{fig:retire-rules} retire a store instruction on top of the ROB and updates the memory accordingly. These rules also update the microarchitectural context with the index of the store. The rule \rulename{retire-store-low}{fig:retire-rules} is applied when a secret is written to public memory (\ie{} \(\memsec(\val{a}) = \botsec\)). In that case, the value is also made observable to the attacker by updating the microarchitectural context (\ie{} it is declassified).
\end{itemize}

\begin{figure}[ht]
    \begin{mathpar}
        \centering
        \retireAssign{}
        \and
        \retireStoreLow{}
        \and
        \retireStoreHigh{}
    \end{mathpar}
    \caption{Evaluation rules for the \textsc{retire} directive.}\label{fig:retire-rules}
\end{figure}

%% file: semantics/secdmp.tex
The semantics for \(\secdmpsem{}\) is identical to the baseline speculative semantics \(\basesem{}\) in \cref{app:baseline-semantics}, except for the modifications to the \rulename{execute-load}{fig:execute-mem-rules} rule shown in \Cref{fig:modified-load}.

\begin{figure}[ht]
    \begin{mathpar}
        \updatedSecLoad{}
    \end{mathpar}
    \caption{Evaluation rules for \rulename{execute-load}{fig:execute-mem-rules} in \(\protect\secdmpsem{}\).}\label{fig:modified-load}
\end{figure}

%% file: sections/101-proof.tex
\subsection{Definitions and Notations}
For brevity, in the proof arguments below we write \(\sdmpname{}\) for the \secdmp{} semantics shorthand, \ie \(\sdmpname{} \equiv \secdmpname{}\).

Given a reorder buffer \(\buf\), we let \(\buf[0..i]\) be the sub-buffer containing the first \(i+1\) instructions in \(\buf\).

\begin{definition}[Public equivalence $\loweq$]\label{def:loweq}
    Two configurations \(\conf{\reg,\mem}\) and \(\conf{\reg',\mem'}\) are \(\parsec\)-equivalent, if:
    \begin{enumerate}
        \item \(\reg(\pc) = \reg'(\pc)\),
        \item for all \(\register{x} \in \Regs\), if \(\regsec(\register{x}) = \botsec\) then \(\reg(\register{x}) = \reg'(\register{x})\), and
        \item \(\sproj{\mem} = \sproj{\mem'}\).
    \end{enumerate}
\end{definition}

\begin{definition}[Strong public equivalence]\label{def:strongloweq}
  Two configurations \(\conf{\aconf, \buf, \micro}\) and \(\conf{\aconf', \buf', {\micro'}}\) are \emph{strongly publicly equivalent} \wrt{} a security memory partition \(\memsec\), written \(\conf{\aconf, \buf, \micro} \strongloweq{} \conf{\aconf', \buf', {\micro'}}\), if:
  \begin{itemize}
    \item The program counters are the same in the register maps: \(\reg(\pc) = \reg'(\pc)\);
    \item The program counters are the same for all subsets of reorder buffers: \(\forall i, \semexpr{\pc}{\apl(\buf[0..i], \reg)} = \semexpr{\pc}{\apl(\buf'[0..i], \reg')}\);
    \item For all address \(\val{} \in \Vals\) if \(\memsec(\val{}) = \botsec\) then \(\mem(\val{}) = \mem'(\val{})\);
    \item The microarchitectural states are the same: \(\micro = {\micro'}\).
  \end{itemize}
\end{definition}

\subsection[Proof of Theorem 1]{Proof of \cref{thm:secdmpsecurity}}

\secdmpthm*

\begin{proof}
  By \cref{def:security}, we need to show that for all programs \(\prog{}\), security partitions \(\parsec{}\), and pairs of architectural configurations \(\aconf,\aconf'\):
  \[
  \hwleak{\prog{}}{\parsec{}}{\aconf}{\aconf'}{\sdmpsem{}}
  \implies
  \hwleak{\prog{}}{\parsec{}}{\aconf}{\aconf'}{\basesem{}}
  \]

  \noindent
  Unrolling \cref{def:hwleak}, this is equivalent to showing:
  \begin{gather*}
      \aconf{} \loweq \aconf{'} \wedge \\
      \exists \sdmp{\attacker}.\
      \bigmicro{\sdmpname{}}{\prog{}}{\aconf}{\sdmp{\attacker}} \neq \bigmicro{\sdmpname{}}{\prog{}}{\aconf'}{\sdmp{\attacker}} \\
    \implies \\
    \exists \base{\attacker}.\
    \bigmicro{\base{}}{\prog{}}{\aconf}{\base{\attacker}} \neq \bigmicro{\base{}}{\prog{}}{\aconf'}{\base{\attacker}}.
  \end{gather*}

  Fix two publicly-equivalent architectural configurations \(\aconf{}_{0}\) and \(\aconf'_{0}\), and assume there exists an attacker \(\sdmp{\attacker} = (\sdmp{\micro_0}, \sdmp{\funcs})\) such that running program \(\prog{}\) on \(\sdmpsem{}\) from initial configurations \(\conf{\aconf_{0}, \varepsilon, \micro_{0}}\) and \(\conf{\aconf'_{0}, \varepsilon, \micro_{0}}\) produces different observation traces.
  This means that there exists a step \(n+1\) such that executions start diverging (\ie{} produce different microarchitectural states):
    \begin{gather*}
      \conf{\aconf_{0}, \varepsilon, \micro_{0}^{\sdmpname{}}} \relmicro{}{}^{\sdmpname{}} \dots \relmicro{}{}^{\sdmpname{}} \conf{\aconf_{n}, \beta_{n}, \micro_{n}^{\sdmpname{}}} \\
      \relmicro{}{}^{\sdmpname{}} \conf{\aconf_{n+1}, \beta_{n+1}, \micro_{n+1}^{\sdmpname{}}} \\
      \conf{\aconf'_{0}, \varepsilon, \micro_{0}^{\sdmpname{}}} \relmicro{}{}^{\sdmpname{}} \dots \relmicro{}{}^{\sdmpname{}} \conf{\aconf_{n}', \beta'_{n}, \micro_{n}'^{\sdmpname{}}} \\
      \relmicro{}{}^{\sdmpname{}} \conf{\aconf_{n+1}', \beta'_{n+1}, \micro_{n+1}'^{\sdmpname{}}} \\
      \text{and } \forall 0 \leq i \leq n, \sdmp\micro_{i} = \sdmp{{\micro'}_{i}} \text{ and } \sdmp\micro_{n+1} \neq \sdmp{{\micro'}_{n+1}}
    \end{gather*}

  Without loss of generality, we assume that the function \(\sdmp{\update} \in \sdmp{\funcs}\) leaks \emph{all} its arguments to the microarchitectural state---\ie{} it represents the strongest possible leakage allowed by the semantics.
  Formally, for two sequences of arguments \(a_0, \dots, a_n\), \(a_0', \dots, a_n'\), if there is an $i$ such that $a_i \neq a_i'$, then:
  \begin{equation*}
    \begin{split}
      \sdmp{\update}({a_0, \dots, a_n}) &\neq \\
      \sdmp{\update}({a_0', \dots, a_n'}) &
    \end{split}
    \tag{MaxUpdate}\label{eq:update}
  \end{equation*}
  It is easy to see that for any \(\update\) function that produces a divergence, this divergence is also captured by such a maximally-leaking \(\update\) function. Therefore, this assumption does not restrict the generality of the proof.

  \paragraph{(Goal)} Our goal is to show that there exists an attacker \(\base{\attacker} = (\base{\micro_0}, \base{\funcs})\) such that executions from the same fixed pair \(\aconf_{0}\) and \(\aconf'_{0}\) also diverge in \(\basesem{}\) at step \(n+1\).
  We need to construct the attacker \(\base{\attacker}\) such that:
  \begin{gather*}
    \conf{\aconf_{0}, \varepsilon, \micro_{0}^{\base{}}} \relmicro{}{}^{\base{}} \dots \relmicro{}{}^{\base{}} \conf{\aconf_{n+1}, \beta_{n+1}, \micro_{n+1}^{\base{}}} \\
    \conf{\aconf'_{0}, \varepsilon, \micro_{0}^{\base{}}} \relmicro{}{}^{\base{}} \dots \relmicro{}{}^{\base{}} \conf{\aconf_{n+1}', \beta'_{n+1}, \micro_{n+1}'^{\base{}}} \\
    \text{and } \micro_{n+1}^{\base{}} \neq \micro_{n+1}'^{\base{}}
  \end{gather*}

  \paragraph{(Induction setup)} The proof proceeds by induction on the small step semantics of \(\sdmpsem{}\).
  At each step \(i\) of the execution in \(\sdmpsem{}\),
  we will show that we can construct an attacker \(\base{\attacker} = (\base{\micro_0}, \base{\funcs})\) such that:
  \begin{itemize}
    \item Executions in \(\basesem{}\) produce the same architectural states and reorder buffers as those in \(\secdmpsem{}\) up to step \(i\):
    \begin{equation}
      \begin{gathered}
      \conf{\aconf_{0}, \varepsilon, \micro_{0}^{\base{}}} \relmicro{}{}^{\base{}} \dots \\
      \relmicro{}{}^{\base{}} \conf{\aconf_{i}^{\base{}}, \beta_{i}^{\base{}}, \micro_{i}^{\base{}}} \\
      \text{ and }
      \conf{\aconf'_{0}, \varepsilon, \micro_{0}^{\base{}}} \relmicro{}{}^{\base{}} \dots \\
      \relmicro{}{}^{\base{}} \conf{\aconf_{i}'^{\base{}}, \beta_{i}'^{\base{}}, \micro_{i}'^{\base{}}}\\
      \text { with }
      \aconf_{j}^{\sdmpname{}} = \aconf_{j}^{\base{}} \text{ and }
      \beta_{j}^{\sdmpname{}} = \beta_{j}^{\base{}} \\ \text{ and }
      \aconf_{j}'^{\sdmpname{}} = \aconf_{j}'^{\base{}} \text{ and } \beta_{j}'^{\sdmpname{}} = \beta_{j}'^{\base{}} \\
      \text{ for all } j \leq i.
      \end{gathered}
      \tag{\textbf{IH1}}\label{eq:secdmp-ih1}
    \end{equation}
    \item If the microarchitectural states diverge at step \(i\) in \(\sdmpsem{}\), then they also diverge at step \(i\) in \(\basesem{}\):
    \begin{equation}
      (\micro_{i}^{\sdmpname{}} \neq {\micro'}_{i}^{\sdmpname{}}) \implies (\micro_{i}^{\base{}} \neq {\micro'}_{i}^{\base{}})
      \tag{\textbf{IH2}}\label{eq:secdmp-ih2}
    \end{equation}
  \end{itemize}
  Additionally, we will also show that before the divergence, pairs of configurations in \(\sdmpsem{}\) are related by strong public equivalence:
    \begin{equation}
      \begin{gathered}
        \micro_{i}^{\sdmpname{}} = {\micro'}_{i}^{\sdmpname{}} \implies \\
        \conf{\sdmp{\aconf_{i}}, \sdmp{\buf_i}, \sdmp{\micro_{i}}} \strongloweq
        \conf{\sdmp{\aconf_{i}'}, \sdmp{{\buf'_i}}, \sdmp{{\micro'_{i}}}}
        \tag{\textbf{IH3}}\label{eq:secdmp-ih3}
      \end{gathered}
    \end{equation}

  \paragraph{(Base case)} The base case \(i = 0\) is trivial.
  \cref{eq:secdmp-ih1} follows from the fact that pairs of executions in \(\sdmpsem{}\) and \(\basesem{}\) start from the same pair of initial architectural configurations \((\aconf_{0}, \aconf_{0}')\) and empty reorder buffers;
  \cref{eq:secdmp-ih2} holds vacuously; and
  \cref{eq:secdmp-ih3} holds from the initial assumption that \(\aconf_{0} \loweq \aconf_{0}'\).

  \paragraph{(Inductive case)}
  For the inductive case, we assume that \cref{eq:secdmp-ih1,eq:secdmp-ih2,eq:secdmp-ih3} hold for step \(i\) and we show that they also hold for step \(i+1\).

\begin{itemize}[leftmargin=*, labelwidth=\widthof{\labelitemi~\cref{eq:secdmp-ih1}}, align=left]
\medskip
\item[\labelitemi~\cref{eq:secdmp-ih1}]
We show that we can construct an attacker \(\base{\attacker} = (\base{\micro_0}, \base{\funcs})\) such that executions in \(\basesem{}\) produce the same architectural states and reorder buffers as those in \(\sdmpsem{}\) up to step \(i+1\).

First, we analyze the directives taken at step \(i+1\) in \(\sdmpsem{}\) and show that we can find an attacker strategy \(\base{\nextdir}\) that mimics these directives in \(\basesem{}\) (in both configurations):
\begin{itemize}
    \item
    If \(\micro_{i}^{\sdmpname{}} = {\micro'}_{i}^{\sdmpname{}}\), we know from the \rulename{step}{fig:step-rule} rule that the directive is the same in both configurations in \(\sdmpsem{}\).
    We can therefore choose an attacker strategy \(\base{\nextdir}\) that mimics this directive for step \(i+1\), regardless of the microarchitectural state.\footnote{It is also important when we construct these attacker strategies that their behavior at step \(i+1\) must ``extend'' the behavior at step \(i\), this can easily be done by embedding a counter in the microarchitectural state, or by embedding a list of directives/predictions in the initial microarchitectural context, similar to related work~\cite{olmos25ctcomp}.}
    \item
    If \(\micro_{i}^{\sdmpname{}} \neq {\micro'}_{i}^{\sdmpname{}}\), then directives taken in \(\sdmpsem{}\) may differ between both configurations.
    From \cref{eq:secdmp-ih2} we know that \(\micro_{i}^{\base{}} \neq {\micro'}_{i}^{\base{}}\).
    Therefore, we can construct an attacker strategy \(\base{\nextdir}\) that predicts different directives given \(\micro_{i}^{\base{}}\) and \({\micro'}_{i}^{\base{}}\) to match the directives in \(\secdmpsem{}\).
\end{itemize}

We have shown that we can match directives in \(\basesem{}\) with those in \(\sdmpsem{}\) at step \(i+1\).
We now need to show that executing these directives in \(\basesem{}\) preserves \cref{eq:secdmp-ih1}. We proceed by case analysis on the directive.

Architectural state and reorder buffer are only updated differently between \(\sdmpsem{}\) and \(\basesem{}\) in two rules: \rulename{fetch-predict-branch-jmp}{fig:fetch-rules} (from different predictions from \(\predict\)) or \rulename{execute-load}{fig:execute-mem-rules} rule (because the rule differs in both semantics). (Note that even if \(\update\) differs in both semantics this does not influence the architectural state and reorder buffer).

\begin{itemize}
  \item For \rulename{fetch-predict-branch-jmp}{fig:fetch-rules}, we can apply a similar reasoning as for the \rulename{step}{fig:step-rule} rule to craft a \(\base{\predict}\) function that returns the same prediction as \(\sdmp{\predict}\).
  \item For \rulename{execute-load}{fig:execute-mem-rules}, the only difference between \(\sdmpsem{}\) and \(\basesem{}\) is in the \(\sdmp{\train}\) and \(\sdmp{\prefetch}\) functions. These functions only affect the microarchitectural context \(\micro\) and do not modify the architectural state \(\aconf\) or the reorder buffer \(\buf\).
\end{itemize}

Since architectural states and reorder buffers are equal at step \(i\) (\cref{eq:secdmp-ih1}), and since we have just shown that for all rules we can construct \(\base{\funcs}\) such that the architectural states and reorder buffers are updated identically in \(\basesem{}\) and \(\secdmpsem{}\), then \cref{eq:secdmp-ih1} holds for step \(i+1\).

\bigskip
\item[\labelitemi~\cref{eq:secdmp-ih2}]
Assume that \(\micro_{i+1}^{\sdmpname{}} \neq {\micro'}_{i+1}^{\sdmpname{}}\).
We must show that \(\micro_{i+1}^{\base{}} \neq {\micro'}_{i+1}^{\base{}}\).

If \(\micro_i^{\base{}} \neq {\micro'}_i^{\base{}}\),
then a maximally-leaking \(\base{\update}\), together with \rulename{step}{fig:step-rule}, preserves this divergence at step \(i+1\).

Assume instead that
\(
  \micro_i^{\base{}} = {\micro'}_i^{\base{}}
\).
By the contrapositive of \cref{eq:secdmp-ih2} at step \(i\), we also
have
\(
  \micro_i^{\sdmpname{}} = {\micro'}_i^{\sdmpname{}}.
\)
Consequently, \(\nextdir\) selects the same directive in both
executions. Moreover, by \cref{eq:secdmp-ih3}, the two configurations
in \(\sdmpsem{}\) are strongly publicly equivalent at step \(i\).

We first consider rules other than
\rulename{execute-load}{fig:execute-mem-rules}, and then consider the
load rule separately.

\begin{itemize}
  \smallskip
  \item[]\textbf{(\(\neg\)\rulename{execute-load}
    {fig:execute-mem-rules})}
  These rules are identical in \(\sdmpsem{}\) and \(\basesem{}\).
  Since the incoming microarchitectural states in \(\sdmpsem{}\) are
  equal, a divergence after the rule implies that some other argument
  passed to \(\sdmp{\update}\) differs between the two executions.
  By \cref{eq:secdmp-ih1}, the architectural states and reorder
  buffers in \(\basesem{}\) match those in \(\sdmpsem{}\).
  Hence, the corresponding argument passed to \(\base{\update}\)
  also differs. We can therefore choose a maximally-leaking
  \(\base{\update}\) such that
  \[
    \micro_{i+1}^{\base{}}
    \neq
    {\micro'}_{i+1}^{\base{}}.
  \]

  \medskip
  \item[]\textbf{(\rulename{execute-load}
    {fig:execute-mem-rules})}
  In \(\sdmpsem{}\), the rule computes
  \[
  \begin{aligned}
    \micro_{\mathrm{train}}
      &=
      \sdmp{\train}(
        \micro_i^{\sdmpname{}},
        \loc,
        \sproj{\mem_i},
        \val{a}),\\
    \micro_{\mathrm{pref}}
      &=
      \sdmp{\prefetch}(
        \micro_i^{\sdmpname{}},
        \loc,
        \sproj{\mem_i},
        \val{a}),
  \end{aligned}
  \]
  and passes both candidate states, together with the incoming
  microarchitectural state and the loaded address, to
  \(\sdmp{\update}\).

  Since the incoming microarchitectural states are equal, a
  divergence produced by this final update implies that either the
  training results, the prefetching results, or the loaded addresses
  differ between the two executions.

  By strong public equivalence,
  \(
    \sproj{\mem_i}=\sproj{\mem'_i},
  \)
  and the current locations
  \[
    \loc=\semexpr{\pc}{\apl(\buf_i,\reg_i)}
    \quad\text{and}\quad
    \loc'=\semexpr{\pc}{\apl(\buf'_i,\reg'_i)}
  \]
  are equal. Thus, if the loaded addresses were also equal, all
  arguments passed to \(\train\) and \(\prefetch\) would be equal.
  Their resulting candidate states would therefore be equal as well.

  It follows that a divergence in the \secdmp{} load rule can only arise
  if
  \(
    \val{a}\neq\val{a}'
  \).
  The baseline \rulename{execute-load}{fig:execute-mem-rules} rule
  passes this address to \(\base{\update}\). Therefore, a
  maximally-leaking \(\base{\update}\) also produces
  \[
    \micro_{i+1}^{\base{}}
    \neq
    {\micro'}_{i+1}^{\base{}}.
  \]
\end{itemize}

In every case, a divergence in \(\sdmpsem{}\) at step \(i+1\) can be
reproduced in \(\basesem{}\). Hence,
\cref{eq:secdmp-ih2} holds at step \(i+1\).

\bigskip
\item[\labelitemi~\cref{eq:secdmp-ih3}]
Assuming \(\sdmp{\micro_{i+1}} = \sdmp{{\micro'}_{i+1}}\), we need to show that:
\begin{gather*}
\conf{\sdmp{\aconf_{i+1}}, \sdmp{\buf_{i+1}}, \sdmp{\micro_{i+1}}}  \strongloweq
\conf{\sdmp{\aconf_{i+1}'}, \sdmp{{\buf'}_{i+1}}, \sdmp{{\micro'}_{i+1}}}
\end{gather*}

Note that from \(\sdmp{\micro_{i+1}} = \sdmp{{\micro'}_{i+1}}\) and \cref{eq:update}, we have \(\sdmp{\micro_{i}} = \sdmp{{\micro'}_{i}}\).

From the definition of strong public equivalence (\cf{} \cref{def:strongloweq}), we need to show that updates to \(\pc\) and low-memory locations are the same in both configurations in \(\sdmpsem{}\).

\begin{itemize}
  \item
  Rules that modify the program counter are \rulename{fetch-predict-branch-jmp}{fig:fetch-rules}, \rulename{fetch-others}{fig:fetch-rules},
  \rulename{branch-commit}{fig:execute-jmp-rules}, \rulename{branch-rollback}{fig:execute-jmp-rules}, \rulename{jmp-commit}{fig:execute-jmp-rules}, \rulename{jmp-rollback}{fig:execute-jmp-rules} and \rulename{retire-assign}{fig:retire-rules}. %

  In \rulename{fetch-predict-branch-jmp}{fig:fetch-rules} the update is respectively determined by \(\sdmp{\predict}(\sdmp{\micro_{i}})\) and \(\sdmp{\predict}(\sdmp{{\micro'}_{i}})\) in both executions. Since \(\sdmp{\micro_{i}} = \sdmp{{\micro'}_{i}}\) both predictions are the same and hence both updates to \(\pc\) are the same as well. In \rulename{fetch-others}{fig:fetch-rules}, $\pc$ is simply incremented so equality of next $\pc$ follows from equality of current $\pc$ \cref{eq:secdmp-ih3}.

  In \rulename{branch-commit}{fig:execute-jmp-rules}, \rulename{branch-rollback}{fig:execute-jmp-rules}, \rulename{jmp-commit}{fig:execute-jmp-rules}, \rulename{jmp-rollback}{fig:execute-jmp-rules}, the locations used to update \(\pc\) are passed to the \(\sdmp{\update}{}\) function.
  Given \(\sdmp{\micro_{i+1}} = \sdmp{{\micro'}_{i+1}}\), \cref{eq:update} implies that the arguments passed to \(\sdmp{\update}\) must be identical in both configurations.
  Thus, the new assignments to the \(\pc\) are the same in both configurations.

  Finally, in \rulename{retire-assign}{fig:retire-rules}, the update to \(\pc\) in the register map is the same as the update to \(\pc\) at the top of the reorder buffer, which are the same in both configurations by \cref{eq:secdmp-ih3}.

  \item
  Rules that modify memory are \rulename{retire-store-high}{fig:retire-rules} and \rulename{retire-store-low}{fig:retire-rules}.
  In both rules, the address being written to is passed to the \(\sdmp{\update}\) function.
  Given that \(\sdmp{\micro_{i+1}} = \sdmp{{\micro'}_{i+1}}\), \cref{eq:update} implies that the addresses are the same in both configurations.
  The rule \rulename{retire-store-high}{fig:retire-rules} writes to the secret part of memory, therefore \cref{eq:secdmp-ih3} at step \(i+1\) follows from \cref{eq:secdmp-ih3} at step \(i\).
  In \rulename{retire-store-low}{fig:retire-rules}, the value being written to memory is leaked to the microarchitectural state via \(\sdmp{\update}\). Again, from \(\sdmp{\micro_{i+1}} = \sdmp{{\micro'}_{i+1}}\) and \cref{eq:update}, we get that these values are the same in both configurations.
\end{itemize}

For each rule that modifies \(\pc\) or low-memory locations, we have shown that the updates are the same in both configurations in \(\sdmpsem{}\). Therefore, \cref{eq:secdmp-ih3} holds for step \(i+1\) as well.
\end{itemize}

We have shown that for all steps \(i\) our induction hypotheses \cref{eq:secdmp-ih1,eq:secdmp-ih2,eq:secdmp-ih3} hold.
Our goal was to show that if microarchitectural states start diverging at step \(n+1\) in \(\sdmpsem{}\), then they also start diverging at step \(n+1\) in \(\basesem{}\). This directly follows from \cref{eq:secdmp-ih2}, which concludes our proof.
\end{proof}

\subsection[Proof of Theorem 2]{Proof of \cref{thm:secdmpbase-security}}\label{app:secdmpbase-security}

First, we formalize the assumptions on the base semantics needed to prove \cref{thm:secdmpbase-security}.

\begin{restatable}{assumption}{hypmicro}\label{hyp:micro}
    In \(\absbasesem{}\), the microarchitectural state \(\micro\) is only updated through the functions \(\update{}\).
\end{restatable}

\cref{hyp:corr} assumes that the semantics cannot encode information from \(\micro\) into the architectural state or reorder buffer without using \(\predict\) or \(\nextdir\). Formally, if two configurations agree on their architectural state, reorder buffer, predicted values, and next directives, then their subsequent architectural states and reorder buffers must also agree:

\begin{restatable}{assumption}{hypcorr}\label{hyp:corr}
    For any architectural configuration \(\aconf{}\), reorder buffer \(\buf\) and microarchitectural states \(\micro\), \(\micro'\):
    \begin{gather*}
        \text{if }
        \smicro{}{\absbase{}}{\conf{\aconf, \buf, \micro}\ }{\conf{\aconf_{1}, \buf_{1}, \micro_{1}}}{} \text{, }
        \smicro{}{\absbase{}}{\conf{\aconf, \buf, \micro'}\ }{\conf{\aconf'_{1}, \buf'_{1}, \micro'_{1}}}{} \text{, } \\
        \predict(\micro) = \predict(\micro') \text{, and }
        \nextdir(\micro) = \nextdir(\micro') \\
        \text{then }
        \aconf_{1} = \aconf'_{1} \text{ and } \buf_{1} = \buf'_{1}
    \end{gather*}
\end{restatable}

Finally, \cref{hyp:ct} assumes that the base semantics satisfies a constant-time leakage model with declassification to public memory. Formally, it means that if we assume an injective \(\update\) function, for any two executions in \(\absbasesem{}\), if the outputs of any of the functions \(\getpc\), \(\loadaddr\), or the public projection of memory differ between two configurations at step \(n\), then their microarchitectural states must also differ at step \(n+1\).

\begin{restatable}{assumption}{hypct}\label{hyp:ct}
    Assume an injective \(\update\).
    For any configurations \(c, c', c_1, c_1'\) with memories \(\mem, \mem', \mem_1, \mem_1'\),
    and for any directive \(\directive{}\), if
    \(c \relmicro{\directive{}}{}^{\absbase{}} c_1\) and
    \(c' \relmicro{\directive{}}{}^{\absbase{}} c_1'\), then:
    \begin{gather*}
        \text{if }
        \getpc(\mconf, \directive{}) \neq \getpc(\mconf', \directive{}) \\
        \text{ or if } \loadaddr(\mconf, \directive{}) \neq \loadaddr(\mconf', \directive{})\
        \\
        \text{or if } (\sproj{\mem} = \sproj{\mem'} \text{ and } \sproj{\mem_1} \neq \sproj{\mem_1'}) \\
        \text{then } \micro_1 \neq \micro_1', \\
        \text{with } \micro_1, \micro_1' \text{ the microarchitectural contexts of } c_1, c_1'.
    \end{gather*}
\end{restatable}

\noindent
Under these assumptions, we can derive the following lemma:

\begin{restatable}{lemma}{}\label{lemma:corr}
For any architectural configuration \(\aconf{}\), reorder buffer \(\buf\) and microarchitectural states \(\sdmpbase{\micro}\), \(\absbase{\micro}\):
\begin{gather*}
\text{if }
\smicro{}{\absbase{}}{\conf{\aconf, \buf, \absbase{\micro}}\ }{\conf{\absbase{{\aconf'}}, \absbase{{\buf'}}, \absbase{{\micro'}}}}{} \text{, } \\
\smicro{}{\sdmpbase{}}{\conf{\aconf, \buf, \sdmpbase{\micro}}\ }{\conf{\sdmpbase{{\aconf'}}, \sdmpbase{{\buf'}}, \sdmpbase{{\micro'}}}}{} \text{, } \\
\absbase{\predict}(\absbase{\micro}) = \sdmpbase{\predict}(\sdmpbase{\micro}) \text{, and } \\
\absbase{\nextdir}(\absbase{\micro}) = \sdmpbase{\nextdir}(\sdmpbase{\micro}) \\
\text{then }
\sdmpbase{{\aconf'}} = \absbase{{\aconf'}} \text{ and } \sdmpbase{{\buf'}} = \absbase{{\buf'}}
\end{gather*}
where \(\absbase{\predict}\), \(\absbase{\nextdir}\) (resp. \(\sdmpbase{\predict}\), \(\sdmpbase{\nextdir}\)) are the prediction and next-directive functions in \(\absbasesem{}\) (resp. \(\sdmpbasesem{}\)).
\end{restatable}
\begin{proof}
  This follows directly from \cref{hyp:corr} and from the fact that the semantics of \(\sdmpbase{}\) does not add extra modifications to the architectural state or reorder buffer compared to \(\absbase{}\) (\ie{} the rule \rulename{dmp-activate}{fig:execute-load-dmp} only adds extra modifications to the microarchitectural state).
\end{proof}

\secdmpbasethm*

\begin{proof}
  By \cref{def:security}, we need to show that for all programs \(\prog{}\), security partitions \(\parsec{}\), and pairs of architectural configurations \(\aconf,\aconf'\):
  \[
  \hwleak{\prog{}}{\parsec{}}{\aconf}{\aconf'}{\sdmpbasesem{}}
  \implies
  \hwleak{\prog{}}{\parsec{}}{\aconf}{\aconf'}{\absbasesem{}}
  \]

  \noindent
  Unrolling \cref{def:hwleak}, this is equivalent to showing:
  \begin{gather*}
      \aconf{} \loweq \aconf{'} \wedge \\
      \exists \sdmpbase{\attacker} \text{ s.t. } \\
      \bigmicro{\sdmpbase{}}{\prog{}}{\aconf}{\sdmpbase{\attacker}} \neq \bigmicro{\sdmpbase{}}{\prog{}}{\aconf'}{\sdmpbase{\attacker}} \\
    \implies \\
    \exists \absbase{\attacker} \text{ s.t. }\
    \bigmicro{\absbase{}}{\prog{}}{\aconf}{\absbase{\attacker}} \neq \bigmicro{\absbase{}}{\prog{}}{\aconf'}{\absbase{\attacker}}
  \end{gather*}

  Fix two publicly-equivalent architectural configurations \(\aconf{}_{0}\) and \(\aconf{' }_{0}\), and assume there exists an attacker \(\sdmpbase{\attacker} = (\sdmpbase{\micro_0}, \sdmpbase{\funcs})\) such that running program \(\prog{}\) on \(\sdmpbasesem{}\) from initial configurations \(\conf{\aconf_{0}, \varepsilon, \micro_{0}}\) and \(\conf{\aconf'_{0}, \varepsilon, \micro_{0}}\) produces different observation traces.
  This means that there exists a step \(n+1\) such that executions start diverging (\ie{} produce different microarchitectural states):
  \begin{gather*}
    \conf{\aconf_{0}, \varepsilon, \micro_{0}^{\sdmpbase{}}} \relmicro{}{}^{\sdmpbase{}} \dots \relmicro{}{}^{\sdmpbase{}} \conf{\aconf_{n}, \beta_{n}, \micro_{n}^{\sdmpbase{}}} \\
    \relmicro{}{}^{\sdmpbase{}} \conf{\aconf_{n+1}, \beta_{n+1}, \micro_{n+1}^{\sdmpbase{}}} \\
    \conf{\aconf'_{0}, \varepsilon, \micro_{0}^{\sdmpbase{}}} \relmicro{}{}^{\sdmpbase{}} \dots \relmicro{}{}^{\sdmpbase{}} \conf{\aconf_{n}', \beta'_{n}, \micro_{n}'^{\sdmpbase{}}} \\
    \relmicro{}{}^{\sdmpbase{}} \conf{\aconf_{n+1}', \beta'_{n+1}, \micro_{n+1}'^{\sdmpbase{}}} \\
        \text{and } \forall 0 \leq i \leq n, \sdmpbase\micro_{i} = \sdmpbase{{\micro'}_{i}} \text{ and } \sdmpbase\micro_{n+1} \neq \sdmpbase{{\micro'}_{n+1}}
  \end{gather*}

  Similar to the proof for \cref{thm:secdmpsecurity} and without loss of generality, we assume an injective \(\sdmpbase{\update}\) (\ie{} \cref{eq:update}).

  \paragraph{(Goal)}
  Our goal is to show that there exists an attacker \(\absbase{\attacker} = (\absbase{\micro_0}, \absbase{\funcs})\) such that executions from the same fixed pair \(\aconf_{0}\) and \(\aconf'_{0}\) also diverge in \(\absbasesem{}\) at step \(n+1\).
  We need to construct the attacker \(\absbase{\attacker}\) such that:
  \begin{gather*}
    \conf{\aconf_{0}, \varepsilon, \micro_{0}^{\absbase{}}} \relmicro{}{}^{\absbase{}} \dots \relmicro{}{}^{\absbase{}} \conf{\aconf_{n+1}, \beta_{n+1}, \micro_{n+1}^{\absbase{}}} \\
    \conf{\aconf'_{0}, \varepsilon, \micro_{0}^{\absbase{}}} \relmicro{}{}^{\absbase{}} \dots \relmicro{}{}^{\absbase{}} \conf{\aconf_{n+1}', \beta'_{n+1}, \micro_{n+1}'^{\absbase{}}} \\
    \text{and } \micro_{n+1}^{\absbase{}} \neq \micro_{n+1}'^{\absbase{}}
  \end{gather*}

  \paragraph{(Induction setup)} The proof proceeds by induction on the small step semantics of \(\sdmpbasesem{}\).
  At each step \(i\) of the execution in \(\sdmpbasesem{}\),
  we will show that we can construct an attacker \(\absbase{\attacker} = (\absbase{\micro_0}, \absbase{\funcs})\) such that:
  \begin{itemize}
    \item Executions in \(\absbasesem{}\) produce the same architectural states and reorder buffers as those in \(\sdmpbasesem{}\) up to step \(i\):
    \begin{equation}
      \begin{gathered}
      \conf{\aconf_{0}, \varepsilon, \micro_{0}^{\absbase{}}} \relmicro{}{}^{\absbase{}} \dots
      \relmicro{}{}^{\absbase{}} \conf{\aconf_{i}^{\absbase{}}, \beta_{i}^{\absbase{}}, \micro_{i}^{\absbase{}}} \\
      \text{ and }
      \conf{\aconf'_{0}, \varepsilon, \micro_{0}^{\absbase{}}} \relmicro{}{}^{\absbase{}} \dots
      \relmicro{}{}^{\absbase{}} \conf{\aconf_{i}'^{\absbase{}}, \beta_{i}'^{\absbase{}}, \micro_{i}'^{\absbase{}}}\\
      \text { with }
      \aconf_{j}^{\sdmpbase{}} = \aconf_{j}^{\absbase{}} \text{ and }
      \beta_{j}^{\sdmpbase{}} = \beta_{j}^{\absbase{}} \\ \text{ and }
      \aconf_{j}'^{\sdmpbase{}} = \aconf_{j}'^{\absbase{}} \text{ and }
      \beta_{j}'^{\sdmpbase{}} = \beta_{j}'^{\absbase{}} \\
      \text{ for all } j \leq i.
      \end{gathered}
      \tag{\textbf{IH1}}\label{eq:secdmpbase-ih1}
    \end{equation}
    \item If the microarchitectural states diverge at step \(i\) in \(\sdmpbasesem{}\), then they also diverge at step \(i\) in \(\absbasesem{}\):
    \begin{equation}
      (\micro_{i}^{\sdmpbase{}} \neq {\micro'}_{i}^{\sdmpbase{}}) \implies (\micro_{i}^{\absbase{}} \neq {\micro'}_{i}^{\absbase{}})
      \tag{\textbf{IH2}}\label{eq:secdmpbase-ih2}
    \end{equation}
  \end{itemize}
  Additionally, we will also show that before the divergence, pairs of configurations in \(\sdmpbase{}\) have the same public memory:
    \begin{equation}
      \micro_{i}^{\sdmpbase{}} = {\micro'}_{i}^{\sdmpbase{}} \implies
      \smemproj{\sdmpbase{\mem_{i}}} = \smemproj{\sdmpbase{{\mem'}_{i}}}. \tag{\textbf{IH3}}\label{eq:secdmpbase-ih3}
    \end{equation}

  \paragraph{(Base case)} The base case \(i = 0\) is trivial.

  \paragraph{(Inductive case)}
  For the inductive case, we assume that \cref{eq:secdmpbase-ih1,eq:secdmpbase-ih2,eq:secdmpbase-ih3} hold for step \(i\) and we show that they also hold for step \(i+1\).

\begin{itemize}[leftmargin=*, labelwidth=\widthof{\labelitemi~\cref{eq:secdmpbase-ih1}}, align=left]
\medskip
\item[\labelitemi~\cref{eq:secdmpbase-ih1}]
We show that we can construct an attacker \(\absbase{\attacker} = (\absbase{\micro_0}, \absbase{\funcs})\) such that executions in \(\absbasesem{}\) produce the same architectural states and reorder buffers as those in \(\sdmpbasesem{}\) up to step \(i+1\).

From \cref{lemma:corr}, we know that if we can match predictions and directives taken at step \(i+1\) in \(\absbasesem{}\) with those in \(\sdmpbasesem{}\), then the architectural states and reorder buffers will be updated identically in both semantics.

\begin{itemize}
    \item
    If \(\micro_{i}^{\sdmpbase{}} = {\micro'}_{i}^{\sdmpbase{}}\), we can simply construct an attacker strategy \(\absbase{\nextdir}\), \(\absbase{\predict}\) that mimics the prediction for step \(i+1\), regardless of the microarchitectural state.
    \item
    If \(\micro_{i}^{\sdmpbase{}} \neq {\micro'}_{i}^{\sdmpbase{}}\), then predictions taken in \(\sdmpbasesem{}\) may differ between both configurations.
    Therefore, we can construct an attacker strategy \(\absbase{\nextdir}\), \(\absbase{\predict}\) that predicts different values given \(\micro_{i}^{\absbase{}}\) and \({\micro'}_{i}^{\absbase{}}\) to match the predictions in \(\sdmpbasesem{}\).
    From \cref{eq:secdmpbase-ih2} we know that \(\micro_{i}^{\absbase{}} \neq {\micro'}_{i}^{\absbase{}}\).
\end{itemize}

We have shown that we can match predictions and directives from \(\absbasesem{}\) with those in \(\sdmpbasesem{}\) at step \(i+1\), which concludes the proof of \cref{eq:secdmpbase-ih1}.

\bigskip
\item[\labelitemi~\cref{eq:secdmpbase-ih2}] Assuming \(\micro_{i+1}^{\sdmpbase{}} \neq {\micro'}_{i+1}^{\sdmpbase{}}\), we must show that \(\micro_{i+1}^{\absbase{}} \neq {\micro'}_{i+1}^{\absbase{}}\) as well.

From \cref{hyp:micro}, we know that, except in the rule \rulename{dmp-activate}{fig:execute-load-dmp}, the microarchitectural state is only updated through the functions \(\absbase{\update{}}\) and \(\sdmpbase{\update{}}\) in \(\basesem{}\) and \(\sdmpbasesem{}\) respectively.
Thus, we need to show that if the \(\sdmpbase{\update{}}\) produces different results, then we can also derive a function \(\absbase{\update{}}\) that produces different results.

If there is already a divergence in \(\absbasesem{}\) (\(\absbase{\micro_{n}} \neq \absbase{\micro_{n}'}\)) then any injective  \(\absbase{\update{}}\) function will preserve this divergence.
The interesting case is when \(\micro_{n}^{\text{base}} = {\micro'}_{n}^{\text{base}}\).
Note that in that case, from \cref{eq:secdmpbase-ih2}, we also have \(\sdmpbase{\micro_{n}} = \sdmpbase{{\micro'}_{n}}\), so the directives are also the same in both executions.
From this, we get that the rules applied at step \(i+1\) in both executions in \(\sdmpbasesem{}\) are the same.
We first consider the case where both executions in \(\sdmpbasesem{}\) follow the rule \rulename{dmp-no-activate}{fig:execute-load-dmp}, and then the case where they both follow \rulename{dmp-activate}{fig:execute-load-dmp}.

\begin{itemize}
  \smallskip
  \item[]\textbf{(\rulename{dmp-no-activate}{fig:execute-load-dmp})}
  Similar to the proof of \cref{thm:secdmpsecurity}, case \textbf{(\(\neg\)\rulename{execute-load}{fig:execute-mem-rules})}.

  \medskip
  \item[]\textbf{(\rulename{dmp-activate}{fig:dmpbase})}
Let \(\sdmpbase{\micro_{\mathrm B}}\) and
\(\sdmpbase{{\micro'}_{\mathrm B}}\) denote the microarchitectural
states produced by the \(\absbasesem{}\) step, and let
\(\sdmpbase{\micro_{\mathrm D}}\) and
\(\sdmpbase{{\micro'}_{\mathrm D}}\) denote the states produced by the
\rulename{dmp-activate}{fig:dmpbase} rule before the final application
of \(\update\) in \rulename{step}{fig:step-rule}.

In the step rule, the incoming microarchitectural states are equal. Therefore, a divergence after the \rulename{step}{fig:step-rule} rule implies
\(\sdmpbase{\micro_{\mathrm D}} \neq \sdmpbase{{\micro'}_{\mathrm D}}\).

The \rulename{dmp-activate}{fig:dmpbase} rule computes these states as
\[
\begin{aligned}
  \sdmpbase{\micro_{\mathrm D}}
  =&
  \sdmpbase{\update}(
    \sdmpbase{\micro_i},
    \sdmpbase{\micro_{\mathrm B}},
    \mathsf{activate},\\
  &
    \sdmpbase{\micro_{\mathrm{train}}},
    \sdmpbase{\micro_{\mathrm{pref}}},
    \sdmpbase{\val{a}}),\\
  \sdmpbase{{\micro'}_{\mathrm D}}
  =&
  \sdmpbase{\update}(
    \sdmpbase{{\micro'}_i},
    \sdmpbase{{\micro'}_{\mathrm B}},
    \mathsf{activate}, \\
  &
    \sdmpbase{{\micro'}_{\mathrm{train}}},
    \sdmpbase{{\micro'}_{\mathrm{pref}}},
    \sdmpbase{\val{a}'}).
\end{aligned}
\]
Consequently, at least one of the base-step result, the training result, the prefetching result, or the loaded address differs between the two executions.

We distinguish two cases:
\begin{itemize}
  \item If
  \(\sdmpbase{\micro_{\mathrm B}}
    \neq \sdmpbase{{\micro'}_{\mathrm B}}\),
  then the divergence already arises from the \(\absbasesem{}\) step.
  As in the \rulename{dmp-no-activate}{fig:dmpbase} case, we can construct an attacker strategy \(\absbase{\update}\) for which
  \(
    \absbase{\micro_{i+1}}
    \neq
    \absbase{{\micro'}_{i+1}}
  \).

  \item Assume instead that
  \(\sdmpbase{\micro_{\mathrm B}}
    = \sdmpbase{{\micro'}_{\mathrm B}}\).
  The training and prefetching results are computed independently
  from the incoming microarchitectural state, the current location,
  the public-memory projection, and the loaded address.
  The incoming microarchitectural states are equal, and
  \cref{eq:secdmpbase-ih3} gives
  \[
    \smemproj{\sdmpbase{\mem_i}}
    =
    \smemproj{\sdmpbase{{\mem'}_i}}.
  \]
  Hence, a difference in the training result, the prefetching result,
  or the loaded address implies that either \(\getpc\) or
  \(\loadaddr\) differs between the two executions.

  By \cref{hyp:ct}, such a difference would imply
  \[
    \sdmpbase{\micro_{\mathrm B}}
    \neq
    \sdmpbase{{\micro'}_{\mathrm B}},
  \]
  contradicting the assumption of this case. Therefore, this case
  cannot produce a divergence.
\end{itemize}

Thus, any divergence produced by
\rulename{dmp-activate}{fig:dmpbase} must already be reflected by the
\(\absbasesem{}\) step, and we can construct an attacker
strategy such that
\(
  \absbase{\micro_{i+1}}
  \neq
  \absbase{{\micro'}_{i+1}}.
\)
\end{itemize}
We have shown for each rule that if there is a divergence in \(\sdmpbasesem{}\) at step \(i+1\), then we can construct an attacker strategy \(\absbase{\update{}}\) such that there is also a divergence in \(\absbasesem{}\) at step \(i+1\).
Thus, \cref{eq:secdmpbase-ih2} holds for step \(i+1\) as well.

\bigskip
\item[\labelitemi~\cref{eq:secdmpbase-ih3}]
Assume that
\(\sdmpbase{\micro_{i+1}} =
  \sdmpbase{{\micro'}_{i+1}}\).
We must show that
\[
  \smemproj{\sdmpbase{\mem_{i+1}}}
  =
  \smemproj{\sdmpbase{{\mem'}_{i+1}}}.
\]

Let \(\sdmpbase{\micro_{\mathrm D}}\) and \(\sdmpbase{{\micro'}_{\mathrm D}}\) denote the microarchitectural states produced by the directive rules before the
final application of \(\sdmpbase{\update}\) in \rulename{step}{fig:step-rule}.
By \cref{eq:update}, equality of the final microarchitectural states implies
\[
  \sdmpbase{\micro_i} = \sdmpbase{{\micro'}_i}
  \quad\text{and}\quad
  \sdmpbase{\micro_{\mathrm D}}
  = \sdmpbase{{\micro'}_{\mathrm D}}.
\]
We can therefore apply the induction assumption
\cref{eq:secdmpbase-ih3} at step \(i\), obtaining
\[
  \smemproj{\sdmpbase{\mem_i}}
  =
  \smemproj{\sdmpbase{{\mem'}_i}}.
\]

Since the incoming microarchitectural states are equal, \(\nextdir\) returns the same directive in both executions, the current locations are equal, and \(\dmpactivate\) therefore selects the same DMP rule.
Let \(\sdmpbase{\micro_{\mathrm B}}\) and \(\sdmpbase{{\micro'}_{\mathrm B}}\) be the microarchitectural states produced by the \(\absbasesem{}\) step.
In both \rulename{dmp-activate}{fig:dmpbase} and \rulename{dmp-no-activate}{fig:dmpbase}, these states are passed as arguments to the final \(\sdmpbase{\update}\).
Hence, from
\(\sdmpbase{\micro_{\mathrm D}}
 = \sdmpbase{{\micro'}_{\mathrm D}}\)
and \cref{eq:update}, we obtain
\[
  \sdmpbase{\micro_{\mathrm B}}
  =
  \sdmpbase{{\micro'}_{\mathrm B}}.
\]

The DMP rules preserve the memory produced by the \(\absbasesem{}\) step, so any difference in public memory at step \(i+1\) would have to come from this base step.
Additionally, the base step starts from equal public memories, and \cref{hyp:ct} tells us that unequal public memories would force the resulting microarchitectural states to differ.
Since we have
\(\sdmpbase{\micro_{\mathrm B}} = \sdmpbase{{\micro'}_{\mathrm B}}\),
the base step cannot produce different public memories.
Therefore,
\[
  \smemproj{\sdmpbase{\mem_{i+1}}}
  =
  \smemproj{\sdmpbase{{\mem'}_{i+1}}},
\]
which concludes the proof of \cref{eq:secdmpbase-ih3}.
\end{itemize}

We have shown that for all steps \(i\) our induction hypotheses \cref{eq:secdmpbase-ih1,eq:secdmpbase-ih2,eq:secdmpbase-ih3} hold.
Our goal was to show that if microarchitectural states start diverging at step \(n+1\) in \(\sdmpbasesem{}\), then they also start diverging at step \(n+1\) in \(\absbasesem{}\). This directly follows from \cref{eq:secdmpbase-ih2}, which concludes our proof.

\end{proof}

%% file: sections/102-artifact-appendix.tex
\section{Artifact Appendix}

The \secdmp{} artifact contains all implementations, benchmarks, and scripts from \secdmp{}'s implementation and evaluation (\cf \cref{sec:implementation-evaluation}).
It consists of the source code of the Proteus \riscv{} processor, the four open-source DMP implementations (the Apple DMP, Intel DMP, CDP, and IMP, see \cref{sec:dmp-models}), and the implementation of \secdmp{} and of the defenses we compare against (\ie permanent disabling, selective disabling, \prospect{} and \prospect{}+\secdmp{}).
It additionally contains all security and performance benchmarks and all evaluation scripts, including the scripts to generate the paper's figures and tables.
A Dockerfile provides the complete build environment for reproducing the evaluation.

\subsection{Description \& Requirements}

\subsubsection{How to access}
The artifact is available at \artifactrepo{} and has been archived on Zenodo at \zenodoartifact{}.
It contains a comprehensive \texttt{README.md} file, describing all artifacts and containing detailed instructions for reproducing the evaluation results, and is released under the MIT license.

\subsubsection{Hardware dependencies}
None.

\subsubsection{Software dependencies}
The artifact includes a Dockerfile containing all necessary dependencies.
As such, the only requirement is Docker and approximately \SI{20}{\giga\byte} of disk space for the Docker container, Docker build cache, additional storage for temporary files during the Docker build process, and the space occupied by the compiled benchmarks.

\subsubsection{Benchmarks}
All benchmarks are included in the artifact repository.
The security evaluation uses our nine DMP leakage programs (IMP Leak, Augury PoC, GoFetch PoC, Peek-a-Walk, DMP-Spectre, Recursion, Load Secret, Same Cache Block, Transient Load) and the five Spectre-PHT tests (PHT-Load, PHT-Store, PHT-Branch, PHT-Jump, PHT-Division).
The performance evaluation uses five cryptographic primitives (Chacha20, Salsa20, SHA2, HMAC, HKDF) from the HACL* cryptographic library, the Embench benchmark suite, CoreMark, the SpMV benchmark from HPCG, our Array-of-Pointers (AoP) benchmark, and finally DMP-Bench, which combines each work benchmark with each cryptographic primitive at four work percentage variations.

\subsection{Artifact Installation \& Configuration}
\label{app:artifact-app-setup}

Perform the following steps to install and configure the artifact and its dependencies:
\begin{enumerate}
    \item Install Docker following the instructions at \url{https://docs.docker.com/engine/install/};
    \item Clone the artifact repository from \artifactrepo{};
    \item Build the Docker image (takes approximately 20 compute-minutes, detailed instructions are provided in the artifact's \texttt{README.md}).
\end{enumerate}

The \verb|verify_installation.sh| script runs a minimal working example in the container to verify the installation, taking about one compute-minute.

Afterwards, perform the following two steps to build all the RISC-V processor prototypes (Proteus+DMP, permanent disabling, selective disabling, \secdmp{}, \prospect{}, and \prospect{}+\secdmp{}, for each of the four DMPs) and to compile all benchmarks used in the evaluation.
This takes about 75 compute-minutes.

\begin{enumerate}
    \item Build the prototypes by executing the \verb|build_processor_sims.sh| script;
    \item Compile the benchmarks by running \verb|make|.
\end{enumerate}

\subsection{Major Claims}
The main claims are that \secdmp{} mitigates all DMP secret leakage (C1) while introducing minimal performance overhead (C2) and low hardware cost (C3), improving over the state-of-the-art DMP defenses, \ie permanent disabling and selective disabling, as they incur high performance overhead (C2) and fail to mitigate all DMP leakage (C1), respectively.
Additionally, \secdmp{} can be combined with strong speculative-execution defenses such as \prospect{} to obtain full microarchitectural security against both DMP leakage and Spectre attacks (C1).
Each claim is explained in detail below.

\begin{itemize}
    \item (C1): \secdmp{} mitigates all DMP secret leakage, while the selective disabling defense only prevents architectural DMP leakage.
    Moreover, combining \secdmp{} with \prospect{} prevents leakage from both Spectre-PHT (which \prospect{}'s existing implementation has been evaluated against) and DMPs.
    This claim is supported by (E1) and reported in \cref{tab:security-evaluation}, \cref{sec:evaluation:sec}.
    \item (C2): \secdmp{} (and selective disabling) preserves the performance benefit of the DMPs, whereas permanently disabling them incurs a significant overhead (up to \SI{66.4}{\percent} on our work benchmarks and \SI{59.9}{\percent} on DMP-Bench).
    \secdmp{}'s own overhead stems solely from the secret partitioning method used (at most \SI{4.8}{\percent} and an average of \SI{1.15}{\percent} on the cryptographic benchmarks).
    This is supported by experiments (E2), (E3) and (E4), whose results are reported in \Cref{fig:hacl_results,fig:work_results_disabling,fig:dmp_bench_results}, \cref{sec:evaluation:performance}.
    \item (C3): \secdmp{}'s hardware cost is negligible: it increases the area by at most \SI{0.5}{\percent}, compared to \SI{0.2}{\percent} for selective disabling, and leaves the critical path unchanged at \SI{12.43}{\nano\second}.
    Combining \secdmp{} with \prospect{} increases the area by only \SI{0.3}{\percent}.
    This claim is supported by (E5) and reported in \cref{tab:hardware-evaluation}, \cref{sec:evaluation:hardware}.
\end{itemize}

\subsection{Evaluation}
\label{app:artifact-app-evaluation}

All experiments share the preparation steps described in \cref{app:artifact-app-setup}, which need to be performed only once.
Moreover, each experiment outputs its results to the \texttt{results/} directory.
The artifact includes an \texttt{expected\_results/} directory, which contains our own output, and which should be compared against the generated output in \texttt{results/}.
Since the simulation is deterministic, results obtained with the provided Docker image match \texttt{expected\_results/} exactly, except for the hardware evaluation, whose numbers can fluctuate negligibly between runs.

Experiments (E1)--(E5) can be run individually as described in the artifact's \texttt{README.md}.
Afterwards, follow the \texttt{README.md}'s instructions to generate the paper's figures and tables, which can be compared against the figures and tables provided in \texttt{expected\_results/}.
The artifact additionally includes the \verb|run_secdmp_eval.sh| script, which runs the entire evaluation with one command.
Running all experiments is expected to take around 17 compute-hours.

\subsubsection{Experiment (E1)}
\emph{Security} [5 human-minutes + 20 compute-minutes]:
verifies non-interference of all programs on all prototypes, checking attacker-observable and internal DMP signals (the latter for leakage tests only; see \cref{sec:evaluation:sec}).

\textit{[Execution]} Run \texttt{source .venv/bin/activate \&\& python3 security\_eval.py ; deactivate}.

\textit{[Results]} The \texttt{results/security\_evaluation/} directory contains a table marking each program as secure or leaking, per DMP, for each prototype.
The outcome should match \texttt{expected\_results/} and correspond to \cref{tab:security-evaluation}.

\subsubsection{Experiment (E2)}
\emph{HACL* cryptographic performance} [5 human-minutes + 2 compute-minutes]:
measures the cycle counts of the five HACL* primitives on every prototype and verifies that the secret annotations do not spill.

\textit{[Execution]} Run the command \texttt{python3 performance\_eval\_hacl\_crypto.py} followed by \texttt{./check\_crypto\_partition.sh}, which prints no exceptions if the partitioning is correct.

\textit{[Results]} \path{results/performance_evaluation/hacl_crypto/} reports the cycle counts and relative overheads of every prototype, supporting (C2) and corresponding to \Cref{fig:hacl_results}, \cref{sec:evaluation:performance}.

\subsubsection{Experiment (E3)}
\emph{Work benchmarks} [5 human-minutes + 2 compute-hours]:
measures Embench, CoreMark, AoP, and SpMV on every prototype.

\textit{[Execution]} Run \path{./performance_eval_embench.sh}, \path{./performance_eval_coremark.sh}, \path{./performance_eval_array_of_pointers.sh} and \path{./performance_eval_spmv.sh}.

\textit{[Results]} The corresponding subdirectories of \path{results/performance_evaluation/} support (C2) and correspond to \Cref{fig:work_results_disabling}, \cref{sec:evaluation:performance}.

\subsubsection{Experiment (E4)}
\emph{DMP-Bench} [5 human-minutes + 6 compute-hours]:
measures the 20 work/crypto combinations at four work percentages, for each DMP on every prototype.

\textit{[Execution]} Run the command \texttt{python3 performance\_eval\_dmp\_bench.py}.

\textit{[Results]} \path{results/performance_evaluation/dmp_bench/} supports (C2) and corresponds to \Cref{fig:dmp_bench_results}.

\subsubsection{Experiment (E5)}
\emph{Hardware cost} [5 human-minutes + 8--24 compute-hours, depending on RAM]:
synthesizes all prototypes with EVAL-HD, reporting area and critical path.

\textit{[Execution]} Run \texttt{./hardware\_eval.sh}.

\textit{[Results]} \path{results/hardware_evaluation/} supports (C3) and corresponds to \cref{tab:hardware-evaluation}; this is the only experiment whose numbers fluctuate between runs.
These fluctuations are negligible, as they occur only in the fourth decimal place (\ie at the \SI{e-4}{\milli\meter\squared} level).

%% file: tables/hacl_table.tex
\begin{table*}[t]
  \centering
  \caption{Performance overhead in \% on the HACL* primitives per DMP, relative to Proteus+DMP.}
  \label{tab:hacl-results}
  \begin{tabular}{l c c c c c c c c c c}\toprule
    \multirow{2}{*}{Defense} & \multicolumn{5}{c}{Apple DMP} & \multicolumn{5}{c}{Intel DMP} \\
  \cmidrule(lr){2-6}\cmidrule(lr){7-11}
     & Chacha20 & Salsa20 & SHA2 & HMAC & HKDF & Chacha20 & Salsa20 & SHA2 & HMAC & HKDF \\ \midrule
    Permanent Disabling & \num{+0.0} & \num{+2.9} & \num{+0.1} & \num{+1.1} & \num{+1.2} & \num{+0.0} & \num{+0.0} & \num{+0.0} & \num{+0.0} & \num{+0.0} \\
    Selective Disabling & \num{+0.0} & \num{+2.9} & \num{+0.1} & \num{+1.1} & \num{+1.2} & \num{+0.0} & \num{+0.0} & \num{+0.0} & \num{+0.0} & \num{+0.0} \\
    Partitioning & \num{+1.3} & \num{+4.8} & \num{+0.2} & \num{+0.2} & \num{+1.5} & \num{+1.3} & \num{+1.8} & \num{+0.1} & \num{-0.9} & \num{+0.3} \\
    \secdmp{} & \num{+1.3} & \num{+4.8} & \num{+0.2} & \num{+0.2} & \num{+1.5} & \num{+1.3} & \num{+1.8} & \num{+0.1} & \num{-0.9} & \num{+0.3} \\
    \prospect{} & \num{+1.3} & \num{+4.8} & \num{+0.2} & \num{-0.2} & \num{+1.5} & \num{+1.3} & \num{+1.8} & \num{+0.1} & \num{-1.3} & \num{+0.3} \\
    \prospect{}+\secdmp{} & \num{+1.3} & \num{+4.8} & \num{+0.2} & \num{-0.2} & \num{+1.5} & \num{+1.3} & \num{+1.8} & \num{+0.1} & \num{-1.3} & \num{+0.3} \\
    \midrule
    \multirow{2}{*}{Defense} & \multicolumn{5}{c}{CDP} & \multicolumn{5}{c}{IMP} \\
  \cmidrule(lr){2-6}\cmidrule(lr){7-11}
     & Chacha20 & Salsa20 & SHA2 & HMAC & HKDF & Chacha20 & Salsa20 & SHA2 & HMAC & HKDF \\ \midrule
    Permanent Disabling & \num{+0.0} & \num{+2.9} & \num{+0.2} & \num{+1.1} & \num{+1.2} & \num{+0.0} & \num{+0.0} & \num{+0.0} & \num{+0.0} & \num{+0.0} \\
    Selective Disabling & \num{+0.0} & \num{+2.9} & \num{+0.1} & \num{+1.1} & \num{+1.2} & \num{+0.0} & \num{+0.0} & \num{+0.0} & \num{+0.0} & \num{+0.0} \\
    Partitioning & \num{+1.2} & \num{+4.8} & \num{+1.8} & \num{+0.6} & \num{+1.7} & \num{+1.3} & \num{+1.8} & \num{+0.1} & \num{-0.9} & \num{+0.3} \\
    \secdmp{} & \num{+1.3} & \num{+4.8} & \num{+1.7} & \num{+0.3} & \num{+1.6} & \num{+1.3} & \num{+1.8} & \num{+0.1} & \num{-0.9} & \num{+0.3} \\
    \prospect{} & \num{+1.3} & \num{+4.8} & \num{+1.8} & \num{+0.2} & \num{+1.6} & \num{+1.3} & \num{+1.8} & \num{+0.1} & \num{-1.3} & \num{+0.3} \\
    \prospect{}+\secdmp{} & \num{+1.3} & \num{+4.8} & \num{+1.7} & \num{-0.1} & \num{+1.6} & \num{+1.3} & \num{+1.8} & \num{+0.1} & \num{-1.3} & \num{+0.3} \\
    \bottomrule
  \end{tabular}
\end{table*}

%% file: tables/dmp_bench_table.tex
\begin{table*}[t]
  \centering
  \caption{Performance overhead in \% on DMP-Bench for the Apple DMP, relative to Proteus+DMP.}
  \label{tab:dmp-bench-AppleDmp}
  \setlength{\tabcolsep}{1.9pt}
  \begin{tabular}{l S[table-format=+2.1] S[table-format=+2.1] S[table-format=+2.1] S[table-format=+2.1] S[table-format=+2.1] S[table-format=+2.1] S[table-format=+2.1] S[table-format=+2.1] S[table-format=+2.1] S[table-format=+2.1] S[table-format=+2.1] S[table-format=+2.1] S[table-format=+2.1] S[table-format=+2.1] S[table-format=+2.1] S[table-format=+2.1]}\toprule
    \multirow{2}{*}{Defense} & \multicolumn{4}{c}{Sglib} & \multicolumn{4}{c}{CoreMark} & \multicolumn{4}{c}{AoP} & \multicolumn{4}{c}{SpMV} \\
  \cmidrule(lr){2-5}\cmidrule(lr){6-9}\cmidrule(lr){10-13}\cmidrule(lr){14-17}
     & \multicolumn{1}{c}{25\%\,W} & \multicolumn{1}{c}{50\%\,W} & \multicolumn{1}{c}{75\%\,W} & \multicolumn{1}{c}{90\%\,W} & \multicolumn{1}{c}{25\%\,W} & \multicolumn{1}{c}{50\%\,W} & \multicolumn{1}{c}{75\%\,W} & \multicolumn{1}{c}{90\%\,W} & \multicolumn{1}{c}{25\%\,W} & \multicolumn{1}{c}{50\%\,W} & \multicolumn{1}{c}{75\%\,W} & \multicolumn{1}{c}{90\%\,W} & \multicolumn{1}{c}{25\%\,W} & \multicolumn{1}{c}{50\%\,W} & \multicolumn{1}{c}{75\%\,W} & \multicolumn{1}{c}{90\%\,W} \\ \midrule
    \multicolumn{17}{l}{\emph{Chacha20}} \\
    Permanent Disabling & +1.7 & +3.3 & +4.8 & +5.8 & +4.7 & +9.6 & +14.4 & +17.4 & +10.3 & +22.1 & +31.8 & +37.2 & +0.0 & +0.0 & +0.0 & +0.0 \\
    Selective Disabling & +0.0 & +0.0 & +0.0 & +0.0 & +0.0 & +0.0 & +0.0 & +0.0 & +0.0 & +0.0 & +0.0 & +0.0 & +0.0 & +0.0 & +0.0 & +0.0 \\
    Partitioning & +0.9 & +0.6 & +0.3 & +0.2 & +1.0 & +0.7 & +0.3 & +0.1 & +1.5 & +0.8 & +0.2 & -0.2 & -2.5 & -1.8 & -1.1 & -0.7 \\
    \secdmp{} & +0.9 & +0.6 & +0.3 & +0.2 & +1.0 & +0.7 & +0.3 & +0.1 & +1.5 & +0.8 & +0.2 & -0.2 & -2.5 & -1.8 & -1.1 & -0.7 \\
    \prospect{} & +1.2 & +1.0 & +0.6 & +0.5 & +1.0 & +0.8 & +0.5 & +0.4 & +1.4 & +0.7 & +0.2 & -0.2 & -2.6 & -1.9 & -1.1 & -0.7 \\
    \prospect{}+\secdmp{} & +1.2 & +1.0 & +0.6 & +0.5 & +1.0 & +0.8 & +0.5 & +0.4 & +1.4 & +0.7 & +0.2 & -0.2 & -2.6 & -1.9 & -1.1 & -0.7 \\
    \midrule
    \multicolumn{17}{l}{\emph{Salsa20}} \\
    Permanent Disabling & +3.5 & +4.6 & +5.6 & +5.9 & +6.1 & +10.5 & +14.9 & +17.5 & +11.7 & +20.2 & +30.3 & +36.8 & +1.3 & +0.9 & +0.5 & +0.2 \\
    Selective Disabling & +1.8 & +1.2 & +0.6 & +0.2 & +1.3 & +0.9 & +0.4 & +0.2 & +1.8 & +1.3 & +0.7 & +0.3 & +1.3 & +0.9 & +0.4 & +0.2 \\
    Partitioning & +2.9 & +1.9 & +1.0 & +0.1 & +3.1 & +2.0 & +1.0 & +0.3 & +2.5 & +1.9 & +1.0 & +0.5 & +2.6 & +2.0 & +1.5 & +1.0 \\
    \secdmp{} & +3.0 & +2.1 & +1.0 & +0.2 & +3.1 & +2.0 & +1.0 & +0.3 & +2.5 & +1.9 & +1.0 & +0.5 & +2.6 & +2.0 & +1.5 & +1.0 \\
    \prospect{} & +3.1 & +2.2 & +1.2 & +0.4 & +3.2 & +2.2 & +1.2 & +0.6 & +2.5 & +1.9 & +1.0 & +0.5 & +3.0 & +2.3 & +1.7 & +1.1 \\
    \prospect{}+\secdmp{} & +3.1 & +2.2 & +1.2 & +0.4 & +3.2 & +2.2 & +1.2 & +0.6 & +2.5 & +1.9 & +1.0 & +0.5 & +3.0 & +2.3 & +1.7 & +1.1 \\
    \midrule
    \multicolumn{17}{l}{\emph{SHA2}} \\
    Permanent Disabling & +2.6 & +4.2 & +5.5 & +6.3 & +5.9 & +10.4 & +14.8 & +17.5 & +11.6 & +21.5 & +31.2 & +37.0 & +0.9 & +0.6 & +0.3 & +0.1 \\
    Selective Disabling & +0.9 & +0.6 & +0.3 & +0.1 & +1.0 & +0.7 & +0.3 & +0.1 & +0.9 & +0.6 & +0.3 & +0.1 & +0.9 & +0.6 & +0.3 & +0.1 \\
    Partitioning & +0.0 & +0.1 & +0.1 & +0.2 & +0.8 & +0.5 & +0.2 & +0.0 & +0.0 & -0.1 & -0.3 & -0.3 & +0.1 & +0.1 & +0.1 & +0.1 \\
    \secdmp{} & +0.0 & +0.1 & +0.1 & +0.2 & +0.8 & +0.5 & +0.2 & +0.0 & +0.0 & -0.1 & -0.3 & -0.3 & +0.1 & +0.1 & +0.1 & +0.1 \\
    \prospect{} & +0.1 & +0.3 & +0.4 & +0.5 & +1.1 & +0.8 & +0.5 & +0.4 & +0.1 & +0.0 & -0.2 & -0.3 & +0.1 & +0.1 & +0.1 & +0.1 \\
    \prospect{}+\secdmp{} & +0.1 & +0.3 & +0.4 & +0.5 & +1.1 & +0.8 & +0.5 & +0.4 & +0.1 & +0.0 & -0.2 & -0.3 & +0.1 & +0.1 & +0.1 & +0.1 \\
    \midrule
    \multicolumn{17}{l}{\emph{HMAC}} \\
    Permanent Disabling & +2.9 & +4.2 & +5.5 & +6.3 & +6.0 & +10.3 & +14.8 & +17.5 & +11.9 & +21.8 & +31.6 & +37.4 & +1.1 & +0.7 & +0.4 & +0.2 \\
    Selective Disabling & +1.2 & +0.8 & +0.4 & +0.2 & +1.2 & +0.8 & +0.4 & +0.2 & +1.2 & +0.7 & +0.4 & +0.2 & +1.1 & +0.7 & +0.4 & +0.1 \\
    Partitioning & +1.8 & +1.3 & +0.9 & +0.7 & +0.8 & +0.7 & +0.6 & +0.5 & +1.2 & +0.9 & +0.7 & +0.5 & +0.7 & +0.6 & +0.5 & +0.5 \\
    \secdmp{} & +1.8 & +1.3 & +0.9 & +0.7 & +0.8 & +0.7 & +0.6 & +0.5 & +1.2 & +0.9 & +0.7 & +0.5 & +0.7 & +0.6 & +0.5 & +0.5 \\
    \prospect{} & +1.5 & +0.8 & +0.5 & +0.2 & +0.8 & +0.8 & +0.8 & +0.8 & +1.2 & +0.9 & +0.6 & +0.5 & +0.8 & +0.7 & +0.6 & +0.5 \\
    \prospect{}+\secdmp{} & +1.5 & +0.8 & +0.5 & +0.2 & +0.8 & +0.8 & +0.8 & +0.8 & +1.2 & +0.9 & +0.6 & +0.5 & +0.8 & +0.7 & +0.6 & +0.5 \\
    \midrule
    \multicolumn{17}{l}{\emph{HKDF}} \\
    Permanent Disabling & +2.4 & +3.6 & +4.7 & +5.5 & +6.0 & +10.2 & +14.7 & +17.4 & +12.9 & +23.3 & +31.4 & +37.4 & +1.2 & +0.8 & +0.4 & +0.2 \\
    Selective Disabling & +0.9 & +0.7 & +0.3 & +0.1 & +1.2 & +0.8 & +0.4 & +0.2 & +0.9 & +0.6 & +0.3 & +0.1 & +1.2 & +0.8 & +0.4 & +0.2 \\
    Partitioning & +0.7 & +0.6 & +0.3 & +0.1 & +0.4 & +0.2 & +0.1 & +0.0 & +1.0 & +0.8 & +0.6 & +0.5 & +0.2 & +0.1 & +0.1 & +0.0 \\
    \secdmp{} & +0.7 & +0.6 & +0.3 & +0.1 & +0.4 & +0.2 & +0.1 & +0.0 & +1.0 & +0.8 & +0.6 & +0.5 & +0.2 & +0.1 & +0.1 & +0.0 \\
    \prospect{} & +0.4 & +0.1 & -0.2 & -0.4 & +0.4 & +0.4 & +0.3 & +0.3 & +1.0 & +0.7 & +0.6 & +0.5 & -0.1 & -0.1 & -0.1 & +0.0 \\
    \prospect{}+\secdmp{} & +0.4 & +0.1 & -0.2 & -0.4 & +0.4 & +0.4 & +0.3 & +0.3 & +1.0 & +0.7 & +0.6 & +0.5 & -0.1 & -0.1 & -0.1 & +0.0 \\
    \bottomrule
  \end{tabular}
\end{table*}

\begin{table*}[t]
  \centering
  \caption{Performance overhead in \% on DMP-Bench for the Intel DMP, relative to Proteus+DMP.}
  \label{tab:dmp-bench-IntelDmp}
  \setlength{\tabcolsep}{1.9pt}
  \begin{tabular}{l S[table-format=+2.1] S[table-format=+2.1] S[table-format=+2.1] S[table-format=+2.1] S[table-format=+2.1] S[table-format=+2.1] S[table-format=+2.1] S[table-format=+2.1] S[table-format=+2.1] S[table-format=+2.1] S[table-format=+2.1] S[table-format=+2.1] S[table-format=+2.1] S[table-format=+2.1] S[table-format=+2.1] S[table-format=+2.1]}\toprule
    \multirow{2}{*}{Defense} & \multicolumn{4}{c}{Sglib} & \multicolumn{4}{c}{CoreMark} & \multicolumn{4}{c}{AoP} & \multicolumn{4}{c}{SpMV} \\
  \cmidrule(lr){2-5}\cmidrule(lr){6-9}\cmidrule(lr){10-13}\cmidrule(lr){14-17}
     & \multicolumn{1}{c}{25\%\,W} & \multicolumn{1}{c}{50\%\,W} & \multicolumn{1}{c}{75\%\,W} & \multicolumn{1}{c}{90\%\,W} & \multicolumn{1}{c}{25\%\,W} & \multicolumn{1}{c}{50\%\,W} & \multicolumn{1}{c}{75\%\,W} & \multicolumn{1}{c}{90\%\,W} & \multicolumn{1}{c}{25\%\,W} & \multicolumn{1}{c}{50\%\,W} & \multicolumn{1}{c}{75\%\,W} & \multicolumn{1}{c}{90\%\,W} & \multicolumn{1}{c}{25\%\,W} & \multicolumn{1}{c}{50\%\,W} & \multicolumn{1}{c}{75\%\,W} & \multicolumn{1}{c}{90\%\,W} \\ \midrule
    \multicolumn{17}{l}{\emph{Chacha20}} \\
    Permanent Disabling & +0.0 & +0.0 & +0.0 & +0.0 & +0.0 & +0.1 & +0.1 & +0.1 & +16.6 & +33.1 & +49.7 & +59.6 & +0.0 & +0.0 & +0.0 & +0.0 \\
    Selective Disabling & +0.0 & +0.0 & +0.0 & +0.0 & +0.0 & +0.0 & +0.0 & +0.0 & +0.0 & +0.0 & +0.0 & +0.0 & +0.0 & +0.0 & +0.0 & +0.0 \\
    Partitioning & +1.0 & +0.8 & +0.6 & +0.4 & +1.1 & +0.7 & +0.3 & +0.1 & +1.6 & +1.1 & +0.6 & +0.2 & -2.5 & -1.8 & -1.2 & -0.7 \\
    \secdmp{} & +1.0 & +0.8 & +0.6 & +0.4 & +1.1 & +0.7 & +0.3 & +0.1 & +1.6 & +1.1 & +0.6 & +0.2 & -2.5 & -1.8 & -1.2 & -0.7 \\
    \prospect{} & +1.2 & +1.0 & +0.8 & +0.7 & +1.0 & +0.6 & +0.3 & +0.0 & +1.5 & +1.0 & +0.6 & +0.2 & -2.6 & -1.9 & -1.2 & -0.7 \\
    \prospect{}+\secdmp{} & +1.2 & +1.0 & +0.8 & +0.7 & +1.0 & +0.6 & +0.3 & +0.0 & +1.6 & +1.1 & +0.6 & +0.2 & -2.6 & -1.9 & -1.2 & -0.7 \\
    \midrule
    \multicolumn{17}{l}{\emph{Salsa20}} \\
    Permanent Disabling & +0.0 & -0.1 & -0.1 & -0.1 & +0.0 & +0.0 & +0.1 & +0.1 & +17.6 & +34.2 & +50.4 & +59.9 & +0.0 & +0.0 & +0.0 & +0.0 \\
    Selective Disabling & +0.0 & +0.0 & +0.0 & +0.0 & +0.0 & +0.0 & +0.0 & +0.0 & -0.1 & +0.0 & +0.0 & +0.0 & +0.0 & +0.0 & +0.0 & +0.0 \\
    Partitioning & +1.2 & +0.8 & +0.5 & +0.3 & +1.8 & +1.2 & +0.6 & +0.2 & +1.0 & +0.7 & +0.3 & +0.2 & +1.3 & +1.1 & +0.9 & +0.7 \\
    \secdmp{} & +1.2 & +0.8 & +0.5 & +0.3 & +1.8 & +1.2 & +0.6 & +0.2 & +1.0 & +0.7 & +0.3 & +0.2 & +1.3 & +1.1 & +0.9 & +0.7 \\
    \prospect{} & +1.4 & +1.1 & +0.8 & +0.6 & +1.9 & +1.2 & +0.6 & +0.2 & +1.0 & +0.7 & +0.3 & +0.2 & +1.6 & +1.3 & +1.0 & +0.8 \\
    \prospect{}+\secdmp{} & +1.4 & +1.1 & +0.8 & +0.6 & +1.9 & +1.2 & +0.6 & +0.2 & +1.0 & +0.7 & +0.3 & +0.2 & +1.6 & +1.3 & +1.0 & +0.8 \\
    \midrule
    \multicolumn{17}{l}{\emph{SHA2}} \\
    Permanent Disabling & +0.0 & +0.1 & +0.1 & +0.1 & +0.0 & +0.0 & +0.1 & +0.1 & +17.2 & +30.6 & +47.2 & +57.6 & +0.0 & +0.0 & +0.0 & +0.0 \\
    Selective Disabling & +0.0 & +0.0 & +0.0 & +0.0 & +0.0 & +0.0 & +0.0 & +0.0 & +0.0 & +0.0 & +0.0 & +0.0 & +0.0 & +0.0 & +0.0 & +0.0 \\
    Partitioning & -0.9 & -0.5 & -0.2 & +0.0 & -0.2 & -0.2 & -0.2 & -0.1 & -1.1 & -0.9 & -0.9 & -0.9 & -0.9 & -0.6 & -0.2 & -0.1 \\
    \secdmp{} & -0.9 & -0.5 & -0.2 & +0.0 & -0.2 & -0.2 & -0.2 & -0.1 & -1.1 & -0.9 & -0.9 & -0.9 & -0.9 & -0.6 & -0.2 & -0.1 \\
    \prospect{} & -0.9 & -0.5 & -0.2 & +0.0 & +0.0 & -0.1 & -0.1 & -0.1 & -0.9 & -0.8 & -0.8 & -0.9 & -0.9 & -0.6 & -0.3 & -0.1 \\
    \prospect{}+\secdmp{} & -0.9 & -0.5 & -0.2 & +0.0 & +0.0 & -0.1 & -0.1 & -0.1 & -0.9 & -0.8 & -0.8 & -0.9 & -0.9 & -0.6 & -0.3 & -0.1 \\
    \midrule
    \multicolumn{17}{l}{\emph{HMAC}} \\
    Permanent Disabling & +0.0 & +0.0 & -0.1 & -0.1 & +0.0 & +0.0 & +0.0 & +0.1 & +15.1 & +31.0 & +48.0 & +59.4 & +0.0 & +0.0 & +0.0 & +0.0 \\
    Selective Disabling & +0.0 & +0.0 & +0.0 & +0.0 & +0.0 & +0.0 & +0.0 & +0.0 & +0.0 & +0.0 & +0.0 & +0.0 & +0.0 & +0.0 & +0.0 & +0.0 \\
    Partitioning & +0.5 & +0.3 & +0.2 & +0.1 & -0.4 & -0.2 & +0.1 & +0.2 & +0.0 & +0.0 & +0.0 & +0.0 & -0.4 & -0.1 & +0.2 & +0.3 \\
    \secdmp{} & +0.5 & +0.3 & +0.2 & +0.1 & -0.4 & -0.2 & +0.1 & +0.2 & +0.0 & +0.0 & +0.0 & +0.0 & -0.4 & -0.1 & +0.2 & +0.3 \\
    \prospect{} & +0.4 & +0.3 & +0.2 & +0.1 & -0.4 & -0.1 & +0.2 & +0.3 & -0.1 & +0.0 & +0.0 & +0.0 & -0.3 & +0.0 & +0.2 & +0.4 \\
    \prospect{}+\secdmp{} & +0.4 & +0.3 & +0.2 & +0.1 & -0.4 & -0.1 & +0.2 & +0.3 & -0.1 & +0.0 & +0.0 & +0.0 & -0.3 & +0.0 & +0.2 & +0.4 \\
    \midrule
    \multicolumn{17}{l}{\emph{HKDF}} \\
    Permanent Disabling & +0.0 & +0.0 & +0.0 & +0.0 & +0.0 & +0.0 & +0.1 & +0.1 & +16.9 & +33.6 & +50.0 & +59.4 & +0.0 & +0.0 & +0.0 & +0.0 \\
    Selective Disabling & +0.0 & +0.0 & +0.0 & +0.0 & +0.0 & +0.0 & +0.0 & +0.0 & +0.0 & +0.0 & +0.0 & +0.0 & +0.0 & +0.0 & +0.0 & +0.0 \\
    Partitioning & -0.1 & +0.0 & +0.2 & +0.4 & -0.8 & -0.6 & -0.3 & -0.1 & +0.1 & +0.4 & +0.7 & +0.8 & -1.0 & -0.7 & -0.3 & -0.2 \\
    \secdmp{} & -0.1 & +0.0 & +0.2 & +0.4 & -0.8 & -0.6 & -0.3 & -0.1 & +0.1 & +0.4 & +0.7 & +0.8 & -1.0 & -0.7 & -0.3 & -0.2 \\
    \prospect{} & -0.3 & -0.1 & +0.0 & +0.1 & -0.8 & -0.5 & -0.2 & -0.1 & +0.1 & +0.4 & +0.7 & +0.8 & -1.3 & -0.9 & -0.5 & -0.2 \\
    \prospect{}+\secdmp{} & -0.3 & -0.1 & +0.0 & +0.1 & -0.8 & -0.5 & -0.2 & -0.1 & +0.1 & +0.4 & +0.7 & +0.8 & -1.3 & -0.9 & -0.5 & -0.2 \\
    \bottomrule
  \end{tabular}
\end{table*}

\begin{table*}[t]
  \centering
  \caption{Performance overhead in \% on DMP-Bench for CDP, relative to Proteus+DMP.}
  \label{tab:dmp-bench-Cdp}
  \setlength{\tabcolsep}{1.9pt}
  \begin{tabular}{l S[table-format=+2.1] S[table-format=+2.1] S[table-format=+2.1] S[table-format=+2.1] S[table-format=+2.1] S[table-format=+2.1] S[table-format=+2.1] S[table-format=+2.1] S[table-format=+2.1] S[table-format=+2.1] S[table-format=+2.1] S[table-format=+2.1] S[table-format=+2.1] S[table-format=+2.1] S[table-format=+2.1] S[table-format=+2.1]}\toprule
    \multirow{2}{*}{Defense} & \multicolumn{4}{c}{Sglib} & \multicolumn{4}{c}{CoreMark} & \multicolumn{4}{c}{AoP} & \multicolumn{4}{c}{SpMV} \\
  \cmidrule(lr){2-5}\cmidrule(lr){6-9}\cmidrule(lr){10-13}\cmidrule(lr){14-17}
     & \multicolumn{1}{c}{25\%\,W} & \multicolumn{1}{c}{50\%\,W} & \multicolumn{1}{c}{75\%\,W} & \multicolumn{1}{c}{90\%\,W} & \multicolumn{1}{c}{25\%\,W} & \multicolumn{1}{c}{50\%\,W} & \multicolumn{1}{c}{75\%\,W} & \multicolumn{1}{c}{90\%\,W} & \multicolumn{1}{c}{25\%\,W} & \multicolumn{1}{c}{50\%\,W} & \multicolumn{1}{c}{75\%\,W} & \multicolumn{1}{c}{90\%\,W} & \multicolumn{1}{c}{25\%\,W} & \multicolumn{1}{c}{50\%\,W} & \multicolumn{1}{c}{75\%\,W} & \multicolumn{1}{c}{90\%\,W} \\ \midrule
    \multicolumn{17}{l}{\emph{Chacha20}} \\
    Permanent Disabling & +2.0 & +3.7 & +5.1 & +5.8 & +5.1 & +10.4 & +15.4 & +18.5 & +10.3 & +22.1 & +31.8 & +37.2 & +0.1 & +0.1 & +0.1 & +0.2 \\
    Selective Disabling & +0.1 & +0.1 & +0.0 & +0.0 & +0.0 & +0.0 & +0.0 & +0.0 & +0.0 & +0.1 & +0.0 & +0.0 & +0.0 & +0.0 & +0.0 & +0.0 \\
    Partitioning & +1.4 & +1.4 & +1.0 & +0.6 & +1.1 & +0.8 & +0.3 & +0.2 & +1.6 & +0.9 & +0.3 & -0.1 & -2.5 & -1.9 & -1.2 & -0.8 \\
    \secdmp{} & +1.4 & +1.4 & +1.0 & +0.5 & +1.1 & +0.8 & +0.3 & +0.2 & +1.5 & +0.8 & +0.3 & -0.1 & -2.5 & -1.9 & -1.2 & -0.8 \\
    \prospect{} & +1.5 & +1.4 & +1.2 & +0.7 & +1.1 & +0.9 & +0.6 & +0.5 & +1.5 & +0.9 & +0.3 & -0.1 & -2.6 & -1.9 & -1.2 & -0.8 \\
    \prospect{}+\secdmp{} & +1.5 & +1.4 & +1.2 & +0.7 & +1.1 & +0.9 & +0.6 & +0.5 & +1.5 & +0.8 & +0.3 & -0.1 & -2.6 & -1.9 & -1.2 & -0.8 \\
    \midrule
    \multicolumn{17}{l}{\emph{Salsa20}} \\
    Permanent Disabling & +3.7 & +4.7 & +5.5 & +6.0 & +6.6 & +11.2 & +15.8 & +18.6 & +12.3 & +20.6 & +30.4 & +36.8 & +1.5 & +1.1 & +0.7 & +0.5 \\
    Selective Disabling & +1.9 & +1.3 & +0.7 & +0.3 & +1.5 & +1.0 & +0.5 & +0.2 & +2.0 & +1.4 & +0.7 & +0.3 & +1.5 & +1.0 & +0.5 & +0.2 \\
    Partitioning & +3.1 & +2.2 & +1.1 & +0.6 & +3.3 & +2.2 & +1.0 & +0.4 & +3.1 & +2.3 & +1.2 & +0.6 & +2.8 & +2.3 & +1.7 & +1.4 \\
    \secdmp{} & +3.1 & +2.2 & +1.2 & +0.6 & +3.3 & +2.2 & +1.1 & +0.4 & +3.1 & +2.2 & +1.2 & +0.6 & +2.9 & +2.3 & +1.7 & +1.4 \\
    \prospect{} & +3.4 & +2.5 & +1.4 & +0.7 & +3.4 & +2.4 & +1.3 & +0.7 & +3.1 & +2.3 & +1.2 & +0.6 & +3.2 & +2.5 & +1.8 & +1.4 \\
    \prospect{}+\secdmp{} & +3.4 & +2.5 & +1.4 & +0.7 & +3.4 & +2.4 & +1.3 & +0.7 & +3.1 & +2.2 & +1.2 & +0.6 & +3.2 & +2.5 & +1.8 & +1.4 \\
    \midrule
    \multicolumn{17}{l}{\emph{SHA2}} \\
    Permanent Disabling & +2.6 & +4.2 & +5.2 & +5.9 & +6.5 & +11.2 & +15.8 & +18.5 & +11.7 & +21.7 & +31.4 & +37.0 & +1.1 & +0.8 & +0.5 & +0.4 \\
    Selective Disabling & +0.8 & +0.6 & +0.3 & +0.1 & +1.3 & +0.9 & +0.4 & +0.2 & +1.0 & +0.7 & +0.4 & +0.1 & +1.0 & +0.7 & +0.3 & +0.1 \\
    Partitioning & +1.0 & +1.0 & +0.8 & +0.5 & +1.2 & +0.7 & +0.2 & -0.1 & +0.7 & +0.3 & +0.0 & -0.2 & +1.4 & +0.9 & +0.4 & +0.1 \\
    \secdmp{} & +0.6 & +0.7 & +0.6 & +0.2 & +1.0 & +0.6 & +0.2 & -0.1 & +0.2 & +0.3 & +0.0 & -0.2 & +0.1 & +0.1 & +0.0 & +0.0 \\
    \prospect{} & +1.0 & +1.1 & +0.6 & +0.7 & +1.4 & +1.0 & +0.6 & +0.4 & +0.8 & +0.4 & +0.1 & -0.2 & +1.4 & +0.9 & +0.5 & +0.2 \\
    \prospect{}+\secdmp{} & +0.7 & +0.9 & +0.4 & +0.6 & +1.3 & +1.0 & +0.6 & +0.4 & +0.4 & +0.4 & +0.0 & -0.2 & +0.2 & +0.1 & +0.1 & +0.1 \\
    \midrule
    \multicolumn{17}{l}{\emph{HMAC}} \\
    Permanent Disabling & +3.2 & +4.4 & +5.4 & +6.0 & +6.3 & +10.9 & +15.6 & +18.5 & +12.2 & +22.1 & +31.6 & +37.4 & +0.8 & +0.6 & +0.5 & +0.4 \\
    Selective Disabling & +1.4 & +1.0 & +0.5 & +0.2 & +1.2 & +0.8 & +0.4 & +0.2 & +1.5 & +1.0 & +0.4 & +0.2 & +0.7 & +0.5 & +0.2 & +0.1 \\
    Partitioning & +2.2 & +1.7 & +1.0 & +0.7 & +1.9 & +1.5 & +1.0 & +0.7 & +1.7 & +1.2 & +0.7 & +0.5 & +0.5 & +0.5 & +0.6 & +0.6 \\
    \secdmp{} & +2.2 & +1.7 & +1.0 & +0.6 & +1.6 & +1.3 & +0.9 & +0.7 & +1.5 & +1.2 & +0.7 & +0.5 & +0.4 & +0.5 & +0.6 & +0.6 \\
    \prospect{} & +2.1 & +1.6 & +1.1 & +0.6 & +2.1 & +1.7 & +1.4 & +1.1 & +1.6 & +1.1 & +0.6 & +0.5 & +0.6 & +0.6 & +0.6 & +0.6 \\
    \prospect{}+\secdmp{} & +2.0 & +1.5 & +1.0 & +0.6 & +1.7 & +1.4 & +1.2 & +1.1 & +1.5 & +1.1 & +0.6 & +0.5 & +0.5 & +0.5 & +0.5 & +0.5 \\
    \midrule
    \multicolumn{17}{l}{\emph{HKDF}} \\
    Permanent Disabling & +2.5 & +3.7 & +4.6 & +5.3 & +6.6 & +11.0 & +15.6 & +18.4 & +12.8 & +23.1 & +31.3 & +37.3 & +1.5 & +1.1 & +0.7 & +0.4 \\
    Selective Disabling & +0.9 & +0.7 & +0.4 & +0.1 & +1.5 & +1.0 & +0.5 & +0.2 & +0.9 & +0.5 & +0.3 & +0.1 & +1.4 & +0.9 & +0.5 & +0.2 \\
    Partitioning & +1.5 & +1.1 & +0.4 & +0.2 & +1.3 & +0.9 & +0.4 & +0.1 & +1.1 & +0.7 & +0.6 & +0.5 & +1.0 & +0.6 & +0.3 & +0.1 \\
    \secdmp{} & +1.2 & +0.9 & +0.3 & +0.2 & +0.9 & +0.6 & +0.3 & +0.1 & +1.0 & +0.6 & +0.5 & +0.4 & +0.9 & +0.6 & +0.3 & +0.1 \\
    \prospect{} & +1.6 & +1.3 & +0.6 & +0.0 & +1.5 & +1.1 & +0.8 & +0.6 & +1.1 & +0.7 & +0.5 & +0.4 & +0.7 & +0.4 & +0.2 & +0.0 \\
    \prospect{}+\secdmp{} & +1.2 & +1.0 & +0.5 & +0.0 & +1.0 & +0.8 & +0.6 & +0.5 & +0.9 & +0.6 & +0.5 & +0.4 & +0.5 & +0.3 & +0.1 & +0.0 \\
    \bottomrule
  \end{tabular}
\end{table*}

\begin{table*}[t]
  \centering
  \caption{Performance overhead in \% on DMP-Bench for IMP, relative to Proteus+DMP.}
  \label{tab:dmp-bench-Imp}
  \setlength{\tabcolsep}{1.9pt}
  \begin{tabular}{l S[table-format=+2.1] S[table-format=+2.1] S[table-format=+2.1] S[table-format=+2.1] S[table-format=+2.1] S[table-format=+2.1] S[table-format=+2.1] S[table-format=+2.1] S[table-format=+2.1] S[table-format=+2.1] S[table-format=+2.1] S[table-format=+2.1] S[table-format=+2.1] S[table-format=+2.1] S[table-format=+2.1] S[table-format=+2.1]}\toprule
    \multirow{2}{*}{Defense} & \multicolumn{4}{c}{Sglib} & \multicolumn{4}{c}{CoreMark} & \multicolumn{4}{c}{AoP} & \multicolumn{4}{c}{SpMV} \\
  \cmidrule(lr){2-5}\cmidrule(lr){6-9}\cmidrule(lr){10-13}\cmidrule(lr){14-17}
     & \multicolumn{1}{c}{25\%\,W} & \multicolumn{1}{c}{50\%\,W} & \multicolumn{1}{c}{75\%\,W} & \multicolumn{1}{c}{90\%\,W} & \multicolumn{1}{c}{25\%\,W} & \multicolumn{1}{c}{50\%\,W} & \multicolumn{1}{c}{75\%\,W} & \multicolumn{1}{c}{90\%\,W} & \multicolumn{1}{c}{25\%\,W} & \multicolumn{1}{c}{50\%\,W} & \multicolumn{1}{c}{75\%\,W} & \multicolumn{1}{c}{90\%\,W} & \multicolumn{1}{c}{25\%\,W} & \multicolumn{1}{c}{50\%\,W} & \multicolumn{1}{c}{75\%\,W} & \multicolumn{1}{c}{90\%\,W} \\ \midrule
    \multicolumn{17}{l}{\emph{Chacha20}} \\
    Permanent Disabling & +0.0 & +0.0 & +0.0 & +0.1 & +0.0 & +0.0 & +0.0 & +0.0 & +0.0 & +0.0 & +0.0 & +0.0 & +5.2 & +10.4 & +15.6 & +18.7 \\
    Selective Disabling & +0.0 & +0.0 & +0.0 & +0.0 & +0.0 & +0.0 & +0.0 & +0.0 & +0.0 & +0.0 & +0.0 & +0.0 & +0.0 & +0.0 & +0.0 & +0.0 \\
    Partitioning & +1.0 & +0.8 & +0.5 & +0.4 & +1.1 & +0.7 & +0.3 & +0.1 & +1.6 & +1.0 & +0.5 & +0.2 & -2.6 & -2.0 & -1.3 & -0.9 \\
    \secdmp{} & +1.0 & +0.8 & +0.5 & +0.4 & +1.1 & +0.7 & +0.3 & +0.1 & +1.6 & +1.0 & +0.5 & +0.2 & -2.6 & -2.0 & -1.3 & -0.9 \\
    \prospect{} & +1.2 & +1.0 & +0.8 & +0.6 & +1.0 & +0.6 & +0.3 & +0.1 & +1.5 & +1.0 & +0.5 & +0.2 & -2.6 & -2.0 & -1.3 & -0.9 \\
    \prospect{}+\secdmp{} & +1.2 & +1.0 & +0.8 & +0.6 & +1.0 & +0.6 & +0.3 & +0.1 & +1.5 & +1.0 & +0.5 & +0.2 & -2.6 & -2.0 & -1.3 & -0.9 \\
    \midrule
    \multicolumn{17}{l}{\emph{Salsa20}} \\
    Permanent Disabling & +0.0 & +0.0 & +0.0 & +0.1 & +0.0 & +0.0 & +0.0 & +0.0 & +0.0 & +0.0 & +0.0 & +0.0 & +5.3 & +10.6 & +15.9 & +19.1 \\
    Selective Disabling & +0.0 & +0.0 & +0.0 & +0.0 & +0.0 & +0.0 & +0.0 & +0.0 & +0.0 & +0.0 & +0.0 & +0.0 & +0.0 & +0.0 & +0.0 & +0.0 \\
    Partitioning & +1.2 & +0.9 & +0.5 & +0.3 & +1.8 & +1.2 & +0.6 & +0.2 & +0.7 & +0.6 & +0.3 & +0.2 & +0.9 & +0.4 & -0.1 & -0.4 \\
    \secdmp{} & +1.2 & +0.9 & +0.5 & +0.3 & +1.8 & +1.2 & +0.6 & +0.2 & +0.7 & +0.6 & +0.3 & +0.2 & +0.9 & +0.4 & -0.1 & -0.4 \\
    \prospect{} & +1.4 & +1.2 & +0.9 & +0.8 & +1.9 & +1.2 & +0.6 & +0.2 & +0.7 & +0.6 & +0.3 & +0.2 & +1.3 & +0.7 & +0.1 & -0.2 \\
    \prospect{}+\secdmp{} & +1.4 & +1.2 & +0.9 & +0.8 & +1.9 & +1.2 & +0.6 & +0.2 & +0.7 & +0.6 & +0.3 & +0.2 & +1.3 & +0.7 & +0.1 & -0.2 \\
    \midrule
    \multicolumn{17}{l}{\emph{SHA2}} \\
    Permanent Disabling & +0.0 & +0.0 & +0.1 & +0.2 & +0.0 & +0.0 & +0.0 & +0.0 & +0.0 & +0.0 & +0.0 & +0.0 & +5.3 & +10.6 & +15.9 & +19.1 \\
    Selective Disabling & +0.0 & +0.0 & +0.0 & +0.0 & +0.0 & +0.0 & +0.0 & +0.0 & +0.0 & +0.0 & +0.0 & +0.0 & +0.0 & +0.0 & +0.0 & +0.0 \\
    Partitioning & -0.9 & -0.6 & -0.2 & +0.0 & -0.2 & -0.2 & -0.2 & -0.2 & -0.9 & -0.6 & -0.3 & -0.1 & -0.9 & -0.5 & -0.3 & -0.1 \\
    \secdmp{} & -0.9 & -0.6 & -0.2 & -0.1 & -0.2 & -0.2 & -0.2 & -0.2 & -0.9 & -0.6 & -0.3 & -0.1 & -0.9 & -0.5 & -0.3 & -0.1 \\
    \prospect{} & -0.9 & -0.6 & -0.2 & +0.0 & +0.0 & +0.0 & -0.1 & -0.1 & -0.8 & -0.5 & -0.2 & -0.1 & -0.8 & -0.5 & -0.1 & +0.1 \\
    \prospect{}+\secdmp{} & -0.9 & -0.6 & -0.2 & +0.0 & +0.0 & +0.0 & -0.1 & -0.1 & -0.8 & -0.5 & -0.2 & -0.1 & -0.8 & -0.5 & -0.1 & +0.1 \\
    \midrule
    \multicolumn{17}{l}{\emph{HMAC}} \\
    Permanent Disabling & +0.0 & +0.0 & +0.0 & +0.0 & +0.0 & +0.0 & +0.0 & +0.0 & +0.0 & +0.0 & +0.0 & +0.0 & +5.3 & +10.6 & +15.9 & +19.1 \\
    Selective Disabling & +0.0 & +0.0 & +0.0 & +0.0 & +0.0 & +0.0 & +0.0 & +0.0 & +0.0 & +0.0 & +0.0 & +0.0 & +0.0 & +0.0 & +0.0 & +0.0 \\
    Partitioning & +0.5 & +0.4 & +0.3 & +0.2 & -0.4 & -0.1 & +0.1 & +0.2 & +0.0 & +0.0 & +0.0 & +0.0 & -0.3 & +0.1 & +0.5 & +0.7 \\
    \secdmp{} & +0.5 & +0.4 & +0.3 & +0.2 & -0.4 & -0.1 & +0.1 & +0.2 & +0.0 & +0.0 & +0.0 & +0.0 & -0.3 & +0.1 & +0.5 & +0.7 \\
    \prospect{} & +0.4 & +0.2 & +0.1 & +0.0 & -0.3 & -0.1 & +0.2 & +0.3 & -0.1 & -0.1 & +0.0 & +0.0 & -0.2 & +0.2 & +0.5 & +0.8 \\
    \prospect{}+\secdmp{} & +0.4 & +0.2 & +0.1 & +0.0 & -0.3 & -0.1 & +0.2 & +0.3 & -0.1 & -0.1 & +0.0 & +0.0 & -0.2 & +0.2 & +0.5 & +0.8 \\
    \midrule
    \multicolumn{17}{l}{\emph{HKDF}} \\
    Permanent Disabling & +0.0 & +0.0 & +0.0 & +0.1 & +0.0 & +0.0 & +0.0 & +0.0 & +0.0 & +0.0 & +0.0 & +0.0 & +5.3 & +10.4 & +15.7 & +18.9 \\
    Selective Disabling & +0.0 & +0.0 & +0.0 & +0.0 & +0.0 & +0.0 & +0.0 & +0.0 & +0.0 & +0.0 & +0.0 & +0.0 & +0.0 & +0.0 & +0.0 & +0.0 \\
    Partitioning & -0.1 & +0.1 & +0.3 & +0.5 & -0.8 & -0.6 & -0.3 & -0.2 & -0.1 & +0.0 & +0.0 & +0.0 & -1.0 & -0.7 & -0.4 & -0.3 \\
    \secdmp{} & -0.1 & +0.1 & +0.3 & +0.4 & -0.8 & -0.6 & -0.3 & -0.2 & -0.1 & +0.0 & +0.0 & +0.0 & -1.0 & -0.7 & -0.4 & -0.3 \\
    \prospect{} & -0.3 & -0.2 & -0.1 & +0.1 & -0.8 & -0.5 & -0.2 & -0.1 & -0.2 & -0.1 & +0.0 & +0.0 & -1.3 & -0.9 & -0.5 & -0.3 \\
    \prospect{}+\secdmp{} & -0.3 & -0.2 & -0.1 & +0.1 & -0.8 & -0.5 & -0.2 & -0.1 & -0.2 & -0.1 & +0.0 & +0.0 & -1.3 & -0.9 & -0.5 & -0.3 \\
    \bottomrule
  \end{tabular}
\end{table*}